\documentclass{article}
\usepackage{bm}

\usepackage{color}
\usepackage{amsmath}
\usepackage{amssymb}\usepackage{amsthm}
\usepackage{amsfonts}
\usepackage{amscd}
\begin{document}
\theoremstyle{plain}
\newtheorem*{ithm}{Theorem}
\newtheorem*{idefn}{Definition}
\newtheorem{thm}{Theorem}[section]
\newtheorem{lem}[thm]{Lemma}
\newtheorem{dlem}[thm]{Lemma/Definition}
\newtheorem{prop}[thm]{Proposition}
\newtheorem{set}[thm]{Setting}
\newtheorem{cor}[thm]{Corollary}
\newtheorem*{icor}{Corollary}
\theoremstyle{definition}
\newtheorem{assum}[thm]{Assumption}
\newtheorem{notation}[thm]{Notation}
\newtheorem{defn}[thm]{Definition}
\newtheorem{clm}[thm]{Claim}
\newtheorem{ex}[thm]{Example}
\theoremstyle{remark}
\newtheorem{rem}[thm]{Remark}
\newcommand{\unit}{\mathbb I}
\newcommand{\ali}[1]{{\mathfrak A}_{[ #1 ,\infty)}}
\newcommand{\alm}[1]{{\mathfrak A}_{(-\infty, #1 ]}}
\newcommand{\nn}[1]{\lV #1 \rV}
\newcommand{\br}{{\mathbb R}}
\newcommand{\dm}{{\rm dom}\mu}
\newcommand{\lb}{l_{\bb}(n,n_0,k_R,k_L,\lal,\bbD,\bbG,Y)}
\newcommand{\Ad}{\mathop{\mathrm{Ad}}\nolimits}
\newcommand{\Proj}{\mathop{\mathrm{Proj}}\nolimits}
\newcommand{\RRe}{\mathop{\mathrm{Re}}\nolimits}
\newcommand{\RIm}{\mathop{\mathrm{Im}}\nolimits}
\newcommand{\Wo}{\mathop{\mathrm{Wo}}\nolimits}
\newcommand{\Prim}{\mathop{\mathrm{Prim}_1}\nolimits}
\newcommand{\Primz}{\mathop{\mathrm{Prim}}\nolimits}
\newcommand{\ClassA}{\mathop{\mathrm{ClassA}}\nolimits}
\newcommand{\Class}{\mathop{\mathrm{Class}}\nolimits}
\newcommand{\diam}{\mathop{\mathrm{diam}}\nolimits}
\def\qed{{\unskip\nobreak\hfil\penalty50
\hskip2em\hbox{}\nobreak\hfil$\square$
\parfillskip=0pt \finalhyphendemerits=0\par}\medskip}
\def\proof{\trivlist \item[\hskip \labelsep{\bf Proof.\ }]}
\def\endproof{\null\hfill\qed\endtrivlist\noindent}
\def\proofof[#1]{\trivlist \item[\hskip \labelsep{\bf Proof of #1.\ }]}
\def\endproofof{\null\hfill\qed\endtrivlist\noindent}

\newcommand{\varphii}{\varphi}
\newcommand{\pgs}{\caP_{\sigma}}
\newcommand{\oo}{{\boldsymbol\varphii}}
\newcommand{\caA}{{\mathcal A}}
\newcommand{\caB}{{\mathcal B}}
\newcommand{\caC}{{\mathcal C}}
\newcommand{\caD}{{\mathcal D}}
\newcommand{\caE}{{\mathcal E}}
\newcommand{\caF}{{\mathcal F}}
\newcommand{\caG}{{\mathcal G}}
\newcommand{\caH}{{\mathcal H}}
\newcommand{\caI}{{\mathcal I}}
\newcommand{\caJ}{{\mathcal J}}
\newcommand{\caK}{{\mathcal K}}
\newcommand{\caL}{{\mathcal L}}
\newcommand{\caM}{{\mathcal M}}
\newcommand{\caN}{{\mathcal N}}
\newcommand{\caO}{{\mathcal O}}
\newcommand{\caP}{{\mathcal P}}
\newcommand{\caQ}{{\mathcal Q}}
\newcommand{\caR}{{\mathcal R}}
\newcommand{\caS}{{\mathcal S}}
\newcommand{\caT}{{\mathcal T}}
\newcommand{\caU}{{\mathcal U}}
\newcommand{\caV}{{\mathcal V}}
\newcommand{\caW}{{\mathcal W}}
\newcommand{\caX}{{\mathcal X}}
\newcommand{\caY}{{\mathcal Y}}
\newcommand{\caZ}{{\mathcal Z}}
\newcommand{\bbA}{{\mathbb A}}
\newcommand{\bbB}{{\mathbb B}}
\newcommand{\bbC}{{\mathbb C}}
\newcommand{\bbD}{{\mathbb D}}
\newcommand{\bbE}{{\mathbb E}}
\newcommand{\bbF}{{\mathbb F}}
\newcommand{\bbG}{{\mathbb G}}
\newcommand{\bbH}{{\mathbb H}}
\newcommand{\bbI}{{\mathbb I}}
\newcommand{\bbJ}{{\mathbb J}}
\newcommand{\bbK}{{\mathbb K}}
\newcommand{\bbL}{{\mathbb L}}
\newcommand{\bbM}{{\mathbb M}}
\newcommand{\bbN}{{\mathbb N}}
\newcommand{\bbO}{{\mathbb O}}
\newcommand{\bbP}{{\mathbb P}}
\newcommand{\bbQ}{{\mathbb Q}}
\newcommand{\bbR}{{\mathbb R}}
\newcommand{\bbS}{{\mathbb S}}
\newcommand{\bbT}{{\mathbb T}}
\newcommand{\bbU}{{\mathbb U}}
\newcommand{\bbV}{{\mathbb V}}
\newcommand{\bbW}{{\mathbb W}}
\newcommand{\bbX}{{\mathbb X}}
\newcommand{\bbY}{{\mathbb Y}}
\newcommand{\bbZ}{{\mathbb Z}}
\newcommand{\str}{^*}
\newcommand{\lv}{\left \vert}
\newcommand{\rv}{\right \vert}
\newcommand{\lV}{\left \Vert}
\newcommand{\rV}{\right \Vert}
\newcommand{\la}{\left \langle}
\newcommand{\ra}{\right \rangle}
\newcommand{\ltm}{\left \{}
\newcommand{\rtm}{\right \}}
\newcommand{\lcm}{\left [}
\newcommand{\rcm}{\right ]}
\newcommand{\ket}[1]{\lv #1 \ra}
\newcommand{\bra}[1]{\la #1 \rv}
\newcommand{\kl}[2]{\bm {#1}_{\hat \Lambda_{#2}}}
\newcommand{\hln}[1]{\hat\Lambda_{#1}}
\newcommand{\lmk}{\left (}
\newcommand{\rmk}{\right )}
\newcommand{\al}{{\mathcal A}}
\newcommand{\md}{M_d({\mathbb C})}
\newcommand{\ainn}{\mathop{\mathrm{AInn}}\nolimits}
\newcommand{\id}{\mathop{\mathrm{id}}\nolimits}
\newcommand{\Tr}{\mathop{\mathrm{Tr}}\nolimits}
\newcommand{\Ran}{\mathop{\mathrm{Ran}}\nolimits}
\newcommand{\Ker}{\mathop{\mathrm{Ker}}\nolimits}
\newcommand{\Aut}{\mathop{\mathrm{Aut}}\nolimits}
\newcommand{\spn}{\mathop{\mathrm{span}}\nolimits}
\newcommand{\Mat}{\mathop{\mathrm{M}}\nolimits}
\newcommand{\UT}{\mathop{\mathrm{UT}}\nolimits}
\newcommand{\DT}{\mathop{\mathrm{DT}}\nolimits}
\newcommand{\GL}{\mathop{\mathrm{GL}}\nolimits}
\newcommand{\spa}{\mathop{\mathrm{span}}\nolimits}
\newcommand{\supp}{\mathop{\mathrm{supp}}\nolimits}
\newcommand{\rank}{\mathop{\mathrm{rank}}\nolimits}
\newcommand{\idd}{\mathop{\mathrm{id}}\nolimits}
\newcommand{\ran}{\mathop{\mathrm{Ran}}\nolimits}
\newcommand{\dr}{ \mathop{\mathrm{d}_{{\mathbb R}^k}}\nolimits} 
\newcommand{\dc}{ \mathop{\mathrm{d}_{\cc}}\nolimits} \newcommand{\drr}{ \mathop{\mathrm{d}_{\rr}}\nolimits} 
\newcommand{\zin}{\mathbb{Z}}
\newcommand{\rr}{\mathbb{R}}
\newcommand{\cc}{\mathbb{C}}
\newcommand{\ww}{\mathbb{W}}
\newcommand{\nan}{\mathbb{N}}\newcommand{\bb}{\mathbb{B}}
\newcommand{\aaa}{\mathbb{A}}\newcommand{\ee}{\mathbb{E}}
\newcommand{\pp}{\mathbb{P}}
\newcommand{\wks}{\mathop{\mathrm{wk^*-}}\nolimits}
\newcommand{\mk}{{\Mat_k}}
\newcommand{\mnz}{\Mat_{n_0}}
\newcommand{\mn}{\Mat_{n}}
\newcommand{\dist}{\dc}
\newcommand{\braket}[2]{\left\langle#1,#2\right\rangle}
\newcommand{\ketbra}[2]{\left\vert #1\right \rangle \left\langle #2\right\vert}
\newcommand{\abs}[1]{\left\vert#1\right\vert}
\newtheorem{nota}{Notation}[section]
\def\qed{{\unskip\nobreak\hfil\penalty50
\hskip2em\hbox{}\nobreak\hfil$\square$
\parfillskip=0pt \finalhyphendemerits=0\par}\medskip}
\def\proof{\trivlist \item[\hskip \labelsep{\bf Proof.\ }]}
\def\endproof{\null\hfill\qed\endtrivlist\noindent}
\def\proofof[#1]{\trivlist \item[\hskip \labelsep{\bf Proof of #1.\ }]}
\def\endproofof{\null\hfill\qed\endtrivlist\noindent}
\newcommand{\Ln}[1]{\Lambda_{#1}}
\newcommand{\ZZ}{\bbZ_2\times\bbZ_2}
\newcommand{\SSS}{\mathcal{S}}
\newcommand{\cs}{S}
\newcommand{\ct}{t}
\newcommand{\hS}{S}
\newcommand{\vv}{{\boldsymbol v}}
\newcommand{\ala}{a}
\newcommand{\bet}{b}
\newcommand{\gam}{c}
\newcommand{\alphas}{\alpha}
\newcommand{\alphai}{\alpha^{(\sigma_{1})}}
\newcommand{\alphan}{\alpha^{(\sigma_{2})}}
\newcommand{\betas}{\beta}
\newcommand{\betai}{\beta^{(\sigma_{1})}}
\newcommand{\betan}{\beta^{(\sigma_{2})}}
\newcommand{\alphass}{\alpha^{{(\sigma)}}}
\newcommand{\uu}{V}
\newcommand{\mopbk}{{\mathop{\mathrm{bk}}}}
\newcommand{\mopbd}{{\mathop{\mathrm{bd}}}}
\newcommand{\vp}{\varsigma}
\newcommand{\vpr}{R}
\newcommand{\tg}{\tau_{\Gamma}}
\newcommand{\sgg}{\Sigma_{\Gamma}^{(\sigma)}}
\newcommand{\nh}{t28}
\newcommand{\rk}{6}
\newcommand{\nii}{2}
\newcommand{\nhh}{28}
\newcommand{\sjt}{30}
\newcommand{\sjtg}{30}
\newcommand{\bcg}{\caB(\caH_{\alpha})\otimes  C^{*}(\Sigma)}
\newcommand{\hu}{\mathop {\mathrm H_{U}}}
\newcommand{\hd}{\mathop {\mathrm H_{D}}}
\newcommand{\Cbk}{\caC_{\mopbk}}
\newcommand{\CUbk}{\caC_{\mopbk}^U}
\newcommand{\CDbk}{\caC_{\mopbk}^D}
\newcommand{\Cr}{\caC^r}
\newcommand{\Cl}{\caC^l}
\newcommand{\Obk}{O_{\mathop{\mathrm{bk}}}}
\newcommand{\OUbk}{O_{\mathop{\mathrm{bk}}}^U}
\newcommand{\Orbd}{O^r_{\mathop{\mathrm{bd}}}}
\newcommand{\OrUbd}{O^{rU}_{\mathop{\mathrm{bd}}}}
\newcommand{\Obkl}{O_{\mathop{\mathrm{bk},\Lambda_0}}}
\newcommand{\OUbkl}{O_{\mathop{\mathrm{bk},\Lambda_0}}^U}
\newcommand{\Orbdl}{O^r_{\mathop{\mathrm{bd},\Lambda_{r0}}}}
\newcommand{\OrUbdl}{O^{rU}_{\mathop{\mathrm{bd},\Lambda_{r0}}}}
\newcommand{\Obu}{O_{\mathop{\mathrm{bd}}}^{\mathop{\mathrm{bu}}}}
\newcommand{\Obul}{O_{\mathop{\mathrm{bd}},\lz}^{\mathop{\mathrm{bu}}}}
\newcommand{\Odl}{O_{\lzr}^{\caD}}
\newcommand{\fd}{F^{\caD}}
\newcommand{\hilb}{\mathop{\mathrm {Hilb}}_f}
\newcommand{\Obun}[1]{O_{\mathop{\mathrm{bd}#1}}^{\mathop{\mathrm{bu}}}}
\newcommand{\OrUbdn}[1]{O^{rU}_{\mathop{\mathrm{bd}#1 }}}
\newcommand{\Vbk}[2]{\caV_{#1#2}}
\newcommand{\Vbd}[2]{\caV_{#1#2}^{\mopbd}}
\newcommand{\VUbd}[2]{\caV_{#1#2}^{\mopbd U} }
\newcommand{\Vbu}[4]{\caV_{#1,\lmk #2, #3\rmk}^{{#4}}}
\newcommand{\lo}{\Lambda_1}
\newcommand{\lt}{\Lambda_2}
\newcommand{\loz}{\Lambda_1^{(0)}}
\newcommand{\ltz}{\Lambda_2^{(0)}}
\newcommand{\Obj}{{\mathop{\mathrm{Obj}}}}
\newcommand{\Mor}{{\mathop{\mathrm{Mor}}}}
\newcommand{\Irr}{{\mathop{\mathrm{Irr}}}}
\newcommand{\lz}{\Lambda^{(0)}}
\newcommand{\lhz}{\Lambda_l^{(0)}}
\newcommand{\lzr}{\Lambda_{r0}}
\newcommand{\lm}[1]{{\Lambda_{#1}}}
\newcommand{\llz}{(\lz,\lzr)}
\newcommand{\lmr}[1]{{\Lambda_{r#1}}}
\newcommand{\hlm}[1]{{\hat\Lambda_{#1}}}
\newcommand{\tlm}[1]{{\tilde\Lambda_{#1}}}
\newcommand{\ld}{\Lambda}
\newcommand{\bl}{\caB_l}
\newcommand{\brr}{\caB_r}
\newcommand{\fbk}{\caF_{\mopbk}^U}
\newcommand{\fbd}{\caF}
\newcommand{\gu}{\caG^U}
\newcommand{\hb}[2]{\iota^{(\lz,\lzr)}\lmk#1: #2\rmk}
\newcommand{\hfc}{\hat F^{\llz}}
\newcommand{\ffc}{F^{\llz}}
\newcommand{\pdc}[1]{\pi\lmk\caA_{#1}\rmk''}
\newcommand{\Tbk}[3]{T_{#1}^{(\frac{3\pi}2,\frac\pi 2), {#2},{#3}}}
\newcommand{\Vrl}[4]{V_{#1,\lmk #2,#3\rmk}^{(#4)}}
\newcommand{\Tbkv}[2]{T_{#1}^{(\frac{3\pi}2,\frac\pi 2), #2, V_{#1,#2}}}
\newcommand{\Tbdv}[2]{T_{#1}^{\mathrm{(l)} #2\Vrl{{#1}}{#2}}}
\newcommand{\Srl}[4]{S_{#1}^{(l), V_{{#1},{(#2,#3)}}^{(#4)}}}
\newcommand{\go}{\Gamma_1}
\newcommand{\gt}{\Gamma_2}
\newcommand{\CCZ}{\mathop{\mathrm{CCZ}}\nolimits}
\newcommand{\CZ}{\mathop{\mathrm{CZ}}\nolimits}
\title{Anyonic symmetry fractionalization in SET phases
}
\author{Jos\'e Garre Rubio and Yoshiko Ogata}

\maketitle
\abstract{ We consider the anyonic spin systems with a
global symmetry, the so-called symmetry enriched topological (SET) phases.
We introduce the phase characterizing  the symmetry
fractionalization of the anyons. Our assumptions on how the global symmetry
acts prevents anyon permutation effects.}

\section{Introduction}
Two-dimensional topological orders \cite{Wentopo} offer one of the most
important applications of quantum computing: topological error correction
\cite{TC,Topomemory}. That potential is based on the existence of anyons,
quasiparticles that are neither bosons nor fermions \cite{Wilczek} whose
existence has been recently realized experimentally \cite{E2,E3}.

Besides, topological order can stand alone without a global symmetry, like
all the models related to the pioneering work of Kitaev toric code
\cite{TC} and its generalizations \cite{Stringnets}, in nature both aspects
are linked. This is the case of the remarkable fractional quantum Hall
effect (fQHE) \cite{TSG,Laughlin} that hosts anyons that fractionalize the
elementary electron's charge. A similar effect occurs in spin liquids where
there is a spin-charge separation \cite{SLH,SLA}. This gives rise to the
study of symmetry fractionalization patterns (like in charge conservation
symmetry or internal symmetries) in topologically ordered systems
which correspond to symmetry enriched topological (SET) phases.

SET phases have been intensively studied and classified through exactly
solvable Hamiltonians \cite{C1,C2,C3,C4,C5,C6,C7} and remarkably a complete
categorical classification was achieved in \cite{2dSET} which includes
anyon permutation effects.

Besides that, a rigorous classification on non-solvable models valid for
spin lattices on the thermodynamic limit is missing. In this paper we
tackle that problem by using the operator algebraic approach that allows us
to state a definition of equivalence relation for SET phases according to a
well-behaved path of gapped Hamiltonians (technically quasi-local
automorphisms).
Our work stands on the contributions made by one of the authors \cite{MTC}
where the right anyonic property which is invariant under quasi-local
automorphisms has been identified (the so-called approximate Haag duality).
We prove under reasonable assumptions the SET classification of symmetry
fractionalization patterns (no anyon permutation is considered).

While this paper was in preparation, there appeared a paper \cite{KVW} in arXiv whose results overlap 
with this article.

\subsection{Quantum spin systems}\label{sec:qss}
Throughout this paper, we fix some $2\le d\in\nan$.
We denote the algebra of $d\times d$ matrices by $\Mat_{d}$.
For each $z\in\bbZ^2$,  let $\caA_{\{z\}}$ be an isomorphic copy of $\Mat_{d}$, and for any finite subset $\Lambda\subset\bbZ^2$, we set $\caA_{\Lambda} = \bigotimes_{z\in\Lambda}\caA_{\{z\}}$.
For finite $\Lambda$, the algebra $\caA_{\Lambda} $ can be regarded as the set of all bounded operators acting on
the Hilbert space $\bigotimes_{z\in\Lambda}{\bbC}^{d}$.
We use this identification freely.
If $\Lambda_1\subset\Lambda_2$, the algebra $\caA_{\Lambda_1}$ is naturally embedded in $\caA_{\Lambda_2}$ by tensoring its elements with the identity. 
For an infinite subset $\Gamma\subset \bbZ^{2}$,
$\caA_{\Gamma}$
is given as the inductive limit of the algebras $\caA_{\Lambda}$ with $\Lambda$, finite subsets of $\Gamma$.
We call $\caA_{\Gamma}$ the quantum spin system on $\Gamma$.
For a subset $\Gamma_1$ of $\Gamma\subset\bbZ^{2}$,
the algebra $\caA_{\Gamma_1}$ can be regarded as a subalgebra of $\caA_{\Gamma}$. 
For $\Gamma\subset \bbR^2$, with a bit abuse of notation, we write $\caA_{\Gamma}$
to denote $\caA_{\Gamma\cap \bbZ^2}$.
Also, $\Gamma^c$ denotes the complement of $\Gamma$ in $\bbR^2$.
The algebra $\caA:=\caA_{\bbZ^2}$ is the two dimensional quantum spin system we consider.
We also set $\caA_{\rm loc}:=\bigcup_{\Lambda\Subset \bbZ^2}\caA_{\Lambda}$.
For a region $X\subset \bbR^2$, $\partial X$ denotes the boundary of 
$X$. For a region $X\subset \bbR^2$ and $l\in\bbN$,
$X^{(l)}$ denotes the set of points with distance less than or equal to $l$
from $X$.
Throughout the paper, we consider a fixed pure state $\omega$ on $\caA$ with a GNS triple $(\caH,\pi,\Omega)$.

Let $G$ be a finite group and  $U$ a unitary representation of $G$ on $\bbC^{d}$.
We assume that the group action
\begin{align}\label{ffl}
G\ni g \mapsto \Ad U_g\in \Aut\Mat_d
\end{align}
is faithful.
Let $\Gamma\subset \bbZ^{2}$ be a non-empty subset.
For each $g\in G$, there exists a unique automorphism $\beta^{\Gamma}$ on $\caA_{\Gamma}$
such that 
\begin{align}\label{tgg}
\beta_{g}^{\Gamma}\lmk A\rmk=\Ad\lmk\bigotimes_{I} U(g)\rmk\lmk A\rmk,\quad A\in\caA_{I},\quad g\in G,
\end{align}
for any finite subset $I$ of $\Gamma$.
We call the group homomorphism $\beta^{\Gamma}: G\to \Aut \caA_{\Gamma}$, 
the on-site action of $G$ 
on $\caA_{\Gamma}$ given by $U$.
For simplicity, we denote $\beta^{\bbZ^{2}}_{g}$ by $\beta_{g}$.

We assume that our state $\omega$ is invariant under the group action $\beta$.
: For all $g\in G$, $\omega\beta_g=\omega$.
From this assumption, there exists a unitary representation $R: G\to \caU(\caH)$
such that
\begin{align}
\Ad R_g\pi (A)=\pi\beta_g(A),\quad A\in \caA,\quad g\in G.
\end{align}

\subsection{Cones}
In this paper, we have to consider various kinds of cones.
In this subsection we collect notations related to cones.
For $\bm a\in \bbR^2$, $\theta\in\bbR$ and $\varphi\in (0,\pi)$,
set 
\begin{align*}
\Lambda_{\bm a, \theta,\varphi}
:=&\left\{
\bm x\in\bbR^{2}\mid (\bm x-\bm a)\cdot \bm e_{\theta}>\cos\varphi\cdot \lV \bm x-{\bm a}\rV
\right\}\\
=&\bm a+\left\{
t\bm e_{\beta}\mid t>0,\quad \beta\in (\theta-\varphi,\theta+\varphi)\right\}.
\end{align*}
Here, we set $\bm e_\theta:=(\cos\theta,\sin\theta)$ for $\theta\in \bbR$.
We call a subset of $\bbR^{2}$ with this form a {\it cone}, and denote by $\Cbk$ the set of all cones.
For a cone $\Lambda=\Lambda_{\bm a, \theta,\varphi}$ given above,
we set
\begin{align*}
\arg\lmk \Lambda\rmk:=
\left\{ e^{it}\in\bbT \mid t\in [\theta-\varphi,\theta+\varphi]\right\},\quad
|\arg\lmk \Lambda\rmk|:=2\varphi,\quad \text{and}\quad
\bm a_{\Lambda}:=\bm a,\quad
\bm e_{\Lambda}:=\bm e_{\theta}.
\end{align*}
For $\varepsilon>0$ and $\Lambda=\Lambda_{\bm a, \theta,\varphi}$ with $\varphi+\varepsilon<\pi$
we denote the ``fattened''  and ``thinned'' cone by
\begin{align*}
\Lambda_{\varepsilon}:=\Lambda_{\bm a, \theta,\varphi+\varepsilon},\quad
\Lambda_{-\varepsilon}:=\Lambda_{\bm a, \theta,\varphi-\varepsilon}
\end{align*}
Furthermore, for $\varepsilon>0$ and $\Lambda=\Lambda_{\bm a, \theta,\varphi}$ with $\varepsilon< \varphi, \pi-\varphi$, we set ``fattened'' edge of 
$\ld$:
\begin{align}
\begin{split}
\lmk \partial \ld\rmk_\varepsilon:=
\Lambda_{\bm a, \theta+\varphi,\epsilon }\cup
\Lambda_{\bm a, \theta-\varphi,\epsilon }.
\end{split}
\end{align}

For each $\theta\in \bbR$ and $\varphi\in (0, \pi)$, set
\begin{align}\label{ctvdef}
\caC_{(\theta,\varphi)}:=\left\{
\Lambda\in \Cbk\mid \arg\Lambda\cap\bbA_{[\theta-\varphi,\theta+\varphi]}=\emptyset
\right\}.
\end{align}
Here for $I\subset \bbR$, we set
\[
\bbA_I:=\left\{e^{it}\mid t\in I\right\}\subset \bbT.
\]
We consider the following sets of cones:
\begin{align}
\begin{split}
&\CUbk:=\caC_{(\frac{3\pi}2, \frac \pi 2)},\quad
\CDbk:=\caC_{(\frac\pi 2, \frac\pi 2)},\\
&\Cr:=\left\{
\Lambda_{(a,0),0,\varphi}\mid a\in\bbR,\; 0<\varphi<\pi 
\right\},\\
&\Cl:=\left\{
\lmk \Lambda_r\rmk^c\mid \Lambda_r\in\Cr
\right\}
=\left\{
\overline{\Lambda_{(a,0),\pi, \varphi}}
\mid
a\in\bbR,;0<\varphi<\pi
\right\}.
\end{split}
\end{align}
We also introduce the following
\begin{align}
PC:=\left\{ 
(\lo,\lt)\in \caC_l\times \caC_r\middle |\begin{gathered}\lo\cap\lt =\emptyset,\\
\lmk \arg\Lambda_1\rmk_\varepsilon \cap\arg\Lambda_2=\emptyset\quad
\text{for some}\quad \varepsilon>0
\end{gathered}
\right\}.
\end{align}
For $\{(\lm {l\alpha}, \lm {r\alpha})\}_{\alpha}, \{(\lm {l\beta}', \lm {r\beta}')\}_{\beta}\subset PC $,
we write
$\{(\lm {l\alpha}, \lm {r\alpha})\}_{\alpha}\leftarrow \{(\lm {l\beta}', \lm {r\beta}')\}_{\beta}$
if there exists $\varepsilon>0$ such that
\begin{align}\label{kaeru}
\begin{split}
\cup_{\beta}\arg\lmk \lmk \lm {l\beta}'\rmk^c \rmk_\varepsilon
\subset \cap_\alpha \arg \lm {r\alpha}.
\end{split}
\end{align}
Note that (\ref{kaeru}) holds if and only if 
\begin{align}\label{kaeru2}
\begin{split}
\cup_\alpha \arg\lmk  \lmk \lm {r\alpha}\rmk^c\rmk_\varepsilon\subset 
\cap_{\beta}\arg\lmk \lm {l\beta}' \rmk
\end{split}
\end{align}
\subsection{Assumptions on local von Neumann algebras for cones}
Consider the setting in subsection \ref{sec:qss}.
Corresponding to various sets of cones we consider the following $C^*$-algebras.
\begin{align}
\begin{split}
&\caF:=\overline{\cup_{(\lo,\lt)\in PC}\pi\lmk\caA_{\lmk \lo \cup \lt \rmk^c}\rmk''
}^n,\\
& \caB_l:=\overline{{
\cup_{\Lambda_l\in \Cl}\; \pi\lmk\caA_{\Lambda_l}\rmk''}}^n,\\
&\caB_r:=\overline{
\cup_{\Lambda_r\in \Cr}\; \pi\lmk\caA_{\Lambda_r}\rmk''}^n,\\
&\caB_{(\theta,\varphi)}:=
\overline{\cup_{\Lambda\in\caC_{(\theta,\varphi)} }\pi(\caA_{\Lambda})'' }^n.
\end{split}
\end{align}
where ${\overline{\cdot}}^n$ indicates the norm closure.
Note that they are all $C^*$-algebras because of the upward filtering property of
$PC$,
$\Cl$, $\Cr$, $\caC_{(\theta,\varphi)} $ with respect to
(anti) inclusions.
We denote by $\caU(\caF)$ (resp. $\caU(\caB_l)$, $\caU(\caB_r)$, $\caU(\caA)$)
 the set of all unitaries in $\caF$ (resp. $\caB_l$, $\caB_r$, $\caA$).

We denote by $\caU(\caH)$ (resp. $\caB(\caH)$) the set of all unitaries 
(resp. bounded opearators)on the Hilbert space $\caH$.
We assume the Approximate Haag duality.
\begin{assum}\label{a1}[Approximate Haag duality]
For any $\varphi\in (0,2\pi)$ and 
 $\varepsilon>0$ with
$\varphi+4\varepsilon<2\pi$,
there is some $R_{\varphi,\varepsilon}>0$ and decreasing
functions $f_{\varphi,\varepsilon,\delta}(t)$, $\delta>0$
on $\bbR_{\ge 0}$
with $\lim_{t\to\infty}f_{\varphi,\varepsilon,\delta}(t)=0$
such that
\begin{description}
\item[(i)]
for any cone $\Lambda$ with $|\arg\Lambda|=\varphi$, there is a unitary 
$U_{\Lambda,\varepsilon}\in \caU(\caH)$
satisfying
\begin{align}\label{lem7p}
\pi\lmk\caA_{\Lambda^c}\rmk'\subset 
\Ad\lmk U_{\Lambda,\varepsilon}\rmk\lmk 
\pi\lmk \caA_{\lmk \Lambda-R_{\varphi,\varepsilon}\bm e_\Lambda\rmk_\varepsilon}\rmk''
\rmk,
\end{align}
and 
\item[(ii)]
 for any $\delta>0$ and $t\ge 0$, there is a unitary 
 $\tilde U_{\Lambda,\varepsilon,\delta,t}\in \pi\lmk \caA_{\Lambda_{\varepsilon+\delta}-t\bm e_{\Lambda}}\rmk''$
 satisfying
\begin{align}\label{uappro}
\lV
U_{\Lambda,\varepsilon}-\tilde U_{\Lambda,\varepsilon,\delta,t}
\rV\le f_{\varphi,\varepsilon,\delta}(t).
\end{align}
\end{description}
\end{assum}
For right cones $\Cr$, we require a slightly stronger condition of approximate Haag duality.
\begin{assum}\label{a1r}
For any $\varphi\in (0,\frac\pi 2)$, $\varepsilon>0$ with $2\varphi+2\varepsilon<\pi$,
there is some $R^{(r)}_{\varphi,\varepsilon}>0$ and decreasing
functions $f^{(r)}_{\varphi,\varepsilon,\delta}(t)$, $\delta>0$
on $\bbR_{\ge 0}$
with $\lim_{t\to\infty}f^{(r)}_{\varphi,\varepsilon,\delta}(t)=0$
such that
\begin{description}
\item[(i)]
for any cone $\Lambda_r=\ld_{(a,0) 0, \varphi}\in \Cr$, there is a unitary 
$U^{(r)}_{\Lambda_r,\varepsilon}\in \caU(\caH)$
satisfying
\begin{align}\label{lem7p}
\pi\lmk\caA_{\ld_r^c}\rmk'\subset 
\Ad\lmk U^{(r)}_{\Lambda_r,\varepsilon}\rmk
\lmk 
\pi\lmk \caA_{ \ld_{{(a-R^{(r)}, 0), 0, \varphi+\varepsilon}}}\rmk''
\rmk,
\end{align}
and 
\item[(ii)]
for any $\delta>0$ and $t\ge 0$, there is a unitary 
$\tilde U^{(r)}_{\Lambda_r,\varepsilon,\delta,t}\in \pi\lmk 
 \caA_{
 \ld_{{(a-t,0)}, 0, {\varphi+\varepsilon+\delta}}
 }\rmk''\cap \caF$
satisfying
\begin{align}\label{uappro}
\lV
U^{(r)}_{\Lambda_r,\varepsilon}-\tilde U^{(r)}_{\Lambda_r,\varepsilon,\delta,t}
\rV\le f_{\varphi,\varepsilon,\delta}(t).
\end{align}
\end{description}
Note in particular, we have $U^{(r)}_{\Lambda_r,\varepsilon}\in \caF$.
\end{assum}
The condition is stronger because we require $U^{(r)}_{\Lambda_r,\varepsilon}$ to be in $\caF$.
Here is the left version.
\begin{assum}\label{a1l}
For any $\varphi\in (0,\frac\pi 2)$, $\varepsilon>0$ with $2\varphi+2\varepsilon<\pi$,
there is some $R^{(l)}_{\varphi,\varepsilon}>0$ and decreasing
functions $f^{(l)}_{\varphi,\varepsilon,\delta}(t)$, $\delta>0$
on $\bbR_{\ge 0}$
with $\lim_{t\to\infty}f^{(l)}_{\varphi,\varepsilon,\delta}(t)=0$
such that
\begin{description}
\item[(i)]
for any cone $\Lambda_l=\ld_{ (a,0), \pi,  \varphi}\in \Cl$, there is a unitary 
$U^{(l)}_{\Lambda_l,\varepsilon}\in \caU(\caH)$
satisfying
\begin{align}\label{lem7p}
\pi\lmk\caA_{(\lm l)^c}\rmk'\subset 
\Ad\lmk U^{(l)}_{\Lambda_l,\varepsilon}\rmk\lmk
\pi\lmk \caA_{ \ld_{\lmk a+R^{(l)}_{\varphi,\varepsilon}, 0\rmk}, \pi, {(\varphi+\varepsilon)}}\rmk''
\rmk,
\end{align}
and 
\item[(ii)]
 for any $\delta>0$ and $t\ge 0$, there is a unitary 
$\tilde U^{(l)}_{\Lambda_l,\varepsilon,\delta,t}\in \pi\lmk 
\caA_{
  \ld_{{(a+t,0)}, \pi {(\varphi+\varepsilon+\delta)}}
 }\rmk''\cap \caF$
 satisfying
\begin{align}\label{uappro}
\lV
U^{(l)}_{\Lambda_l,\varepsilon}-\tilde U^{(l)}_{\Lambda_l,\varepsilon,\delta,t}
\rV
\le f^{(l)}_{\varphi,\varepsilon,\delta}(t).
\end{align}
\end{description}
Note in particular, we have $U^{(l)}_{\Lambda_l,\varepsilon}\in \caF$.
\end{assum}

Furthermore, we assume the following.
\begin{assum}\label{a2}
For any cone $\Lambda$, $\pi\lmk\caA_{\Lambda}\rmk''$ is properly infinite.
\end{assum}
This condition is satisfied automatically if $\omega$ is a gapped ground state.
\subsection{Bulk braided $C^*$-tensor category}
Consider the setting in subsection \ref{sec:qss}.
For a representation $\rho$ of $\caA$ on $\caH$, and a cone $\Lambda$,
we set
\begin{align}
\begin{split}
&\Vbk{\rho}{\Lambda}:=
\left\{
V_{\rho\Lambda}\in\caU\lmk \caH\rmk
\mid \left.\Ad\lmk V_{\rho\Lambda}\rmk\circ\rho\right\vert_{\caA_{\Lambda^c}}
=\left.\pi \right\vert_{\caA_{\Lambda^c}}
\right\}.
\end{split}
\end{align}
We say a representation $\rho$ of $\caA$ on $\caH$
satisfies the superselection criterion for $\pi$ if 
$\Vbk{\rho}{\Lambda}$ is not empty
for any cone $\Lambda$ in $\bbR^2$.
We denote by $O$  the set of all representations of $\caA$ on $\caH$
satisfying the superselection criterion for $\pi$.
For a cone $\Lambda$, we denote by
$O_\Lambda$ the set of all representations of $\caA$ on $\caH$
satisfying the superselection criterion for $\pi$
such that
\begin{align}
 \left.\rho\right\vert_{\caA_{\Lambda^c}}
=\left.\pi \right\vert_{\caA_{\Lambda^c}}.
\end{align}

For the rest of the paper, we fix $\varphi_0\in (0,\frac\pi 2)$
a pair of cones $(\loz,\ltz)\in PC$, and cones $\lz\in\Cbk^U\cap\caC_{(0,\varphi_0)}$,
$\lhz\in\Cl\cap \caC_{(0,\varphi_0)}$  such that
$\lz\subset \lmk \loz\cup\ltz\rmk^c\subset\lhz $. 
Note that $O_{\lz}\subset O_{\lhz}$.
Recall from Lemma 2.13 of \cite{MTC} that for each
$\rho\in O_{\lhz}$,
 there exists a unique $*$-homomorphism $T_\rho^{\unit}
: \caB_{(0,\varphi_0)}\to \caB(\caH)$ such that 
\begin{description}
\item[(i)]
$T_\rho^{\unit}$ is $\sigma$w-continuous on $\pi\lmk \caA_\Lambda\rmk''$
for all $\Lambda\in\caC_{(0,\varphi_0)}$,
\item[(ii)]
$T_\rho^{\unit}\circ\pi(A)=\rho(A)$,
for all $A\in \caA$.
\end{description}
This $T_\rho^{\unit}$ is an endomorphism on $\caB_{(0,\varphi_0)}$.

The following was shown in Theorem 5.2 \cite{MTC}.
(Note for the proof there, only Assumption \ref{a1} and Assumption \ref{a2} are used.)
\begin{thm}\label{MTC}
Consider the setting in subsection \ref{sec:qss}.
Assume Assumption \ref{a1} and Assumption \ref{a2}.
There exists a braided $C^*$-tensor category $\hat C$ with the following structures.
\begin{enumerate}
\item The set of all objects $\Obj \hat C$  is  $O_{\lhz}$.
\item For $\rho,\sigma\in \Obj \hat C$,
morphisms between them are
given as intertwiners:
\begin{align}
\Mor(\rho,\sigma):=
\left\{ X\in \caB(\caH)\mid X\rho(A)=\sigma(A) X,\quad \text{for all}\quad A\in {\caA}\right\}.
\end{align}
\item The tensor product of $\rho,\sigma\in \Obj \hat C$ is given by
\begin{align}
\rho\otimes\sigma:=T_\rho^{\unit}T_\sigma^{\unit}\pi.
\end{align}
For morphisms $X\in \Mor(\rho,\rho')$, $Y\in \Mor(\sigma,\sigma')$,
with $\rho,\rho',\sigma,\sigma'\in \Obj \hat C$,
the tensor product is given by
\begin{align}
X\otimes Y:= XT_\rho^\unit (Y).
\end{align}
\end{enumerate}
The full subcategory of $\hat C$ consisting of
objects $O_{\lz}$ forms a sub braided $C^*$-tensor category $C$
 of $\hat C$.
\end{thm}
\begin{nota}
We denote by $\epsilon(\rho,\sigma)$ the braiding of
the category $\hat C$ for $\rho,\sigma\in\Obj\hat C$.
\end{nota}
\subsection{Group action on the ground state $\omega$}
Consider the setting in subsection \ref{sec:qss}.
We will see that fractionalization occurs if
half-space group action results in a local excitation at the boundary.
The latter condition is mathematically described as follows.
\begin{assum}\label{a5} [Local boundary excitation] We assume the following.
\begin{description}
\item[(i)]
For any cone $\ld$, $0<\varepsilon<\frac 12\min \{|\arg\ld|, 2\pi-|\arg\ld|\}$,
and $g\in G$, there exits an automorphism $\gamma_{g,\ld}\in \Aut\caA$,
$\gamma_{g,\ld,\varepsilon}\in \Aut\caA_{(\partial \ld)_\varepsilon}$,
and $u_{g,\ld,\varepsilon}\in \caU(\caA)$ such that
\begin{align}
&\omega\beta_{g}^\ld=\omega\gamma_{g,\ld},\quad g\in G,\quad\label{takao}\\
&\gamma_{g,\ld}=\Ad\lmk u_{g,\ld,\varepsilon}\rmk\gamma_{g,\ld,\varepsilon}.
\end{align}
Let $v_{g,\ld}\in \caU(\caH)$ be the unitary given by (\ref{takao})
such that
\begin{align}
v_{g,\ld}\pi(A)\Omega=\pi\lmk\beta_{g}^\ld \gamma_{g,\ld}^{-1}(A)\rmk\Omega,\quad
A\in \caA.
\end{align}
\item[(ii)]
For any $\ld, \ld'\in \Cl$,
$v_{g,\ld}\caF v_{g,\ld'}^*, v_{g,\ld}^*\caF v_{g,\ld'}\subset \caF$
hold.
\end{description}
\end{assum}
This property holds when $\omega$ is a short range entangled state \cite{O3}.
It also holds when $\omega$ is a gapped ground state and $G=U(1)$.

\subsection{Group action on anyons}
The symmetry group acts on the braided $C^*$-tensor category 
$C$ as follows.
\begin{lem}\label{mizuumi}Consider the setting in subsection \ref{sec:qss}.
Assume Assumption \ref{a1} and Assumption \ref{a2}.
For each $g\in G$, 
\begin{align}
\begin{split}
&\Theta(g)(\rho):=\Ad R_g \rho\beta_{g^{-1}},\quad \rho\in \Obj C,\\
&\Theta(g)(S):=\Ad R_g(S),\quad S\in \Mor(\rho,\sigma),\quad \rho,\sigma\in \Obj C
\end{split}
\end{align}
defines 
an auto-equivalence $\Theta(g)\in \Aut C$
of the $C^*$-tensor category $C$.
Furthermore, $\Theta : G\ni g\mapsto \Theta(g)\in \Aut C$ is
a group homomorphism.
\end{lem}
\begin{proof}
Let us show that $\Theta(g) : C\to C$ is a functor.
First we check $\Theta(g)(\rho)\in \Obj C$
for any $\rho\in \Obj C$ and $g\in G$.
For any cone $\ld\in \Cbk$ and $V_{{\rho}{\Lambda}}\in \Vbk{\rho}{\Lambda}$,
we have $R_g V_{{\rho}{\Lambda}} R_g^*\in \Vbk{\Theta(g)(\rho)}{\Lambda}$
because
\begin{align}
\begin{split}
&\Ad\lmk R_g V_{{\rho}{\Lambda}} R_g^*\rmk\circ\Theta(g)(\rho)(A)
=\Ad\lmk R_g V_{{\rho}{\Lambda}} \rmk\rho\beta_g^{-1}(A)\\
&=\Ad\lmk R_g \rmk\pi\beta_g^{-1}(A)
=\pi (A)
\end{split}
\end{align}
for any $A\in \caA_{\ld^c}$. Here we used $\beta_g^{-1}(A)\in \caA_{\ld^c}$.
Hence $\Theta(g)(\rho)\in \Obj C$ and $R_g\lmk \Vbk{\rho}{\Lambda}\rmk R_g^*\subset \Vbk{\Theta(g)(\rho)}{\Lambda}$.
By the same argument with $\rho, g$ replaced by
$\Theta(g)(\rho), g^{-1}$, we have 
\begin{align}
R_g\lmk \Vbk{\rho}{\Lambda}\rmk R_g^*= \Vbk{\Theta(g)(\rho)}{\Lambda}
\end{align}
for all
$\ld\in \Cbk$.

For any $\rho,\sigma\in \Obj C$ and $S\in \Mor(\rho,\sigma)$, we have
\begin{align}
\begin{split}
\Theta(g)(S)\cdot \Theta(g)(\rho)(A)
=\Ad R_g\lmk S \rho\beta_{g^{-1}}(A)\rmk
=\Ad R_g\lmk \sigma\beta_{g^{-1}}(A) S\rmk
=\Theta(g)(\sigma)(A)\cdot \Theta(g)(S),\quad A\in \caA.
\end{split}
\end{align}
Hence we have $\Theta(g)(S)\in \Mor(\Theta(g)(\rho), \Theta(g)(\sigma))$.
For any $\rho\in\Obj C$, we have
$\Theta(g)(id_\rho)=\id_{\caH}=\id_{\Theta(g)(\rho)}$.
For any $\rho,\sigma,\tau\in \Obj C$ and $X\in \Mor(\rho,\sigma)$, $Y\in \Mor(\sigma, \tau)$
we have
$\Theta(g)(Y)\Theta(g)(X)=\Theta(g)(YX)$.
Hence $\Theta(g)$ is a functor.

Next, we see that $\Theta(g)$ is a tensor functor.
Note for the tensor unit $\pi$ of $\Obj C$ that
\begin{align}
\Theta(g)(\pi)=\Ad R_g\pi\beta_{g^{-1}}=\pi.
\end{align}
We also note that
\begin{align}\label{tori}
\begin{split}
T_{\Theta(g)(\rho)}^\unit
=\Ad R_g T_\rho^\unit\Ad R_g^*,\quad \rho\in \Obj C.
\end{split}
\end{align}
In fact, because $\Ad R_g$ preserves $\caB_{(0,\varphi_0)}$,
\begin{align}
\begin{split}
\Ad R_g T_\rho^\unit\Ad R_g^*
\end{split}
\end{align}
is a well-defined endomorphism of $\caB_{(0,\varphi_0)}$ which is 
$\sigma$weak continuous on each $\pi(\caA_{\ld})''$ with $\ld\in \caC_{(0,\varphi)}$
such that
\begin{align}
\begin{split}
\Ad R_g T_\rho^\unit\Ad R_g^*\pi
=\Ad R_g\rho\beta_g^{-1}=\Theta(g)(\rho).
\end{split}
\end{align}
Hence by the uniqueness, we get (\ref{tori}).
From (\ref{tori}), we obtain
\begin{align}
\begin{split}
&\Theta(g)(\rho)\otimes\Theta(g)(\sigma)
=T_{\Theta(g)(\rho)}^\unit T_{\Theta(g)(\sigma)}^\unit\pi
=\Ad R_g T_\rho^\unit\Ad R_g^*\Ad R_g T_\sigma^\unit\Ad R_g^*\pi\\
&=\Ad R_g T_\rho^\unit T_\sigma^\unit\Ad R_g^*\pi
=\Ad R_g\lmk \rho\otimes \sigma \rmk\beta_{g^{-1}}
=\Theta(g)\lmk \rho\otimes \sigma\rmk,
\end{split}
\end{align}
for all $\rho,\sigma\in\Obj C$.
Then we have
\begin{align}
\begin{split}
&\varphi_0:=\id_{\caH}\in \Mor\lmk \pi,\Theta(g)(\pi)\rmk,\\
&\varphi_2(\rho,\sigma):=\id_{\caH}
\in \Mor\lmk \Theta(g)(\rho)\otimes\Theta(g)(\sigma), \Theta(g)\lmk \rho\otimes \sigma\rmk\rmk,\quad
\rho,\sigma\in \Obj C.
\end{split}
\end{align}
We claim that $(\Theta(g), \varphi_0, \varphi_2)$ gives a tensor functor from $C$ to $C$.
Note that $\varphi_2$ is natural because
\begin{align}
\begin{split}
&\Theta(g)(X)\otimes \Theta(g)(Y)
=\Theta(g)(X)T_{\Theta(g)(\rho)}^{\unit}\lmk \Theta(g)(Y)\rmk\\
&=\Ad R_g\lmk X\rmk
\Ad R_g T_\rho^\unit\Ad R_g^*
\Ad R_g\lmk Y\rmk\\
&=\Ad R_g\lmk X T_\rho^\unit (Y)\rmk
=\Theta(g) \lmk X\otimes Y\rmk
\end{split}
\end{align}
for $\rho,\rho',\sigma,\sigma'\in \Obj C$, $X\in \Mor(\rho,\rho')$,
$Y\in \Mor(\sigma,\sigma')$.
That $\varphi_0$, $\varphi_2$ are consistent with associativity morphisms and left/right constraint is
trivial because all the involved morphisms are $\id_{\caH}$.
Hence $\Theta(g)$ is a tensor functor from $C$ to $C$.

From the definition, it is clear that
the composition of tensor functors $(\Theta(g), \id_{\caH}, \id_{\caH})$,
$(\Theta(h), \id_{\caH}, \id_{\caH})$
is $(\Theta(gh), \id_{\caH}, \id_{\caH})$.
It in particular tells us that $(\Theta(g), \id_{\caH}, \id_{\caH})$ is an auto-equivalence of
$C$, and 
$\Theta : G\ni g\mapsto \Theta(g)\in \Aut C$ is
a group homomorphism.

\end{proof}

\subsection{Main Theorem}

Let $\Irr O_{\lz}$ be the set of all irreducible elements in $O_{\lz}$,
and $[\Irr O_{\lz}]$ the set of all isomorphic classes. 
For each $\rho\in \Irr O_{\lz}$ we denote by $[\rho]$ the isomorphism class
to which $\rho$ belongs.
For each $a\in [\Irr O_{\lz}]$, we fix a representative $\rho_a\in a$.
We set
\begin{align}
a^{(g)}:=[\Theta(g)(\rho_a)]
\end{align}
for each $a\in [\Irr O_{\lz}]$ and $g\in G$.
Note that this $a^{(g)}$ does not depend on the choice of $\rho_a$.

By the definition, for each $a\in [\Irr O_{\lz}]$ and $g\in G$, 
there exists a unitary $W_a^{(g)}$ on $\caH$
 such that
 \begin{align}\label{wa}
 \Theta(g)(\rho_a)=\Ad(W_a^{(g)})\rho_{a^{(g)}}.
 \end{align}
 From the fact that $\Theta$ is a group homomorphism, we have the following Lemma.
 \begin{lem}Consider the setting in subsection \ref{sec:qss}.
 For any $g,h\in G$ and  $a\in [\Irr O_{\lz}]$, we have
\begin{align}
\omega^{(a)}(g,h):
=\lmk W_{a^{(h)}}^{(g)}\rmk^* R_g \lmk W_a^{(h)}\rmk^* R_g^*  W_a^{(gh)} 
\in U(1).
\end{align}
\end{lem}
\begin{proof}
 From (\ref{tori}) in the proof of Lemma \ref{mizuumi},  we have
 \begin{align}\label{kaki}
\Ad R_g T_{\rho_a}^\unit \Ad R_g^*=\Ad(W_a^{(g)})T_{\rho_{a^{(g)}}}^\unit. 
 \end{align}
 Because of 
 \begin{align}
 \begin{split}
 &\Ad(W_a^{(gh)})\rho_{a^{(gh)}}=\Theta(gh) (\rho_a)
 =\Theta(g)\lmk \Theta(h)(\rho_a)\rmk\\
 &=\Theta(g)\Ad(W_a^{(h)})\rho_{a^{(h)}}
 =\Ad\lmk\Theta(g)\lmk W_a^{(h)}\rmk\rmk
 \Theta(g)\rho_{a^{(h)}},
  \end{split}
 \end{align}
we have
\begin{align}
\lmk a^{(h)}\rmk^{(g)}=a^{(gh)} 
\end{align}
and 
\begin{align}
\Ad\lmk R_g W_a^{(h)} R_g^* W_{a^{(h)}}^{(g)}\rmk\circ\rho_{a^{(gh)}}
=\Ad W_a^{(gh)} \rho_{a^{(gh)}}.
\end{align}
By the irreducibility of $ \rho_{a^{(gh)}}$,
\begin{align}
\omega^{(a)}(g,h):
=\lmk W_{a^{(h)}}^{(g)}\rmk^* R_g \lmk W_a^{(h)}\rmk^* R_g^*  W_a^{(gh)} 
\in U(1).
\end{align}
\end{proof}
We consider a $G$-module 
\begin{align}
\begin{split}
\caM:=\bigoplus_{a\in [\Irr O_{\lz}]} U(1)
\end{split}
\end{align}
with right $G$-action
\begin{align}
\begin{split}
\lmk x_a\rmk_{a\in [\Irr O_{\lz}]}\mapsto \lmk x_{a^{(g^{-1})}}\rmk_{a\in [\Irr O_{\lz}]},\quad g\in G.
\end{split}
\end{align}
\begin{lem}\label{fushimi}Consider setting in the subsection \ref{sec:qss}.
If we set
\begin{align}
\begin{split}
&\eta(g,h):=(\eta_a(g,h))_{a\in [\Irr O_{\lz}]}\in \caM,\\
&\eta_a(g,h):= \omega^{(a^{(gh)^{-1}})}(g,h),\\
&g,h\in G^{\times 2},\quad a\in [\Irr O_{\lz}],
\end{split}
\end{align}
then $\eta\in Z^2(G,\caM)$.
The cohomology class $[\eta]_{Z^2(G,\caM)}$
is independent of the choice of $W_a^{(g)}$s.
\end{lem}
\begin{proof}
By definition
\begin{align}
\begin{split}
&\eta_a(g,h):= \omega^{(a^{(gh)^{-1}})}(g,h)\\
&=\lmk W_{a^{(g^{-1})}}^{(g)}\rmk^* R_g \lmk W_{a^{(gh)^{-1}}}^{(h)}\rmk^* R_g^*  W_{a^{(gh)^{-1}}}^{(gh)} 
\end{split}
\end{align}
we have
\begin{align}
\begin{split}
&\lmk \eta_a(g,h)\rmk^*=
\lmk W_{a^{(gh)^{-1}}}^{(gh)} \rmk^* R_g
W_{a^{(gh)^{-1}}}^{(h)} R_g^*
 W_{a^{(g^{-1})}}^{(g)},
\end{split}
\end{align}
\begin{align}
\begin{split}
&\eta_{a^{(g^{-1}) }}(h, k)=
\lmk W_{(a^{(g^{-1}) })^{(h^{-1})}}^{(h)}\rmk^* R_h \lmk W_{(a^{(g^{-1}) })^{(hk)^{-1}}}^{(k)}\rmk^* R_h^*  W_{(a^{(g^{-1}) })^{(hk)^{-1}}}^{(hk)} \\
&=\lmk W_{(a^{((gh)^{-1}) })}^{(h)}\rmk^* R_h \lmk W_{(a^{((ghk)^{-1}) })}^{(k)}\rmk^* R_h^*  W_{(a^{((ghk)^{-1}) })}^{(hk)} \\
\end{split}
\end{align}

Then we have
\begin{align}
\begin{split}
&\eta_{a^{(g^{-1}) }}(h, k) \eta_a(g,hk)
\eta_a(gh, k)^*\eta_a(g,h)^*\\
&=
\lmk W_{a^{(g^{-1})}}^{(g)}\rmk^* R_g
\eta_{a^{(g^{-1}) }}(h, k) \lmk W_{a^{(ghk)^{-1}}}^{(hk)}\rmk^* R_g^*  W_{a^{(ghk)^{-1}}}^{(ghk)} \\
&\lmk W_{a^{(ghk)^{-1}}}^{(ghk)} \rmk^* R_{gh}
W_{a^{(ghk)^{-1}}}^{(k)} R_{gh}^*
 W_{a^{((gh)^{-1})}}^{(gh)}\\
& \lmk W_{a^{(gh)^{-1}}}^{(gh)} \rmk^* R_g
W_{a^{(gh)^{-1}}}^{(h)} R_g^*
 W_{a^{(g^{-1})}}^{(g)}\\
 &=
 \lmk W_{a^{(g^{-1})}}^{(g)}\rmk^* R_g
\lmk W_{(a^{((gh)^{-1}) })}^{(h)}\rmk^* R_h \lmk W_{(a^{((ghk)^{-1}) })}^{(k)}\rmk^* R_h^*  W_{(a^{((ghk)^{-1}) })}^{(hk)} \\
\\
&\lmk W_{a^{(ghk)^{-1}}}^{(hk)}\rmk^* R_g^*  W_{a^{(ghk)^{-1}}}^{(ghk)} \\
&\lmk W_{a^{(ghk)^{-1}}}^{(ghk)} \rmk^* R_{gh}
W_{a^{(ghk)^{-1}}}^{(k)} R_{gh}^*
 W_{a^{((gh)^{-1})}}^{(gh)}\\
& \lmk W_{a^{(gh)^{-1}}}^{(gh)} \rmk^* R_g
W_{a^{(gh)^{-1}}}^{(h)} R_g^*
 W_{a^{(g^{-1})}}^{(g)}\\
 &=\unit.
\end{split}
\end{align}
Hence $\eta\in Z^2(G,\caM)$.
The cohomology class $[\eta]_{Z^2(G,\caM)}$
is independent of the choice of $W_a^{(g)}$s, because
from the irreducibility of $\rho_a$
there is only a phase freedom of choice for $W_a^{(g)}$,
which ends up with a coboundary.
\end{proof}
\begin{lem}\label{washi}Consider setting in the subsection \ref{sec:qss}.Assume Assumption \ref{a1}.
For $a,b\in [\Irr O_{\lz}]$, and $g\in G$ set 
\begin{align}
Y_{a,b}^{(g)}:= \lmk W_a^{(g)} T_{\rho_{a^{(g)}}}^{\unit}\lmk W_b^{(g)}\rmk\rmk^*. 
\end{align}
Then we have\begin{align}\label{usagi}
\begin{split}
&Y_{a^{(h)},b^{(h)}}^{(g)}R_g Y_{a,b}^{(h)} R_g^*
\lmk Y_{a,b}^{(gh)}\rmk^*=\omega^{(a)}(g,h)\omega^{(b)}(g,h).
\end{split}
\end{align}
(Note that because of the approximate Haag duality, $W_n^{(g)}\in \caB_{(0,\varphi_0)}$.)
\end{lem}
\begin{proof}
Using (\ref{kaki}),
\begin{align}\label{usagi}
\begin{split}
&Y_{a^{(h)},b^{(h)}}^{(g)}R_g Y_{a,b}^{(h)} R_g^*
\lmk Y_{a,b}^{(gh)}\rmk^*\\
&=T_{\rho_{a^{(gh)}}}^{\unit}\lmk\lmk  W_{b^{(h)}}^{(g)}\rmk^*\rmk \lmk W_{a^{(h)}}^{(g)} \rmk^*
R_g T_{\rho_{a^{(h)}}}^{\unit}\lmk\lmk  W_b^{(h)}\rmk^*\rmk \lmk W_a^{(h)}\rmk^* R_g^*
W_a^{(gh)} T_{\rho_{a^{(gh)}}}^{\unit}\lmk W_b^{(gh)}\rmk\\
&=T_{\rho_{a^{(gh)}}}^{\unit}\lmk\lmk  W_{b^{(h)}}^{(g)}\rmk^*\rmk \lmk W_{a^{(h)}}^{(g)} \rmk^*
R_g T_{\rho_{a^{(h)}}}^{\unit}\lmk\lmk  W_b^{(h)}\rmk^*\rmk 
R_g^* W_{a^{(h)}}^{(g)}
\lmk W_{a^{(h)}}^{(g)} \rmk^*
R_g
\lmk W_a^{(h)}\rmk^* R_g^*
W_a^{(gh)} T_{\rho_{a^{(gh)}}}^{\unit}\lmk W_b^{(gh)}\rmk\\
&=T_{\rho_{a^{(gh)}}}^{\unit}\lmk\lmk  W_{b^{(h)}}^{(g)}\rmk^*\rmk 
\lmk W_{a^{(h)}}^{(g)} \rmk^* R_g
 T_{\rho_{a^{(h)}}}^{\unit}\lmk\lmk  W_b^{(h)}\rmk^*\rmk 
R_g^* W_{a^{(h)}}^{(g)}
\omega^{(a)}(g,h) T_{\rho_{a^{(gh)}}}^{\unit}\lmk W_b^{(gh)}\rmk\\
&=T_{\rho_{a^{(gh)}}}^{\unit}\lmk\lmk  W_{b^{(h)}}^{(g)}\rmk^*\rmk
T_{\rho_{a^{(gh)}}}^{\unit}\Ad R_g\lmk\lmk  W_b^{(h)}\rmk^*\rmk 
\omega^{(a)}(g,h) T_{\rho_{a^{(gh)}}}^{\unit}\lmk W_b^{(gh)}\rmk\\
&=\omega^{(a)}(g,h)
T_{\rho_{a^{(gh)}}}^{\unit}\lmk
\lmk  W_{b^{(h)}}^{(g)}\rmk^*
\Ad R_g\lmk W_b^{(h)}\rmk^*
W_b^{(gh)}\rmk \\
&=\omega^{(a)}(g,h)\omega^{(b)}(g,h).
\end{split}
\end{align}
\end{proof}

Here is the main theorem of this paper.
\begin{thm}\label{mainthm}
Consider setting in the subsection \ref{sec:qss}.
Assume Assumption \ref{a1}, Assumption \ref{a1r}, Assumption \ref{a1l},  
and Assumption \ref{a5}.
Then 
\begin{enumerate}
\item for any $a\in [\Irr O_{\lz}]$ and $g\in G$, we have
$a^{(g)}=a$, and 
\item we may choose $W_a^{(g)}$ so that 
we have 
\begin{align}
\omega^{(a)}(g,h)\omega^{(b)}(g,h)=\omega^{(c)}(g,h),\quad g,h\in G.
\end{align}
for any $a, b,c\in [\Irr O_{\lz}]$ with $\Mor(\rho_a\otimes\rho_b,\rho_c)\neq\{0\}$.
\end{enumerate}

\end{thm}
In the next section, as a tool to prove this theorem, we introduce 
$G$-localized superselection sectors, an analog of the framework
in \cite{muger2005conformal}.
Using this, we prove Theorem \ref{mainthm} in section \ref{frac}.
In section \ref{suf}, we introduce a suffient condition for Assumption \ref{a5}
which can be used to analyze concrete models.
In section \ref{ex}, we give an concrete example satisfying all the assumptions.
A G-crossed category is formulated in Appendix \ref{gcross}.

\section{ $G$-localized superselection sectors and their braidings}\label{gsu}
In this section we consider $g$-localized superselection sectors and their braidings,
as was done in \cite{muger2005conformal} in one-dimensional systems.
The argument is a deformation of that of \cite{muger2005conformal},
taking account of ``tails'' due to approximate Haag duality.
It is carried out using the same argument in \cite{MTC}, \cite{BA}.
\begin{defn}Consider the setting in the subsection \ref{sec:qss}.
For a representation $\rho$ of $\caA$ on $\caH$, $g\in G$
and $(\lo,\lt)\in PC$,
set
\begin{align}
\begin{split}
&\Vbk{\rho}{(\lo,\lt)}^{(g)}:=
\left\{
V_{{\rho}{(\lo,\lt)}}^{(g)}\in\caU\lmk \caF\rmk
\mid \left.\Ad\lmk V_{{\rho}{(\lo,\lt)}}^{(g)}\rmk\circ\rho\right\vert_{\caA_{\lo\cup\lt}}
=\left.\pi \circ\beta_g^{\Lambda_1}\right\vert_{\caA_{\lo\cup\lt}}
\right\}.
\end{split}
\end{align}
For $g\in G$, we set
\begin{align}
\begin{split}
&O^{(g)}:=
\left\{
\rho\in O\mid
\Vbk{\rho}{(\lo,\lt)}^{(g)}\neq\emptyset,\quad\text{for all}\quad
(\lo,\lt)\in PC
\right\}.\\
&O^{(g)}_{(\loz,\ltz)}:=
\left\{
\rho\in O^{(g)}\mid
\unit\in \Vbk{\rho}{(\loz,\ltz)}^{(g)}
\right\}.
\end{split}
\end{align}
(Recall that $O$ is the set of all representations of $\caA$ on $\caH$ satisfying the superselection criterion.)
We set
\begin{align}
\begin{split}
O_G:=\cup_{g\in G} O^{(g)}
\end{split}
\end{align}
Note that 
\begin{align}
\begin{split}
\Obj C=O_{\lz}\subset O^{(e)},\quad
O_G\subset \Obj \hat C,
\end{split}
\end{align}
because of 
$\lz\subset (\loz\cup\ltz)^c\subset \lhz$, $\lz\in \Cbk$, and Lemma 2.4 of \cite{MTC}.
For $\rho\in O^{(g)}_{(\loz,\ltz)}$ and 
$\sigma\in O^{(h)}_{(\loz,\ltz)}$ with $g,h\in G$,
we set the morphisms between them as
\begin{align}
\Mor_G\lmk\rho,\sigma\rmk:=
\left\{
X\in \caF\mid
X\rho(A)=\sigma(A)X,\quad A\in \caA 
\right\}.
\end{align}

\end{defn}

\subsection{Tensor product}
By the same argument as Lemma 2.1 of \cite{BA},
we have the following.
\begin{lem}Consider setting in the subsection \ref{sec:qss}.
Let $g\in G$ and $\rho\in O^{(g)}_{(\loz,\ltz)}$.
Then 
there exists a unique $*$-homomorphism $T_\rho^{(l)\unit}: \caB_l\to\caB(\caH)$
such that
\begin{description}
\item[(i)]
$T_\rho^{(l)\unit}\pi\beta_g =\rho$,
\item[(ii)]
$T_\rho^{(l)\unit}$ is $\sigma$-weak continuous on
$\pi(\caA_{\Lambda_l})''$ for all $\Lambda_l\in\Cl$.
\end{description}
It satisfies $T_\rho^{(l)\unit}(\bl)\subset \bl$ and defines an endomorphism
$T_\rho^{(l)\unit} : \bl\to \bl$.
Furthermore we have $T_\rho^{(l)\unit}(\fbd)\subset \fbd$.
\end{lem}
\begin{proof}
Each $\lm l\in\Cl$ can be written as 
$\lm l=\Lambda_{ (a,0), \pi, \varphi}\in\Cl$.
Set
\[
\kappa_{\lm l}:=\left\{
(\Gamma_l, \Gamma_r)\in PC \mid\Gamma_l=\Lambda_{ (b,0), \pi, \varphi_1},\;
a<b,\; 0<\varphi<\varphi_1<\pi
\right\}.
\]
Then, as in Lemma 2.11 of \cite{MTC},
for any $\rho\in O^{(g)}_{(\loz,\ltz)}$, 
\begin{align}
T_\rho^{(0)}(x):=\Ad\lmk {V_{\rho,(\Gamma_l, \Gamma_r)}^{(g)}}^*\rmk(x),\quad\text{if}\;
x\in \pi\lmk\caA_{\lm l }\rmk'',\;\lm l\in\Cl
\end{align}
defines an isometric $*$-homomorphism $T_\rho^{(0)}: \caB_l^{(0)}:=
\cup_{\Lambda_l\in \Cl}\; \pi\lmk\caA_{\Lambda_l}\rmk''\to\caB_l$,
independent of the choice of $(\Gamma_l, \Gamma_r)\in \kappa_{\lm l}$, ${V_{\rho,(\Gamma_l, \Gamma_r)}^{(g)}}
\in \Vbk{\rho}{(\Gamma_l,\Gamma_r)}^{(g)}$, $\lm l$.
As in Lemma 2.13 of \cite{MTC}, this $T_\rho^{(0)}$ extends to
the $T_\rho^{(l)\unit}$ with the desired property.
In particular, for $A\in \lm l$,
\begin{align}
T_\rho^{(l)\unit}\pi\beta_g(A)
=\Ad\lmk {V_{\rho,(\Gamma_l, \Gamma_r)}^{(g)}}^*\rmk\pi\beta_g(A)
=\rho(A),
\end{align}
with $(\Gamma_l, \Gamma_r)\in \kappa_{\lm l}$, ${V_{\rho,(\Gamma_l, \Gamma_r)}^{(g)}}
\in \Vbk{\rho}{(\Gamma_l,\Gamma_r)}^{(g)}$
because $\lm l\subset\Gamma_l$ by definition.
The last statement is trivial from the definition above, because
$V_{\rho,(\Gamma_l, \Gamma_r)}^{(g)}\in \caF$.
\end{proof}
From this, we obtain the following.
\begin{lem}\label{kiri}
Consider setting in the subsection \ref{sec:qss}.
Let $g\in G$ and $\rho\in O^{(g)}_{(\loz,\ltz)}$.
Then 
there exists a unique $*$-homomorphism $S_\rho^{(l)\unit}: \caB_l\to\caB(\caH)$
such that
\begin{description}
\item[(i)]
$S_\rho^{(l)\unit}\pi=\rho$,
\item[(ii)]
$S_\rho^{(l)\unit}$ is $\sigma$-weak continuous on
$\pi(\caA_{\Lambda_l})''$ for all $\Lambda_l\in\Cl$.
\end{description}
It satisfies $S_\rho^{(l)\unit}(\bl)\subset \bl$ and defines an endomorphism
$S_\rho^{(l)\unit} : \bl\to \bl$.
Furthermore we have $S_\rho^{(l)\unit}(\fbd)\subset \fbd$.
\end{lem}
\begin{proof}
Note that $\Ad R_g(\caF)=\caF$ and $\Ad R_g(\caB_l)=\caB_l$.
Therefore,
\begin{align}
S_\rho^{(l)\unit}:=T_\rho^{(l)\unit}\Ad R_g : \caB_l\to \caB_l
\end{align}
is well defined and satisfies the desired properties.
\end{proof}
We have the right version of this.
\begin{lem}\label{kaze}
Consider setting in the subsection \ref{sec:qss}.
Let $g\in G$ and $\rho\in O^{(g)}_{(\loz,\ltz)}$.
Then 
there exists a unique $*$-homomorphism $S_\rho^{(r)\unit}: \caB_r\to\caB(\caH)$
such that
\begin{description}
\item[(i)]
$S_\rho^{(r)\unit}\pi=\rho$,
\item[(ii)]
$S_\rho^{(r)\unit}$ is $\sigma$-weak continuous on
$\pi(\caA_{\Lambda_r})''$ for all $\Lambda_r\in\Cr$.
\end{description}
It satisfies $S_\rho^{(r)\unit}(\brr)\subset \brr$ and defines an endomorphism
$S_\rho^{(r)\unit} : \brr\to \brr$.
Furthermore we have $S_\rho^{(l)\unit}(\fbd)\subset \fbd$.
\end{lem}
\begin{lem}\label{slsr}Consider setting in the subsection \ref{sec:qss}.
Let $g\in G$ and $\rho\in O^{(g)}_{(\loz,\ltz)}$.
The following holds.
\begin{enumerate}
\item $\left. S_\rho^{(l)\unit}\right\vert_\caF =\left. S_\rho^{(r)\unit}\right\vert_\caF$.
\item $\left. S_\rho^{(l)\unit}\right\vert_{\caB_{(0,\varphi_0)}} 
=\left. T_\rho^{\unit}\right\vert_{\caB_{(0,\varphi_0)}} $
\end{enumerate}
\end{lem}
\begin{proof}
{\it 1. }follows from the fact that for each $(\lo,\lt)\in PC$, there exist
$\tilde \lo\in \Cl$ and $\tilde \lt\in\Cr$ such that 
\begin{align}
\lmk \lo \cup \lt \rmk^c\subset \tilde \lo\cap\tilde \lt
\end{align}
and the property (i), (ii) Lemma \ref{kiri}, Lemma \ref{kaze} of each map.
{\it 2.} follows from the fact that for each $\ld\in \caC_{(0,\varphi_0)}$
there exists a $\lo\in \Cl$ such that $\ld\subset \lo $.
\end{proof}
\begin{nota}Consider setting in the subsection \ref{sec:qss}.
For each $\rho\in O^{(g)}_{(\loz,\ltz)}$, $(\lo,\lt)\in PC$ and
$V_{{\rho}{(\lo,\lt)}}^{(g)} \in \Vbk{\rho}{(\lo,\lt)}^{(g)}$,
we set
\begin{align}
\begin{split}
 S_\rho^{(l)V_{{\rho}{(\lo,\lt)}}^{(g)}}:=\Ad\lmk V_{{\rho}{(\lo,\lt)}}^{(g)}\rmk\circ S_\rho^{(l)\unit},\\
 S_\rho^{(r)V_{{\rho}{(\lo,\lt)}}^{(g)}}:=\Ad\lmk V_{{\rho}{(\lo,\lt)}}^{(g)}\rmk\circ S_\rho^{(r)\unit}.
\end{split}
\end{align}
They are endomorphisms on $\caB_l$, $\caB_r$ respectively, preserving $\caF$.
\end{nota}
\begin{lem}\label{ookami}Consider setting in the subsection \ref{sec:qss}.
For each $\rho\in O^{(g)}_{(\loz,\ltz)}$, $(\lo,\lt)\in PC$ and
$V_{{\rho}{(\lo,\lt)}}^{(g)} \in \Vbk{\rho}{(\lo,\lt)}^{(g)}$,
we have
\begin{align}
\begin{split}
\left. S_\rho^{(l)V_{{\rho}{(\lo,\lt)}}^{(g)}}\right\vert_{\pdc{\lo}}
=\left. \Ad(R_g)\right\vert_{\pdc{\lo}},\\
\left. S_\rho^{(r)V_{{\rho}{(\lo,\lt)}}^{(g)}}\right\vert_{\pdc{\lt}}
=\left. \id\right\vert_{\pdc{\lt}},\\
\left. S_\rho^{(l)V_{{\rho}{(\lo,\lt)}}^{(g)}}\right\vert_{\pdc{\lt}\cap \caF}
=\left. \id\right\vert_{\pdc{\lt}\cap\caF},\\
\left. S_\rho^{(r)V_{{\rho}{(\lo,\lt)}}^{(g)}}\right\vert_{\pdc{\lo}\cap\caF}
=\left. \Ad R_g\right\vert_{\pdc{\lo}\cap\caF}.
\end{split}
\end{align}
\end{lem}
\begin{proof}
The first half is immediate from the $\sigma$weak continuity of the maps
and 
\begin{align}
\begin{split}
\left.\Ad\lmk V_{{\rho}{(\lo,\lt)}}^{(g)}\rmk\circ\rho\right\vert_{\caA_{\lo\cup\lt}}
=\left.\pi \circ\beta_g^{\Lambda_1}\right\vert_{\caA_{\lo\cup\lt}}.
\end{split}
\end{align}
The second half follows from the first half and Lemma \ref{slsr}.

\end{proof}
Now we consider tensor product.
\begin{lem}\label{framingo}Consider setting in the subsection \ref{sec:qss}.
For $\rho\in O^{(g)}_{(\loz,\ltz)}$ and $\sigma\in O^{(h)}_{(\loz,\ltz)}$
with $g,h\in G$, we have
\begin{align}
\rho\otimes\sigma= S_\rho^{(l)\unit}  S_\sigma^{(l)\unit}\pi\in O^{(gh)}_{(\loz,\ltz)}
\end{align}
Here, $\otimes$ is restriction of the tensor product given in Theorem \ref{MTC}.
\end{lem}
\begin{proof}
Set $\rho\otimes_G\sigma= S_\rho^{(l)\unit}  S_\sigma^{(l)\unit}\pi$. Because of Lemma \ref{slsr}  2, we have
$ S_\sigma^{(l)\unit}\vert_{\caB_{(0,\varphi_0)}}
=T_\sigma^{\unit}\vert_{\caB_{(0,\varphi_0)}}$
and 
$ S_\rho^{(l)\unit}\vert_{\caB_{(0,\varphi_0)}}
=T_\rho^{\unit}\vert_{\caB_{(0,\varphi_0)}}$.
Therefore, we have
\begin{align}
\begin{split}
\rho\otimes_G\sigma= S_\rho^{(l)\unit}  S_\sigma^{(l)\unit}\pi
= S_\rho^{(l)\unit}  T_\sigma^{\unit}\pi= T_\rho^{\unit}T_\sigma^{\unit}\pi
=\rho\otimes \sigma\in O_{\lz}.
\end{split}
\end{align}
For any $(\lo,\lt)\in PC$, from Lemma \ref{ookami}, we have 
\begin{align}\begin{split}
&\left.\rho\otimes_G\sigma\right\vert_{\caA_{\lo\cup \lt}}
=\left. S_\rho^{(l)\unit}  S_\sigma^{(l)\unit}\pi\right\vert_{\caA_{\lo\cup \lt}}\\
&=\left. \Ad\lmk \lmk V_{{\rho},{(\lo,\lt )}}^{(g)}\rmk^*\rmk S_\rho^{(l)V_{{\rho},{(\lo,\lt)}}^{(g)}}
 \circ \Ad\lmk \lmk V_{{\sigma},{(\lo,\lt )}}^{(h)}\rmk^*\rmk
S_\sigma^{(l)V_{{\sigma},{(\lo,\lt)}}^{(h)}}
  \pi\right\vert_{\caA_{\lo\cup \lt}}\\
&=\left. 
\Ad\lmk 
\lmk V_{{\rho},{(\lo,\lt )}}^{(g)}\rmk^*
S_\rho^{(l)V_{{\rho},{(\lo,\lt)}}^{(g)}}
\lmk\lmk V_{{\sigma},{(\lo,\lt )}}^{(h)}\rmk^*\rmk
\rmk 
S_\rho^{(l)V_{{\rho},{(\lo,\lt)}}^{(g)}}
S_\sigma^{(l)V_{{\sigma},{(\lo,\lt)}}^{(h)}}
  \pi\right\vert_{\caA_{\lo\cup \lt}}\\
  &=\left. 
\Ad\lmk 
\lmk V_{{\rho},{(\lo,\lt )}}^{(g)}\rmk^*
S_\rho^{(l)V_{{\rho},{(\lo,\lt)}}^{(g)}}
\lmk\lmk V_{{\sigma},{(\lo,\lt )}}^{(h)}\rmk^*\rmk
\rmk 
  \pi\beta_{gh}^{\lo}
  \right\vert_{\caA_{\lo\cup \lt}}\\
\end{split}
\end{align}
with 
$V_{{\rho},{(\lo,\lt)}}^{(g)} \in \Vbk{\rho}{(\lo,\lt)}^{(g)}$,
$V_{{\sigma},{(\lo,\lt)}}^{(h)} \in \Vbk{\sigma}{(\lo,\lt)}^{(h)}$.
Note that $\lmk V_{{\rho},{(\lo,\lt )}}^{(g)}\rmk^*
S_\rho^{(l)V_{{\rho},{(\lo,\lt)}}^{(g)}}
\lmk\lmk V_{{\sigma},{(\lo,\lt )}}^{(h)}\rmk^*\rmk
$ belongs to $\caU(\caF)$ and 
we may take $V_{{\rho},{(\loz,\ltz)}}^{(g)}:=\unit  $,
$V_{{\sigma},{(\loz,\ltz)}}^{(h)}=\unit$.
This proves $\rho\otimes_G\sigma\in O^{(gh)}_{(\loz,\ltz)}$.
\end{proof}

\begin{lem}
Consider setting in the subsection \ref{sec:qss}.
For any $\rho,\rho'\in O^{(g)}_{(\loz,\ltz)}$ and $\sigma,\sigma'\in O^{(h)}_{(\loz,\ltz)}$,
$X\in\Mor_G(\rho,\rho')$, $Y\in \Mor_G(\sigma,\sigma')$,
we have
\begin{align}\label{ame}
\begin{split}
&X\otimes Y=X S_\rho^{(l)\unit}(Y)\in \Mor_G(\rho\otimes\sigma,\rho'\otimes\sigma').
\end{split}
\end{align}
Here, $\otimes$ is restriction of the tensor product given in Theorem \ref{MTC}.
\end{lem}
Hence we obtain the following proposition
\begin{prop}
Consider setting in subsection \ref{sec:qss}.Assume Assumption \ref{a1} and Assumption \ref{a2}.
The  the full subcategory $\caE_G$ of $\hat C$
with objects $O_G$ is a full sub $C^*$-tensor category of
$\hat C$.
\end{prop}

\subsection{$\Theta(g)$ on $O_G$}
We extend the definition of $\Theta(g)$ to $O_G$. 
\begin{defn}
For each $g\in G$, 
\begin{align}\label{ninnaji}
\begin{split}
&\Theta(g)(\rho):=\Ad R_g \rho\beta_{g^{-1}},\quad \rho\in O_G,\\
&\Theta(g)(S):=\Ad R_g(S),\quad S\in \Mor_G(\rho,\sigma),\quad \rho,\sigma\in O_G.
\end{split}
\end{align}

\end{defn}
\begin{lem}\label{tokage}Consider setting in the subsection \ref{sec:qss}.
If $\rho\in O^{(h)}_{(\loz,\ltz)}$,
then we have
\begin{enumerate}
\item
$\Theta(g)(\rho)\in O^{(ghg^{-1})}_{(\loz,\ltz)}$
\item
$\Theta(g)\lmk \Vbk{\rho}{(\lo,\lt)}^{(h)}\rmk 
=\Vbk{\Theta(g)\lmk \rho\rmk}{(\lo,\lt)}^{(ghg^{-1})}$
\item
$S_{\Theta(g)(\rho)}^{(l),\Theta(g)\lmk V_{{\rho},{(\lo,\lt)}}^{(h)}\rmk}
=\Ad R_g\circ S_{\rho}^{(l), V_{{\rho},{(\lo,\lt)}}^{(h)}}\circ \Ad R_g^*$
\end{enumerate}
\end{lem}
\begin{proof}
For any $(\lo,\lt)\in PC$
and $A\in \caA_{\lo\cup \lt}$,
we have
\begin{align}
\begin{split}
&\Theta(g)(\rho)(A)
=\Ad R_g\rho\beta_g^{-1}(A)
=\Ad R_g\Ad \lmk \lmk V_{{\rho},{(\lo,\lt)}}^{(h)}\rmk^*\rmk
\pi\lmk\beta_h^{\lo}\rmk\beta_g^{-1}(A)\\
&=\Ad \lmk R_g \lmk V_{{\rho},{(\lo,\lt)}}^{(h)}\rmk^* R_g^*\rmk
\Ad R_g \pi\lmk\beta_h^{\lo}\rmk\beta_g^{-1}(A)\\
&=\Ad \lmk R_g \lmk V_{{\rho},{(\lo,\lt)}}^{(h)}\rmk^* R_g^*\rmk
\pi\lmk\beta_g\beta_h^{\lo}\beta_g^{-1}(A)\rmk
=\Ad \lmk R_g \lmk V_{{\rho},{(\lo,\lt)}}^{(h)}\rmk^* R_g^*\rmk
\pi\lmk\beta_{ghg^{-1}}^{\lo}(A)\rmk\\
\end{split}
\end{align}
This means {\it 1.} holds because $R_g \lmk V_{{\rho},{(\lo,\lt)}}^{(h)}\rmk^* R_g^*\in\caF$
and that $\Theta(g)\lmk \Vbk{\rho}{(\lo,\lt)}^{(h)}\rmk 
\subset \Vbk{\Theta(g)\lmk \rho\rmk}{(\lo,\lt)}^{(ghg^{-1})}$ holds.
Applying the same argument with $\rho$, $g$, $h$ replaced by
$\Theta(g)\lmk \rho\rmk$, $g^{-1}$, $ghg^{-1}$, we obtain
{\it 2.}.
Because $\Ad R_g$ preserves $\caB_l$, the right hand side of
{\it 3.} is well-defined and
$\sigma$weak-continuous on each $\pi(\caA_{\lm l})''$ with $\lm l\in \Cl$.
With
\begin{align}
\begin{split}
\Ad R_g\circ S_{\rho}^{(l), V_{{\rho},{(\lo,\lt)}}^{(h)}}\circ \Ad R_g^*\pi
=\Ad\lmk \Theta (g)\lmk V_{{\rho},{(\lo,\lt)}}^{(h)} \rmk \rmk\Theta(g)(\rho),
\end{split}
\end{align}
the uniqueness implies {\it 3.}.
\end{proof}
We extend our definition of tensor product and $\Theta(g)$ to general endomorphisms of $\caB_l$.
For any endomorphisms $S_1, S_2$ on $\caB_l$,
we set 
\begin{align}
\begin{split}
&S_1\otimes_{\caB_l} S_2:=S_1S_2,\\
&\Mor_{\caB_l}\lmk S_1,  S_2 \rmk 
:=\left\{ Y\in \caF\mid 
Y S_1(x)
=S_2(x) Y,\quad
x\in \caB_l\right\},\\
&\Theta(g)(X):=R_g X R_g^*
\end{split}
\end{align}
Furthermore, for any endomorphisms $S_1, S_2, S_1', S_2'$ on $\caB_l$
and $X\in \Mor_{\caB_l}\lmk S_1,  S_1' \rmk $, 
$Y\in \Mor_{\caB_l}\lmk S_2,  S_2' \rmk $,
\begin{align}
\begin{split}
X\otimes_{\caB_l} Y:=X S_1(Y).
\end{split}
\end{align}
With this notation, we have the following.
\begin{lem}Consider setting in the subsection \ref{sec:qss}.
For any $g,h\in G$ and $\rho,\sigma\in  O^{(h)}_{(\loz,\ltz)}$
, we have the following.
\begin{enumerate}
\item For any $X\in \Mor_G(\rho,\sigma)$,
 $\Theta(g)(X)\in \Mor_G\lmk \Theta(g)(\rho),\Theta(g)(\sigma))\rmk$.
\item For any $V_{{\rho},{(\lo,\lt)}}^{(h)}\in \Vbk{\rho}{(\lo,\lt)}^{(h)}$ ,
$V_{{\sigma},{(\lo,\lt)}}^{(h)}\in \Vbk{\sigma}{(\lo,\lt)}^{(h)}$,
$Y\in \Mor_{\caB_l}\lmk S_{\rho}^{(l), V_{{\rho},{(\lo,\lt)}}^{(h)}},  S_{\sigma}^{(l), V_{{\sigma},{(\lo,\lt)}}^{(h)}} \rmk $
we have
\begin{align}
\begin{split}
\Theta(g)(Y)\in \Mor_{\caB_l}\lmk S_{\Theta(g)(\rho)}^{(l),\Theta(g)\lmk V_{{\rho},{(\lo,\lt)}}^{(h)}\rmk},
S_{\Theta(g)(\sigma)}^{(l),\Theta(g)\lmk V_{{\sigma},{(\lo,\lt)}}^{(h)}\rmk}\rmk.
\end{split}
\end{align}
\end{enumerate}
\end{lem}
\begin{proof}
Follows immediately from the definition and Lemma \ref{tokage}.
\end{proof}

\subsection{Braiding of $\caE_G$}\label{gato}
We introduce a braiding on our $\caA_G$.
The proof is the same as that in \cite{BA}, and  will be given in Appendix \ref{brprf}.

 
\begin{lem}\label{lem38}Consider setting in the subsection \ref{sec:qss}.
Assume Assumption \ref{a1l} and Assumption \ref{a1r}.
Let $g,h\in G$ and $\rho\in O^{(g)}_{(\loz,\ltz)}$, $\sigma\in O^{(h)}_{(\loz,\ltz)}$.
Let $(\lm {1\rho},\lm {2\rho}), 
(\lm {1\sigma},\lm {2\sigma})\in PC$ such that
$\{(\lm {1\rho},\lm {2\rho})\}
\leftarrow 
\{
(\lm {1\sigma},\lm {2\sigma})
\}$.
We set
\begin{align}
\begin{split}
&\lm{ i\rho}(t):=\lm {i\rho}-t\bm e_0,\quad 
\lm {i\sigma}(s):=\lm {i\sigma}+s\bm e_0,
\end{split}
\end{align}
with $i=1,2$, $t,s\ge 0$.
Let
\begin{align}
\begin{split}
\Vrl\rho{\lm {1\rho}(t)}{\lm {2\rho}(t)}g\in \Vbu\rho{\lm {1\rho}(t)}{\lm {2\rho}(t)}g,\quad
\Vrl\sigma{\lm {1\sigma}(s)}{\lm {2\sigma}(s)}h\in \Vbu\sigma{\lm {1\sigma}(s)}{\lm {2\sigma}(s)}h
\end{split}
\end{align}
for $t,s\ge 0$.
Then the limit 
\begin{align}\label{fune}
\begin{split}
&\epsilon_G\lmk \rho,\sigma\rmk\\
&:=\lim_{t,s\to\infty}
\lmk
\Vrl\sigma{\lm {1\sigma}(s)}{\lm {2\sigma}(s)}h\otimes_{\caB_l}
\Theta(h^{-1})\lmk \Vrl\rho{\lm {1\rho}(t)}{\lm {2\rho}(t)}g\rmk\rmk^*
\lmk
\Vrl\rho{\lm {1\rho}(t)}{\lm {2\rho}(t)}g
\otimes_{\caB_l}
\Vrl\sigma{\lm {1\sigma}(s)}{\lm {2\sigma}(s)}h
\rmk\\
&\in \caU(\caF)
\end{split}
\end{align}
exits and it is independent of the choices of
$(\lm {1\rho},\lm {2\rho})$, 
$(\lm {1\sigma},\lm {2\sigma})$,
$\Vrl\rho{\lm {1\rho}(t)}{\lm {2\rho}(t)}g$, $\Vrl\sigma{\lm {1\sigma}(s)}{\lm {2\sigma}(s)}h$.
Here the tensor product is taken for
\begin{align}
\begin{split}
\Vrl\rho{\lm {1\rho}(t)}{\lm {2\rho}(t)}g\in
\Mor_{\caB_l}\lmk
S_\rho^{(l)\unit}, S_{\rho}^{(l)\Vrl\rho{\lm {1\rho}(t)}{\lm {2\rho}(t)}g}
\rmk,\quad
\Vrl\sigma{\lm {1\sigma}(s)}{\lm {2\sigma}(s)}h\in
\Mor_{\caB_l}\lmk
S_\sigma^{(l)\unit}, S_{\rho}^{(l), \Vrl\sigma{\lm {1\sigma}(s)}{\lm {2\sigma}(s)}h}
\rmk
\end{split}
\end{align}

\end{lem}
It can be rewritten as follows.
\begin{lem}\label{tako}Consider setting in the subsection \ref{sec:qss}.
Assume Assumption \ref{a1l} and Assumption \ref{a1r}.
Let $g,h\in G$ and $\rho\in O^{(g)}_{(\loz,\ltz)}$, $\sigma\in O^{(h)}_{(\loz,\ltz)}$.
Let $(\lm{1\rho},\lm{2\rho}), (\lm{1\rho},\lm{2\rho})\in PC$.
Then we have the following.
\begin{description}
\item[(i)]
If $\{(\loz,\ltz)\}\leftarrow \{ (\lm {1\sigma},\lm {2\sigma})\}$,
then
\begin{align}
\begin{split}
\epsilon_G(\rho,\sigma)=
\lim_{s\to\infty} 
\lmk
\Vrl\sigma{\lm {1\sigma}(s)}{\lm {2\sigma}(s)}h
\rmk^*
S_\rho^{(l) \unit}
\lmk
\Vrl\sigma{\lm {1\sigma}(s)}{\lm {2\sigma}(s)}h
\rmk
\end{split}
\end{align}
for $\lm {i\sigma}(s):=\lm {i\sigma}+s\bm e_0$, $i=1,2$, $s\ge 0$ and
any 
$\Vrl\sigma{\lm {1\sigma}(s)}{\lm {2\sigma}(s)}h\in \Vbu\sigma{\lm {1\sigma}(s)}{\lm {2\sigma}(s)}h$.
\item[(ii)]
If $\{(\lm {1\rho},\lm {2\rho})\}\leftarrow \{  (\loz,\ltz)\}$,
then
\begin{align}
\begin{split}
\epsilon_G(\rho,\sigma)=
\lim_{t\to\infty} 
S_\sigma^{(l) \unit}\lmk\Theta(h^{-1})\lmk
\lmk
\Vrl\rho{\lm {1\rho}(t)}{\lm {2\rho}(t)}g
\rmk^*\rmk
\rmk
\Vrl\rho{\lm {1\rho}(t)}{\lm {2\rho}(t)}g
\end{split}
\end{align}
for $\lm {i\sigma}(s):=\lm {i\sigma}+s\bm e_0$, $i=1,2$, $s\ge 0$ and
any 
$\Vrl\rho{\lm {1\rho}(t)}{\lm {2\rho}(t)}g\in \Vbu\rho{\lm {1\rho}(t)}{\lm {2\rho}(t)}g$.

\end{description}
\end{lem}
It satisfies the axioms of braidings:
\begin{lem}\label{intw}Consider setting in the subsection \ref{sec:qss}.
Assume Assumption \ref{a1l} and Assumption \ref{a1r}.
Let $g,h\in G$ and $\rho\in O^{(g)}_{(\loz,\ltz)}$, $\sigma\in O^{(h)}_{(\loz,\ltz)}$.
Then we have 
\begin{align}
\begin{split}
\epsilon_G(\rho,\sigma)\in \Mor_G\lmk \rho\otimes \sigma,
\sigma\otimes\Theta(h^{-1})(\rho)
\rmk.
\end{split}
\end{align}
\end{lem}

\begin{lem}\label{hex}Consider setting in the subsection \ref{sec:qss}.
Assume Assumption \ref{a1l} and Assumption \ref{a1r}.
For any
$g,h,k\in G$ and $\rho\in O^{(g)}_{(\loz,\ltz)}$, $\sigma\in O^{(h)}_{(\loz,\ltz)}$,
$\gamma\in O^{(k)}_{(\loz,\ltz)}$, we have
\begin{align}
\begin{split}
\epsilon_G\lmk\rho\otimes \sigma,\gamma\rmk
=\lmk
\epsilon_G(\rho,\gamma)\otimes  \id_{\Theta(k^{-1})(\sigma)}
\rmk
\lmk
\id_\rho\otimes\epsilon_G( \sigma,\gamma)
\rmk.
\end{split}
\end{align}
\end{lem}
\begin{lem}Consider setting in the subsection \ref{sec:qss}.
Assume Assumption \ref{a1l} and Assumption \ref{a1r}.
For any
$g,h,k\in G$ and $\rho\in O^{(g)}_{(\loz,\ltz)}$, $\sigma\in O^{(h)}_{(\loz,\ltz)}$,
$\gamma\in O^{(k)}_{(\loz,\ltz)}$, we have
\begin{align}
\begin{split}
\epsilon_G\lmk\rho, \sigma\otimes \gamma\rmk
=\lmk
\id_{\sigma} \otimes \epsilon_G\lmk \Theta(h^{-1})(\rho),\gamma\rmk
\rmk
\lmk
\epsilon_G( \rho, \sigma) \otimes \id_\gamma
\rmk.
\end{split}
\end{align}
\end{lem}

\begin{lem}\label{egn}Consider setting in the subsection \ref{sec:qss}.
Assume Assumption \ref{a1l} and Assumption \ref{a1r}.
For any
$g,h\in G$ and $\rho,\rho'\in O^{(g)}_{(\loz,\ltz)}$, $\sigma,\sigma'\in O^{(h)}_{(\loz,\ltz)}$,
and
$X\in \Mor_G(\rho,\rho')$, $Y\in \Mor_G(\sigma,\sigma')$, we have
\begin{align}
\begin{split}
\epsilon_G(\rho',\sigma')\lmk X\otimes Y\rmk=
\lmk Y\otimes\Theta(h^{-1})(X)\rmk\epsilon_G(\rho,\sigma)
\end{split}
\end{align}
\end{lem}
We further have the following:
\begin{lem}\label{tentomushi}Consider setting in the subsection \ref{sec:qss}.
Assume Assumption \ref{a1l} and Assumption \ref{a1r}.
For $\rho\in O^{(g)}_{(\loz,\ltz)}$ and $\sigma\in O^{(h)}_{(\loz,\ltz)}$,
we have
\begin{align}
\begin{split}
\Theta(k)\lmk\epsilon_G(\rho,\sigma)\rmk
=\epsilon_G(\Theta(k)(\rho),\Theta(k)(\sigma))
\end{split}
\end{align}
\end{lem}

\section{Fractionalization}\label{frac}
Consider setting in the subsection \ref{sec:qss}.
In this section we show Theorem \ref{mainthm}

\subsection{Subgroup $H$ of $G$ associated with  anyons}
Consider setting in the subsection \ref{sec:qss}.
We consider the following subgroup associated with anyons.
\begin{defn}\label{ag}Consider setting in the subsection \ref{sec:qss}.
We denote by $H$ the set of all
$g\in G$ which allows the existence of irreducible
$\sigma_g\in O^{(g)}_{(\loz,\ltz)}$
and some $\sigma_g^{-1}\in  O^{(g^{-1})}_{(\loz,\ltz)}$
satisfying
$\sigma_g^{-1}\otimes\sigma_g=\pi$.
\end{defn}
\begin{lem}\label{hnor}Consider setting in the subsection \ref{sec:qss}.
Assume Assumption \ref{a1} Assumption \ref{a1r} and Assumption \ref{a1l}.
The set $H$ is a normal subgroup of $G$.
\end{lem}
\begin{proof}
Clearly, the identity $e$ of $G$ belongs to $H$ with
$\sigma_e=\sigma_e^{-1}=\pi$.\\
Next we show if $g,h\in H$, then $gh\in H$.
If $g, h\in H$, then
$\sigma_g\otimes \sigma_h\in O^{(gh)}_{(\loz,\ltz)}$
and $\sigma_h^{-1}\otimes\sigma_g^{-1}\in  O^{((gh)^{-1})}_{(\loz,\ltz)}$.
Furthremore, we have
\begin{align}
\begin{split}
&\lmk \sigma_h^{-1}\otimes \sigma_g^{-1}\rmk \otimes\lmk \sigma_g\otimes\sigma_h\rmk
=\lmk \sigma_h^{-1}\otimes \lmk \lmk  \sigma_g^{-1}\rmk \otimes\lmk \sigma_g\rmk\rmk\rmk\otimes\sigma_h\\
&=\lmk \sigma_h^{-1}\otimes \pi \rmk\otimes\sigma_h
=\pi,
\end{split}
\end{align}
because the associators are identities.
Next we note that $\sigma_g\otimes\sigma_h$ is irreducible:
Let $X\in \Mor\lmk \sigma_g\otimes\sigma_h, \sigma_g\otimes\sigma_h\rmk$.
Then we have
\begin{align}
\begin{split}
XT_{\sigma_g}^\unit T_{\sigma_h}^\unit\pi (A)=X \sigma_g\otimes\sigma_h(A) = \sigma_g\otimes\sigma_h(A)\cdot X
=T_{\sigma_g}^\unit T_{\sigma_h}^\unit\pi (A)\cdot X,\quad A\in \caA.
\end{split}
\end{align}
Acting on this by $T_{\sigma_g^{-1}}^\unit$, we obtain
\begin{align}
T_{\sigma_g^{-1}}^\unit(X)\cdot\sigma_h(A)
=T_{\sigma_g^{-1}}^\unit(X)\cdot T_{\sigma_h}^\unit\pi (A)
=T_{\sigma_h}^\unit\pi (A)\cdot T_{\sigma_g^{-1}}^\unit(X)
=\sigma_h(A)\cdot T_{\sigma_g^{-1}}^\unit(X),\quad A\in \caA.
\end{align}
Because $\sigma_h$ is irreducible, 
this means $T_{\sigma_g^{-1}}^\unit(X)\in\bbC$.
Hence we have
\begin{align}
\begin{split}
\Ad\epsilon\lmk\sigma_{g}^{-1}, \sigma_g\rmk(X)
=\Ad\epsilon\lmk\sigma_{g}^{-1}, \sigma_g\rmk
T_{\sigma_g^{-1}}^\unit T_{\sigma_g}^\unit (X)
=T_{\sigma_g}^\unit T_{\sigma_g^{-1}}^\unit (X)\in\bbC
\end{split}
\end{align}
This means $X\in \bbC$, and $\sigma_g\otimes\sigma_h$ is irreducible.

Next we show $g^{-1}\in H$ if $g\in H$.
Let $g\in H$.
First note that $\sigma_g^{-1}\in  O^{(g^{-1})}_{(\loz,\ltz)}$ is irreducible.
In fact if $X\in\Mor(\sigma_g^{-1},\sigma_g^{-1})$
then 
\begin{align}
\begin{split}
X\pi(A)=X\sigma_g^{-1}\otimes \sigma_g(A)
=X T_{\sigma_g^{-1}}^\unit T_{\sigma_g}^\unit \pi(A)
=T_{\sigma_g^{-1}}^\unit T_{\sigma_g}^\unit \pi(A) X
=\sigma_g^{-1}\otimes\sigma_g(A) X=\pi(A) X,\quad
A\in\caA.
\end{split}
\end{align}
Hence we have $X\in\bbC$ and $\sigma_g^{-1}$ is irreducible.
Next note that 
\begin{align}
\begin{split}
\lmk \Ad\lmk \epsilon_G\lmk  \Theta(g^{-1})(\sigma_g), \sigma_g^{-1}\rmk\rmk \Theta(g^{-1})(\sigma_g)\rmk\otimes\sigma_g^{-1}
=
\lmk \sigma_g^{-1}\otimes\sigma_g\rmk
=\pi.
\end{split}
\end{align}
Because $\epsilon_G\lmk  \Theta(g^{-1})(\sigma_g), \sigma_g^{-1}\rmk\in \caU(\caF)$, we have
$\Ad\lmk \epsilon_G\lmk  \Theta(g^{-1})(\sigma_g), \sigma_g^{-1}\rmk\rmk \Theta(g^{-1})(\sigma_g)\in O^{(g)}$.
Since $\sigma_g\in O^{(g)}_{(\loz,\ltz)}$
and $\sigma_g^{-1}\in  O^{(g^{-1})}_{(\loz,\ltz)}$,
we have 
\begin{align}
\epsilon_G\lmk \Theta(g^{-1})(\sigma_g), \sigma_g^{-1}\rmk
\in \pi\lmk\caA_{\loz\cup\ltz}\rmk'.
\end{align}
Therefore, 
we have
\begin{align}
\begin{split}
\left. \Ad\lmk \epsilon_G\lmk  \Theta(g^{-1})(\sigma_g), \sigma_g^{-1}\rmk\rmk \Theta(g^{-1})(\sigma_g)\right\vert_{\caA_{\loz\cup\ltz}}
=\left.\pi\right\vert_{\caA_{\loz\cup\ltz}}\beta_g^{\loz},
\end{split}
\end{align}
and we have 
$ \Ad\lmk \epsilon_G\lmk  \Theta(g^{-1})(\sigma_g), \sigma_g^{-1}\rmk\rmk \Theta(g^{-1})(\sigma_g)\in O^{(g)}_{(\loz,\ltz)}$.
Hence we get $g^{-1}\in H$.
This completes the proof that $H$ is a group.

Next let $g\in H$ and $h\in G$.
We show that $hgh^{-1}\in H$.
From Lemma \ref{tokage}, we have
\begin{align}
\begin{split}
\Theta(h)\lmk \sigma_g\rmk \in O^{(hgh^{-1})}_{(\loz,\ltz)},\quad
\Theta(h)\lmk\sigma_g^{-1}\rmk\in  O^{(\lmk hgh^{-1}\rmk^{-1})}_{(\loz,\ltz)}
\end{split}
\end{align}
and 
\begin{align}
\begin{split}
\Theta(h)\lmk\sigma_g^{-1}\rmk\otimes\Theta(h)\lmk \sigma_g\rmk
=\Theta(h)\lmk\sigma_g^{-1}\otimes \sigma_g\rmk
=\Theta(h)\pi=\pi.
\end{split}
\end{align}
Clearly $\Theta(h)\lmk \sigma_g\rmk $ is irreducible.
Hence we get $hgh^{-1}\in H$, proving $H$ is a normal subgroup.
\end{proof}
\begin{lem}\label{gengorou}
Consider setting in the subsection \ref{sec:qss}.
Assume Assumption \ref{a1} Assumption \ref{a1r} and Assumption \ref{a1l}.
Suppose $g\in H$ and $\sigma_g\in O^{(g)}_{(\loz,\ltz)}$
and $\sigma_g^{-1}\in  O^{(g^{-1})}_{(\loz,\ltz)}$
satisfying
$\sigma_g^{-1}\otimes\sigma_g=\pi$.
Suppose that $\sigma_g$ is irreducible.
Then
$\sigma_g\otimes\sigma_g^{-1}=\pi$.
\end{lem}
\begin{proof}
Note from the proof of Lemma \ref{hnor} that
$\sigma_g\otimes\sigma_g$ is irreducible.
Therefore, 
\begin{align}
\begin{split}
\epsilon(\sigma_g,\sigma_g)\in\Mor\lmk\sigma_g\otimes\sigma_g, \sigma_g\otimes\sigma_g\rmk
=\bbC
\end{split}
\end{align}
and we have $\epsilon(\sigma_g,\sigma_g)\in \bbT$.
We have
\begin{align}
\begin{split}
\unit=\epsilon\lmk\sigma_g^{-1}\otimes\sigma_g,\sigma_g\rmk
=\lmk
\epsilon(\sigma_g^{-1},\sigma_g)\otimes  \id_{\sigma_g}
\rmk
\lmk
\id_{\sigma_g^{-1}}\otimes\epsilon( \sigma_g,\sigma_g)
\rmk\\
=\lmk
\epsilon(\sigma_g^{-1},\sigma_g)\otimes  \id_{\sigma_g}
\rmk
T_{\sigma_g^{-1}}^\unit \lmk\epsilon( \sigma_g,\sigma_g)
\rmk.
\end{split}
\end{align}
Combining this with $\epsilon(\sigma_g,\sigma_g)\in \bbT$, we get $\epsilon(\sigma_g^{-1},\sigma_g)\in U(1)$.
From this, we have
\begin{align}
\begin{split}
\sigma_g\otimes\sigma_g^{-1}
=\Ad\lmk \epsilon_G(\sigma_g^{-1},\sigma_g)\rmk\sigma_g^{-1}\otimes\sigma_g
=\sigma_g^{-1}\otimes\sigma_g=\pi.
\end{split}
\end{align}
\end{proof}

\begin{lem}\label{hituji}
Consider setting in the subsection \ref{sec:qss}.
Assume Assumption \ref{a1} Assumption \ref{a1r} and Assumption \ref{a1l}.
For any $a\in [\Irr O_{\lz}]$ and $h\in H$, we have 
$a^{(h)}=a$, and we may take $W_a^{(h)}$ in (\ref{wa}) as
\begin{align}\label{hotaru}
\begin{split}
W_a^{(h)}=T_{\sigma_{h^{-1}}^{-1}}^{\unit}\lmk
\epsilon_G\lmk \rho_a, \sigma_{h^{-1}}\rmk\lmk \epsilon\lmk \rho_a, \sigma_{h^{-1}}\rmk\rmk^* \rmk
\end{split}
\end{align}
\end{lem}
\begin{proof}
Let $\epsilon\lmk \rho_a, \sigma_{h^{-1}}\rmk\in\Mor\lmk \rho_a\otimes \sigma_{h^{-1}},
\sigma_{h^{-1}}\otimes  \rho_a\rmk$ 
be the braiding given by Theorem \ref{MTC}.
Let $\epsilon_G\lmk \rho_a, \sigma_{h^{-1}}\rmk\in\Mor\lmk \rho_a\otimes \sigma_{h^{-1}},
\sigma_{h^{-1}}\otimes  \Theta(h)\lmk \rho_a\rmk\rmk$ 
be the braiding given by Lemma \ref{intw}.
Then we have
\begin{align}
\begin{split}
&\sigma_{h^{-1}}\otimes  \Theta(h)\lmk \rho_a\rmk\\
&=
\Ad\lmk \epsilon_G\lmk \rho_a, \sigma_{h^{-1}}\rmk\lmk \epsilon\lmk \rho_a, \sigma_{h^{-1}}\rmk\rmk^* \rmk
\lmk \sigma_{h^{-1}}\otimes  \rho_a\rmk.
\end{split}
\end{align}
Applying $\sigma_{h^{-1}}^{-1}\otimes$, we obtain
\begin{align}
\begin{split}
&\Theta(h)\lmk \rho_a\rmk\\
&=
\Ad\lmk T_{\sigma_{h^{-1}}^{-1}}^{\unit}\lmk
\epsilon_G\lmk \rho_a, \sigma_{h^{-1}}\rmk\lmk \epsilon\lmk \rho_a, \sigma_{h^{-1}}\rmk\rmk^* \rmk\rmk
\lmk \sigma_{h^{-1}}^{-1}\otimes \lmk \sigma_{h^{-1}}\otimes  \rho_a\rmk\rmk\\
&=
\Ad\lmk T_{\sigma_{h^{-1}}^{-1}}^{\unit}\lmk
\epsilon_G\lmk \rho_a, \sigma_{h^{-1}}\rmk\lmk \epsilon\lmk \rho_a, \sigma_{h^{-1}}\rmk\rmk^* \rmk\rmk
\lmk  \rho_a\rmk.
\end{split}
\end{align}
Hence we have $a^{(h)}=a$ and we may take $W_a^{(h)}$ as (\ref{hotaru}).
This proves the claim.
\end{proof}

\begin{lem}\label{kamakiri}
Consider setting in the subsection \ref{sec:qss}.
Assume Assumption \ref{a1} Assumption \ref{a1r} and Assumption \ref{a1l}.
Let $a,b, c\in [\Irr O_{\lz}]$
and $S\in \Mor\lmk\rho_a\otimes\rho_b,\rho_c\rmk$.
With the choice of $W$s in Lemma \ref{hituji}, we have
\begin{align}
\begin{split}
W_c^{(g)} S=\Theta(g)(S)
\lmk Y_{ab}^{(g)}\rmk^*,\quad g\in H.
\end{split}
\end{align}
\end{lem}
\begin{proof}
Note from Lemma 2.4 of \cite{MTC} that $S\in\caF$
hence $S\in \Mor_G(\rho_a\otimes\rho_b,\rho_c)$.
Applying naturality of $\epsilon$ to 
\begin{align}
\begin{split}
&W_c^{(g)} S
=T_{\sigma_{g^{-1}}^{-1}}^{\unit}\lmk
\epsilon_G\lmk \rho_c, \sigma_{g^{-1}}\rmk\lmk \epsilon\lmk \rho_c, \sigma_{g^{-1}}\rmk\rmk^* \rmk
S\\
&=T_{\sigma_{g^{-1}}^{-1}}^{\unit}\lmk
\epsilon_G\lmk \rho_c, \sigma_{g^{-1}}\rmk\lmk \epsilon\lmk \rho_c, \sigma_{g^{-1}}\rmk\rmk^* 
T_{\sigma_{g^{-1}}}^{\unit}(S)
\rmk,
\end{split}
\end{align}
we obtain
\begin{align}
\begin{split}
&W_c^{(g)} S
=T_{\sigma_{g^{-1}}^{-1}}^{\unit}
\lmk
\epsilon_G\lmk \rho_c, \sigma_{g^{-1}}\rmk
S
\lmk \epsilon\lmk \rho_a\otimes\rho_b, \sigma_{g^{-1}}\rmk\rmk^* \rmk.
\end{split}
\end{align}
Applying naturality of $\epsilon_G$ Lemma \ref{egn} to this, we obtain
\begin{align}
\begin{split}
&W_c^{(g)} S
=T_{\sigma_{g^{-1}}^{-1}}^{\unit}
\lmk
T_{\sigma_{g^{-1}}}^{\unit}\lmk \Theta(g)(S)\rmk 
\epsilon_G\lmk \rho_a\otimes \rho_b, \sigma_{g^{-1}}\rmk
\lmk \epsilon\lmk \rho_a\otimes\rho_b, \sigma_{g^{-1}}\rmk\rmk^* \rmk.
\end{split}
\end{align}
Substituting Lemma \ref{hex} of $\epsilon_G$ and analogous relation for
$\epsilon$, we obtain
\begin{align}
\begin{split}
&W_c^{(g)} S
=\Theta(g)(S) T_{\sigma_{g^{-1}}^{-1}}^{\unit}\lmk
\lmk
\epsilon_G\lmk \rho_a, \sigma_{g^{-1}}\rmk
T_{\rho_a}^\unit\lmk \epsilon_G\lmk \rho_b,\sigma_{g^{-1}}\rmk\rmk
\rmk
\lmk
\epsilon\lmk \rho_a, \sigma_{g^{-1}}\rmk
T_{\rho_a}^\unit\epsilon \lmk \rho_b,\sigma_{g^{-1}}\rmk
\rmk^*\rmk\\
&=\Theta(g)(S)\cdot  T_{\sigma_{g^{-1}}^{-1}}^{\unit}\lmk
\epsilon_G\lmk \rho_a, \sigma_{g^{-1}}\rmk\rmk
T_{\sigma_{g^{-1}}^{-1}}^{\unit}T_{\rho_a}^\unit\lmk 
\epsilon_G\lmk \rho_b,\sigma_{g^{-1}}\rmk
\epsilon \lmk \rho_b,\sigma_{g^{-1}}\rmk^*\rmk
T_{\sigma_{g^{-1}}^{-1}}^{\unit}
\lmk
\epsilon\lmk \rho_a, \sigma_{g^{-1}}\rmk^*
\rmk.
\end{split}
\end{align}
Now from Lemma \ref{gengorou}, we have
\begin{align}
\begin{split}
\epsilon_G\lmk \rho_b,\sigma_{g^{-1}}\rmk
\epsilon \lmk \rho_b,\sigma_{g^{-1}}\rmk^*
=T_{\sigma_{g^{-1}}}^{\unit}T_{\sigma_{g^{-1}}^{-1}}^{\unit}
\lmk
\epsilon_G\lmk \rho_b,\sigma_{g^{-1}}\rmk
\epsilon \lmk \rho_b,\sigma_{g^{-1}}\rmk^*
\rmk
=T_{\sigma_{g^{-1}}}^{\unit}\lmk W_b^{(g)}\rmk.
\end{split}
\end{align}
Substituting this, we obtain
\begin{align}
\begin{split}
&W_c^{(g)} S
=\Theta(g)(S)\cdot  T_{\sigma_{g^{-1}}^{-1}}^{\unit}\lmk
\epsilon_G\lmk \rho_a, \sigma_{g^{-1}}\rmk\rmk
T_{\sigma_{g^{-1}}^{-1}}^{\unit}T_{\rho_a}^\unit
\lmk 
T_{\sigma_{g^{-1}}}^{\unit}\lmk W_b^{(g)}\rmk
\rmk
T_{\sigma_{g^{-1}}^{-1}}^{\unit}
\lmk
\epsilon\lmk \rho_a, \sigma_{g^{-1}}\rmk^*
\rmk\\
&=
\Theta(g)(S)\cdot  T_{\sigma_{g^{-1}}^{-1}}^{\unit}\lmk
\epsilon_G\lmk \rho_a, \sigma_{g^{-1}}\rmk\rmk
T_{\sigma_{g^{-1}}^{-1}}^{\unit}
\lmk \Ad\lmk \epsilon\lmk \rho_a, \sigma_{g^{-1}}\rmk^*\rmk
T_{\sigma_{g^{-1}}}^{\unit}
T_{\rho_a}^\unit
 \lmk W_b^{(g)}\rmk\rmk
T_{\sigma_{g^{-1}}^{-1}}^{\unit}
\lmk
\epsilon\lmk \rho_a, \sigma_{g^{-1}}\rmk^*
\rmk\\
&=
\Theta(g)(S)\cdot  T_{\sigma_{g^{-1}}^{-1}}^{\unit}\lmk
\lmk \epsilon_G\lmk \rho_a, \sigma_{g^{-1}}\rmk\rmk\rmk
T_{\sigma_{g^{-1}}^{-1}}^{\unit}
\lmk \lmk \epsilon\lmk \rho_a, \sigma_{g^{-1}}\rmk^*\rmk\rmk
\cdot 
T_{\sigma_{g^{-1}}^{-1}}^{\unit}
T_{\sigma_{g^{-1}}}^{\unit}
T_{\rho_a}^\unit
 \lmk W_b^{(g)}\rmk\\
 &=
\Theta(g)(S)\cdot 
 T_{\sigma_{g^{-1}}^{-1}}^{\unit}\lmk
\epsilon_G\lmk \rho_a, \sigma_{g^{-1}}\rmk
\lmk \epsilon\lmk \rho_a, \sigma_{g^{-1}}\rmk^*\rmk
\rmk
\cdot 
T_{\sigma_{g^{-1}}^{-1}}^{\unit}
T_{\sigma_{g^{-1}}}^{\unit}
T_{\rho_a}^\unit
 \lmk W_b^{(g)}\rmk\\
   &=
\Theta(g)(S)\cdot 
 T_{\sigma_{g^{-1}}^{-1}}^{\unit}\lmk
\epsilon_G\lmk \rho_a, \sigma_{g^{-1}}\rmk
\lmk \epsilon\lmk \rho_a, \sigma_{g^{-1}}\rmk^*\rmk
\rmk
\cdot 
T_{\rho_a}^\unit
 \lmk W_b^{(g)}\rmk\\
 &= \Theta(g)(S)\cdot W_a^{(g)}T_{\rho_a}^\unit
 \lmk W_b^{(g)}\rmk
 =\Theta(g)(S)
\lmk Y_{ab}^{(g)}\rmk^*.
 \end{split}
\end{align}

\end{proof}
\begin{thm}\label{arashiyama}
Consider setting in the subsection \ref{sec:qss}.
Assume Assumption \ref{a1} Assumption \ref{a1r} and Assumption \ref{a1l}.
For each $a,b, c\in [\Irr O_{\lz}]$ with $\Mor(\rho_a\otimes\rho_b,\rho_c)\neq 0$,
with the choice of $W$s in Lemma \ref{hituji},
we have 
\begin{align}
\begin{split}
\frac{\omega^{(a)}(g,h) \omega^{(b)}(g,h)}{\omega^{(c)} (g,h)}=1,\quad g,h\in H.
\end{split}
\end{align}

\end{thm}
\begin{proof}
Consider $\sigma_{h^{-1}}$ associated to $h\in H$.
Let $S\in \Mor\lmk\rho_a\otimes\rho_b,\rho_c\rmk$ be non-zero.
With the choice of $W$s in Lemma \ref{hituji}, we have
\begin{align}
\begin{split}
&\Theta(g)\lmk W_c^{(h)}\rmk W_c^{(g)} S
=\Theta(g)\lmk W_c^{(h)}\rmk \Theta(g)(S)
\lmk Y_{ab}^{(g)}\rmk^*
=\Theta(g)\lmk W_c^{(h)}S \rmk 
\lmk Y_{ab}^{(g)}\rmk^*\\
&=\Theta(g)\lmk \Theta(h)(S)
\lmk Y_{ab}^{(h)}\rmk^*\rmk \lmk Y_{ab}^{(g)}\rmk^*
=\Theta(gh)\lmk S\rmk \Theta(g)\lmk \lmk Y_{ab}^{(h)}\rmk^*\rmk \lmk Y_{ab}^{(g)}\rmk^*
\end{split}
\end{align}
using Lemma \ref{kamakiri}.
By Lemma \ref{washi}, Lemma \ref{kamakiri}
we have
\begin{align}
\begin{split}
&\Theta(g)\lmk W_c^{(h)}\rmk W_c^{(g)} S
=\Theta(gh)\lmk S\rmk \Theta(g)\lmk \lmk Y_{ab}^{(h)}\rmk^*\rmk \lmk Y_{ab}^{(g)}\rmk^*\\
&=\overline{\omega^{(a)}(g,h)\omega^{(b)}(g,h)} \Theta(gh)\lmk S\rmk\lmk Y_{a,b}^{(gh)}\rmk^*
=\overline{\omega^{(a)}(g,h)\omega^{(b)}(g,h)} W_c^{(gh)}S,\quad g,h\in H.
\end{split}
\end{align}
Hence we have
\begin{align}
\begin{split}
S=\omega^{(c)}(g,h)\overline{\omega^{(a)}(g,h)\omega^{(b)}(g,h)} S.
\end{split}
\end{align}
Because $S$ is not zero, this proves the Lemma.
\end{proof}

\subsection{Proof of Theorem \ref{mainthm}}
From Lemma \ref{fushimi}, Lemma \ref{hituji} and Theorem \ref{arashiyama},
in order to show  Theorem \ref{mainthm}, it suffices to show that
$H=G$ under Assumption \ref{a5}.
Namely, the following Lemma completes the proof.
\begin{lem}
Consider setting in the subsection \ref{sec:qss}.
Assume Assumption \ref{a1} Assumption \ref{a1r} Assumption \ref{a1l}, and Assumption \ref{a5}.
For any $\ld\in \Cl$, $\varepsilon>0$ such that 
$\lmk \partial \ld\rmk_\varepsilon\subset \lmk \loz\cup\ltz\rmk^c$,
$\loz\subset \ld\subset \lmk \ltz\rmk^c$
we have
\begin{align}
\begin{split}
\sigma_g:=\pi\beta_g^\ld\gamma_{g,\ld,\varepsilon}^{-1}
\in O^{(g)}_{(\loz,\ltz)},\quad
\sigma_g^{-1}:=\pi\gamma_{g,\ld,\varepsilon}{\beta_g^\ld}^{-1}
\in O^{(g^{-1})}_{(\loz,\ltz)},\quad g\in G
\end{split}
\end{align}
and $\sigma_g$ is irreducible for all $g\in G$.
Here we used notation in Assumption \ref{a5}.
Furthermore, we have
\begin{align}
\begin{split}
\sigma_g^{-1}\otimes \sigma_g=\pi,\quad g\in G.
\end{split}
\end{align}
\end{lem}
\begin{proof}
Because 
\begin{align}
\sigma_g=\pi\beta_g^\ld\gamma_{g,\ld,\varepsilon}^{-1}
=\pi\beta_g^\ld\gamma_{g,\ld}^{-1}\Ad\lmk u_{g\ld\varepsilon}\rmk
=\Ad\lmk v_{g\ld } \pi\lmk   u_{g\ld\varepsilon}\rmk \rmk\pi,
\end{align}
we have $\sigma_g\in \caO$.
In fact this is a trivial anyon.

Now we see that $\sigma_g:=\pi\beta_g^\ld\gamma_{g,\ld,\varepsilon}^{-1}
\in O^{(g)}_{(\loz,\ltz)}$.
For any $(\lm 1,\lm 2)\in PC$, choose 
$\tilde \ld\in \Cl$, $\tilde \varepsilon>0$ such that 
$\lmk \partial \tilde \ld\rmk_{\tilde \varepsilon}
\subset \lmk \lo\cup\lt\rmk^c$,
$\lo\subset \tilde \ld\subset \lmk \lt\rmk^c$.
Then $\tilde \sigma_g:=\pi\beta_g^{\tilde \ld}\gamma_{g,{\tilde \ld},\varepsilon}^{-1}$
satisfies
\begin{align}\label{kurama}
\begin{split}
\tilde \sigma_g\vert_{\caA_{\lo\cup\lt}}=
\pi\beta_g^{\tilde \ld}\gamma_{g,{\tilde \ld},\varepsilon}^{-1}\vert_{\caA_{\lo\cup\lt}}
=\pi\beta_g^{\tilde \ld}\vert_{\caA_{\lo\cup\lt}}
=\pi\beta_g^{ \ld_1}\vert_{\caA_{\lo\cup\lt}}.
\end{split}
\end{align}
Note in case $(\lo,\lt)=(\loz,\ltz)$, we may choose $\tilde\ld=\ld$
and obtain $\sigma_g\vert_{\caA_{\loz\cup\ltz}}=\pi\beta_g^{\ld}\vert_{\caA_{\loz\cup\ltz}}$.
With the notation in Assumption \ref{a5}, we have
\begin{align}
\begin{split}
&\tilde \sigma_g=\pi\beta_g^{\tilde \ld}\gamma_{g,{\tilde \ld},{\tilde\varepsilon}}^{-1}
=\pi\beta_g^{\tilde \ld}\gamma_{g,{\tilde \ld}}^{-1}\Ad u_{g\tilde\ld{\tilde\varepsilon}}
=\Ad v_{g\tilde\ld}\pi\Ad u_{g\tilde\ld{\tilde\varepsilon}}
=\Ad\lmk v_{g\tilde\ld} v_{g\ld}^*\rmk
\pi\beta_g^{\ld}\gamma_{g,{\ld}}^{-1}\Ad u_{g\tilde\ld{\tilde\varepsilon}}\\
&=\Ad\lmk v_{g\tilde\ld} v_{g\ld}^*
\pi\beta_g^{\ld}\gamma_{g,{\ld}}^{-1}\lmk u_{g\tilde\ld{\tilde\varepsilon}}\rmk
\rmk
\pi\beta_g^{\ld}\gamma_{g,{\ld}}^{-1}
=\Ad\lmk v_{g\tilde\ld} v_{g\ld}^*
\pi\beta_g^{\ld}\gamma_{g,{\ld}}^{-1}\lmk u_{g\tilde\ld\tilde\varepsilon}\rmk
\rmk
\pi\beta_g^{\ld}\gamma_{g,{\ld},\varepsilon}^{-1}\Ad u_{g\ld\varepsilon}^*\\
&=\Ad\lmk
 v_{g\tilde\ld} v_{g\ld}^*
 \pi\beta_g^{\ld}\gamma_{g,{\ld}}^{-1}\lmk u_{g\tilde\ld\tilde\varepsilon}\rmk
 \pi\beta_g^{\ld}\gamma_{g,{\ld},\varepsilon}^{-1}\lmk u_{g\ld\varepsilon}^*\rmk
\rmk\pi\beta_g^{\ld}\gamma_{g,{\ld},\varepsilon}^{-1}\\
&=:\Ad \lmk V_{\sigma_g(\lo, \lt)}^{g}\rmk
\sigma_g,
\end{split}
\end{align}
where 
\begin{align}
\begin{split}
 V_{\sigma_g(\lo, \lt)}^{g}
 := v_{g\tilde\ld} v_{g\ld}^*
 \pi\lmk \beta_g^{\ld}\gamma_{g,{\ld}}^{-1}\lmk u_{g\tilde\ld\tilde\varepsilon}\rmk\rmk
 \pi\lmk \beta_g^{\ld}\gamma_{g,{\ld},\varepsilon}^{-1}\lmk u_{g\ld\varepsilon}^*\rmk\rmk\in\caU(\caF),
\end{split}
\end{align}
by the assumption.
Combining this with (\ref{kurama}) gives
$\sigma_g
\in O^{(g)}_{(\loz,\ltz)}$.

Next we see $\sigma_g^{-1}:=\pi\gamma_{g,\ld,\varepsilon}{\beta_g^\ld}^{-1}
\in O^{(g^{-1})}_{(\loz,\ltz)}$.
Note that
\begin{align}\label{kiyomizu}
\begin{split}
\sigma_g^{-1}=\pi\gamma_{g,\ld,\varepsilon}{\beta_g^\ld}^{-1}=\Ad 
\pi\lmk u_{g\ld\varepsilon}^*\rmk v_{g\ld}^* \pi,
\end{split}
\end{align}
hence $\sigma_g^{-1}\in \caO$.

For any $(\lm 1,\lm 2)\in PC$, we have
\begin{align}
\begin{split}
&\pi\beta_g^{\lo}\vert_{\caA_{\lo\cup\lt}}
=\Ad V_{\sigma_g (\lo,\lt)}^{(g)}\sigma_g\vert_{\caA_{\lo\cup\lt}}
=\Ad V_{\sigma_g (\lo,\lt)}^{(g)}\pi\beta_g^\ld \gamma_{g\ld\varepsilon}^{-1}\vert_{\caA_{\lo\cup\lt}}\\
&=\Ad V_{\sigma_g (\lo,\lt)}^{(g)}\pi\beta_g^\ld \gamma_{g\ld}^{-1}\Ad u_{g\ld\varepsilon}\vert_{\caA_{\lo\cup\lt}}\\
&=\Ad \lmk V_{\sigma_g (\lo,\lt)}^{(g)} 
\pi\beta_g^\ld \gamma_{g\ld}^{-1}\lmk u_{g\ld\varepsilon}\rmk
\rmk\pi\beta_g^\ld \gamma_{g\ld}^{-1}\vert_{\caA_{\lo\cup\lt}}\\
&=\Ad\lmk
V_{\sigma_g (\lo,\lt)}^{(g)} 
\pi\beta_g^\ld \gamma_{g\ld}^{-1}\lmk u_{g\ld\varepsilon}\rmk v_{g\ld}
\rmk\pi \vert_{\caA_{\lo\cup\lt}}.
\end{split}
\end{align}
Hence from (\ref{kiyomizu}) we have
\begin{align}
\begin{split}
&\sigma_g^{-1}\vert_{\caA_{\lo\cup\lt}}
=\Ad \pi\lmk u_{g\ld\varepsilon}^*\rmk  v_{g\ld}^*\pi\vert_{\caA_{\lo\cup\lt}}
=\Ad\lmk  \pi\lmk u_{g\ld\varepsilon}^*\rmk v_{g\ld}^*
V_{\sigma_g (\lo,\lt)}^{(g)} 
\pi\beta_g^\ld \gamma_{g\ld}^{-1}\lmk u_{g\ld\varepsilon}\rmk v_{g\ld}
\rmk\pi {\beta_{g^{-1}}^{\lo}}\vert_{\caA_{\lo\cup\lt}}.
\end{split}
\end{align}
Because
\begin{align}
\begin{split}
 \pi\lmk u_{g\ld\varepsilon}^*\rmk v_{g\ld}^*
V_{\sigma_g (\lo,\lt)}^{(g)} 
\pi\beta_g^\ld \gamma_{g\ld}^{-1}\lmk u_{g\ld\varepsilon}\rmk v_{g\ld}\in \caU(\caF)
\end{split}
\end{align}
by the assumption,
we get
$\sigma_g^{-1}\in O^{(g^{-1})}_{(\loz,\ltz)}$.
Clearly, 
$\sigma_g$ is irreducible.
Furthermore, we have
\begin{align}
\begin{split}
\sigma_g^{-1}\otimes \sigma_g=
T_{\sigma_g^{-1}}^\unit T_{\sigma_g}^\unit \pi
=T_{\sigma_g^{-1}}^\unit \sigma_g
=T_{\sigma_g^{-1}}^\unit\pi\beta_g^\ld\gamma_{g,\ld,\varepsilon}^{-1}
=\pi \gamma_{g,\ld,\varepsilon}{\beta_g^\ld}^{-1}
\beta_g^\ld\gamma_{g,\ld,\varepsilon}^{-1}
=\pi.
\end{split}
\end{align}
\end{proof}
\section{A sufficient condition of Assumption \ref{a5}}\label{suf}
In order to think of examples, it is convinient to have the following sufficient condition of Assumption \ref{a5}.
\begin{assum}\label{a7}
\begin{enumerate}
\item There exits $l\in\bbN$ satisfying the following :
for any cone $\ld$, $g\in G$,
and $N\in\bbN$, there exists a unitary 
$W_{ \partial\ld,N}^g\in\caA_{\lmk \partial \lmk \ld\cap[-N,N]^2\rmk\rmk^{(l)}}$
such that 
\begin{align}
\begin{split}
\pi\lmk \bigotimes_{x\in \ld\cap[-N,N]^2} U_g\rmk^*\Omega
=\pi\lmk W_{\partial \ld,N}^g\rmk^*\Omega.
\end{split}
\end{align}
\item
For any cones $\ld\subset \ld'$, $g\in G$, $N\in\bbN$,
we have
\begin{align}
\begin{split}
\lmk W_{\partial \ld,N}^g\rmk^* W_{\partial \ld',N}^g, \quad \lmk W_{\partial \ld,N}^g\rmk
 \caA_{\ld'\setminus \ld\cap [-N,N]^2}.
\lmk  W_{\partial \ld',N}^g\rmk^*
\in \caA_{
\lmk \lmk
\ld'\setminus \ld\cap [-N,N]^2
\rmk\rmk^{(l)}
}.
\end{split}
\end{align}
\item
For any cone $\ld$, $g\in G$, and $A\in\caA_{\rm loc}$,
\begin{align}
\begin{split}
\Ad\lmk W_{\partial \ld,N}^g \rmk(A)
=\Ad \lmk W_{\partial \ld,M}^g \rmk (A),\quad
\Ad \lmk {W_{\partial \ld,N}^g}^*\rmk (A)=\Ad \lmk{W_{\partial \ld,M}^g}^* \rmk(A)
\end{split}
\end{align}
for $M, N$ large enough.
\item 
For any $\ld$ cone, $0<\varepsilon<\frac 12\min \{|\arg\ld|, 2\pi-|\arg\ld|\}$,
and $g\in G$,
there exists a unitary $u_{g\ld\varepsilon}\in \caU(\caA)$
such that $u_{g\ld\varepsilon}^*W_{\partial\ld, N}^g\in \caA_{(\partial\ld)_\varepsilon\cup
\lmk\partial [-N,N]^2\rmk^{(l)}}$, for all $N\in\bbN$.
\end{enumerate}
\end{assum}
\begin{lem}\label{amayadori}
Consider the setting in subsection \ref{sec:qss}.
Then Assumption \ref{a7} implies
Assumption \ref{a5}.
\end{lem}
\begin{proof}
First we prove (i) of Assumption \ref{a5}.
Let  $\ld$ be a cone, $0<\varepsilon<\frac 12\min \{|\arg\ld|, 2\pi-|\arg\ld|\}$,
and $g\in G$.
By {\it 3.} of Assumption \ref{a7}, we get endomorphisms of $\caA$ by
\begin{align}
\begin{split}
&\gamma_{g\ld}(A)=\lim_{N}\Ad\lmk W_{\partial \ld,N}^g \rmk(A),\\
&\tilde \gamma_{g\ld}(A)=\lim_{N}\Ad\lmk W_{\partial \ld,N}^g \rmk^*(A),\quad A\in \caA.
\end{split}
\end{align}
Because $\tilde \gamma_{g\ld}=\lmk \gamma_{g\ld}\rmk^{-1}$,
$\gamma_{g\ld}$ is an automorphism.

We claim that $\gamma_{g\ld}$ is localized in 
$\caA_{\lmk\partial \ld \rmk^{(l)}}$.
To see this, let $A\in \caA_{\lmk\partial \ld \rmk^{(l)}}$ be an arbitrary local element.
Because $A$ is local, there exists an $N_0\in\bbN$ such that
$A\in \caA_{[-N_0,N_0]^2}$ and 
$\gamma_{g\ld}(A)=\Ad\lmk W_{\partial \ld,N}^g \rmk(A)$,
for all $N\ge N_0$.
Then we have
\begin{align}
\begin{split}
&\gamma_{g\ld}(A)
\in \cap_{N\ge N_0} 
\Ad\lmk W_{\partial \ld,N}^g \rmk\lmk
\caA_{\lmk\partial \ld \rmk^{(l)}\cap [-N_0,N_0]^2}
\rmk\\
&\subset
 \cap_{N\ge N_0} 
\lmk
\caA_{\lmk \lmk\partial \ld \rmk^{(l)}\cap [-N_0,N_0]^2\rmk\cup 
\lmk \partial \lmk \ld\cap[-N,N]^2\rmk\rmk^{(l)}}
\rmk\subset \caA_{\lmk\partial \ld \rmk^{(l)}}.
\end{split}
\end{align}
Hence we have $\gamma_{g\ld}\lmk \caA_{\lmk\partial \ld \rmk^{(l)}}\rmk \subset
\caA_{\lmk\partial \ld \rmk^{(l)}}$.
Similarly, we have 
$\gamma_{g\ld}^{-1}\lmk \caA_{\lmk\partial \ld \rmk^{(l)}}\rmk \subset
\caA_{\lmk\partial \ld \rmk^{(l)}}$, and we get
$\gamma_{g\ld}\lmk \caA_{\lmk\partial \ld \rmk^{(l)}}\rmk =
\caA_{\lmk\partial \ld \rmk^{(l)}}$.
On the other hand, for all local
$A\in \caA_{\lmk \lmk\partial \ld\rmk^{(l)}\rmk^c}$,
we have
$\gamma_{g\ld}(A)=\Ad\lmk W_{\partial \ld,N}^g \rmk(A)=A$
for $N$ large enough.
Hence $\gamma_{g\ld}$ is localized in $\lmk\partial \ld\rmk^{(l)}$.
Similarly, by {\it 4.}$, \gamma_{g\ld\varepsilon}:=\Ad \lmk u_{g,\ld,\varepsilon}^*\rmk \gamma_{g\ld }\in 
\Aut\caA_{\lmk \partial\ld\rmk_\varepsilon}$.
For any local $A\in \caA_{\rm loc}$, we have
\begin{align}
\begin{split}
\omega\beta_g^\ld(A)
=\omega\Ad\lmk \bigotimes_{x\in \ld\cap[-N,N]^2} U_g\rmk(A)
=\omega \Ad\lmk W_{\partial \ld,N}^g \rmk(A)
=\omega\gamma_{g\ld}(A),
\end{split}
\end{align}
for $N$ large enough.
This completes the proof of (i).

Next we show (ii) of Assumption \ref{a5}.
First we show that for any cone $\ld=\ld_{\bm a, \pi,\varphi}\in\Cl$,
\begin{align}\label{takos}
\begin{split}
v_{g\ld}=s*-\lim\pi
\lmk
\lmk \bigotimes_{\ld\cap[-N,N]^2} U_g\rmk  \lmk W_{\partial \ld,N}^g\rmk^*
\rmk
\end{split}
\end{align}
In fact, for any local $A\in \caA$, 
there exists an $N_0\in\bbN$ such that
\begin{align}
\begin{split}
\lmk W_{\partial \ld,N}^g\rmk^* A W_{\partial \ld,N}^g
=\lmk W_{\partial \ld,N_0}^g\rmk^* A W_{\partial \ld,N_0}^g
\in \caA_{[-N_0-l,N_0+l]},\quad N\ge N_0.
\end{split}
\end{align}
Therefore, we have
\begin{align}
\begin{split}
\beta_{g}^\ld\lmk \lmk W_{\partial \ld,N}^g\rmk^* A W_{\partial \ld,N}^g\rmk 
=\beta_{g}^{\ld\cap[-N,N]^2}\lmk \lmk W_{\partial \ld,N}^g\rmk^* A W_{\partial \ld,N}^g\rmk,\quad
N\ge N_0+2l. 
\end{split}
\end{align}
Substituting this, we obtain
we have
\begin{align}
\begin{split}
&v_{g\ld}\pi(A)\Omega
=\pi\lmk \beta_{g}^\ld\gamma_{g\ld}^{-1}(A)\rmk\Omega
=\pi  \beta_{g}^\ld\lmk  \lmk W_{\partial \ld,N}^g\rmk^* A W_{\partial \ld,N}^g  \rmk\Omega\\
&=\pi\lmk
\lmk \bigotimes_{\ld\cap[-N,N]^2} U_g\rmk\lmk W_{\partial \ld,N}^g\rmk^*\rmk
\pi(A)
\pi\lmk W_{\partial \ld,N}^g
\lmk \bigotimes_{\ld\cap[-N,N]^2} U_g\rmk^*
\rmk\Omega\\
&=\pi\lmk
\lmk \bigotimes_{\ld\cap[-N,N]^2} U_g\rmk\lmk W_{\partial \ld,N}^g\rmk^*\rmk
\pi(A)\Omega
\end{split}
\end{align}
from {\it 1.} of Assumption \ref{a7}. Because the strong convergence of unitaries to a unitary implies
that of strong $*$-convergence, this proves the claim.

Now we prove (ii).
Note that for any $\ld,\ld'\in\Cl$, we may find $\tilde\ld\in \Cl$
such that $\ld,\ld'\subset \tilde\ld$.
Therefore, it suffices to consider the case $\ld\subset\ld'$.
Let $\ld=\ld_{(a,0), \pi,\varphi}, \ld'=\ld_{(a',0), \pi,\varphi'}\in \Cl$
with $a\le a'$, $\varphi\le \varphi'$.
By (\ref{takos}), we have
\begin{align}
\begin{split}
&v_{g\ld}v_{g\ld'}^*
=w-\lim_N
\pi\lmk
\lmk \bigotimes_{\ld\cap[-N,N]^2} U_g\rmk\lmk W_{\partial \ld,N}^g\rmk^*
W_{\partial \ld',N}^g
 \lmk \bigotimes_{\ld'\cap[-N,N]^2} U_g\rmk^* \rmk\\
 &=w-\lim_N
\pi\lmk
\beta_g^{\ld}
\lmk \lmk W_{\partial \ld,N}^g\rmk^*W_{\partial \ld',N}^g\rmk
\lmk \bigotimes_{\ld'\setminus \ld\cap[-N,N]^2} U_g\rmk^*\rmk
\\
\end{split}
\end{align}
Note by {\it 2.} of Assumption \ref{a7} that
\begin{align}
\begin{split}
\beta_g^{\ld}
\lmk \lmk W_{\partial \ld,N}^g\rmk^*W_{\partial \ld',N}^g\rmk
\lmk \bigotimes_{\ld'\setminus \ld\cap[-N,N]^2} U_g\rmk^*
\in \caA_{
\lmk
\ld'\setminus \ld\cap [-N,N]^2
\rmk^{(l)}
}\subset \caA_{
\lmk
\ld'\setminus \ld
\rmk^{(l)}
}
.
\end{split}
\end{align}
Therefore, we have
\begin{align}
\begin{split}
v_{g\ld}v_{g\ld'}^*\in 
\pi\lmk  \caA_{
\lmk\ld'\setminus \ld \rmk^{(l)}}\rmk''
\subset \caF.
\end{split}
\end{align}
Similarly we have
\begin{align}
\begin{split}
v_{g\ld}^*v_{g\ld'}
=w-\lim
\pi\lmk \lmk W_{\partial \ld,N}^g\rmk^*
\lmk \bigotimes_{\ld'\setminus \ld\cap[-N,N]^2} U_g\rmk
W_{\partial \ld',N}^g
\rmk\in\caF.
\end{split}
\end{align}

Note for any $(\lo,\lt)\in PC$, there exists a $(\lo',\lt')\in PC$
such that $(\lo\cup \lt)^c\cup \lmk \partial \ld\rmk^{(l)}\subset (\lo'\cup \lt')^c$.
Therefore, we have 
\begin{align}
\begin{split}
\Ad \lmk v_{g\ld}\rmk\lmk\pi\lmk\caA_{(\lo\cup \lt)^c}\rmk''\rmk\subset
\pi\lmk\caA_{(\lo'\cup \lt')^c}\rmk'',
\end{split}
\end{align}
because $\gamma_{g\ld}$ is localized in $\lmk \partial \ld\rmk^{(l)}$.
This proves $\Ad\lmk v_{g,\ld}\rmk \lmk \caF\rmk\subset\caF$.
Hence we get 
$v_{g,\ld}\caF v_{g,\ld'}^*= \Ad\lmk v_{g,\ld}\rmk (\caF)\cdot v_{g,\ld}v_{g,\ld'}^*\subset \caF$.
Similarly we have $v_{g,\ld}^*\caF v_{g,\ld'}\subset \caF$
proving (ii).

\end{proof}

\section{Example}\label{ex}
In this section we provide a concrete example, the infinite version of  the
model in \cite{GIS}.
We denote by $\Gamma$ a honeycomb lattice.
The result in the previous sections for a square goes through in this honeycomb lattice as well.
We denote by $\bbV$ the set of all vertices of $\Gamma$
and by $\bbE$ the set of all edges of $\Gamma$.
The vertex part $\bbV$  is bipartite : we split it into $\bbV_{\bbA}$ and $\bbV_{\bbB}$.
Because of this bipartite picture, each edge $e$ has one
vertex $v_{eA}\in\bbA$ at its end and
another vertex $v_{eB}\in\bbB$ at the other end.
For each vertex $v\in\bbV$, $s(v)$
denotes the set of the three edges which have $v$ at their ends.
We put qubits on each vertex and edge and consider
$C^*$-algebras
\begin{align}
\begin{split}
\caA_{\bbE}:= \bigotimes_{e\in \bbE}\Mat_2,\quad
\caA_{\bbV}:=\bigotimes_{v\in \bbV}\Mat_2.
\end{split}
\end{align}
Our quantum spin system is then given by 
\begin{align}
\begin{split}
\caA:=\caA_{\bbE}\otimes \caA_{\bbV}.
\end{split}
\end{align}
We consider a toric code on edges $\bbE$.
For each vertex $v\in \bbV$, and a hexagonal plaquette $p$ of $\Gamma$, we set
\begin{align}
\begin{split}
A_v:=\bigotimes_{e\in s(v)}\sigma_z^{(e)},\quad  B_p:=\bigotimes_{e\in\bbE: e\subset\partial p}\sigma_x^{(e)},
\end{split}
\end{align}
with Pauli matrices $\sigma_x^{(e)}$, $\sigma_z^{(e)}$
associated to the edge $e$.
Here, 
${e\subset \partial p}$ means the edge $e$ is part of the boundary of $p$.
By the same argument as in \cite{Na1}, we can show that there exists a unique state
$\omega_{\bbE}$ on $\caA_{\bbE}$ such that
\begin{align}
\begin{split}
\omega_{\bbE}(A_v)=\omega_{\bbE}(B_p)=1
\end{split}
\end{align}
for all vertex $v$ and hexagon $p$.
Because of this uniqueness, $\omega_{\bbE}$  is pure.
By the same argument as in \cite{Na2}, $\omega_{\bbE}$  satisfies the Haag duality.
In particular, the hexagonal version of Assumption \ref{a1}, Assumption \ref{a1r}, 
Assumption \ref{a1l} hold. 
On $\caA_{\bbV}$, we consider a product state
$\psi_{\bbV}:=\psi_+^\otimes$, where $\psi_+$ is a pure state on
$\Mat_2$ such that $\psi_+(A)=\frac12\braket{\lmk \bm e_{-1}+\bm e_{+1}\rmk}{A\lmk \bm e_{-1}+\bm e_{+1}\rmk}$, $A\in\Mat_2$.
Here, $\bm e_{+1}$, $\bm e_{-1}$ are eigen vectors of $\sigma_z$ with eigenvaules
$+1$, $-1$ respectively.
Let $(\caH_{\bbE},\pi_{\bbE}, \Omega_{\bbE})$, $(\caH_{\bbV},\pi_{\bbV}, \Omega_{\bbV})$
be the GNS triples of $\omega_\bbE$, $\psi_{\bbV}$ respectively.
Because $\omega_{\bbE}$ and clearly $\psi_{\bbV}$ satisfy the 
 hexagonal version of Assumption \ref{a1}, Assumption \ref{a1r}, 
Assumption \ref{a1l}, so as $\omega_{\bbE}\otimes \psi_{\bbV}$.

Next we entangle the two systems.
For each finite simple loop $L$ of edges in $\Gamma$,
and the closed area $\ld_L$ surrounded by $L$,
we denote by $\bbV_{int}(L)$,  $\bbV_{ext}(L)$ the set of all vertices
in the interior, exterior of $\ld_L$, respectively.
We also denote by $\bbV_{bd}(L)$ the set of all vertices on $L$.
The set of all edges $e$ such that 
$v_{eA}\in \bbV_{int}(L)$ or  $v_{eB}\in \bbV_{int}(L)$
is denoted by $\bbE_{int}(L)$.
The set of all edges $e$ such that 
$v_{eA}\in \bbV_{ext}(L)$ or  $v_{eB}\in \bbV_{ext}(L)$
is denoted by $\bbE_{ext}(L)$.
We denote by $\bbE_{bd}(L)$ the set of all edges on $L$.
Then we obtain the partition $\bbE=\bbE_{int}(L)\cup \bbE_{ext}(L)\cup \bbE_{bd}(L)$.
Set
\begin{align}
\begin{split}
U_{\CCZ}^{\ld(L)}:=
\prod_{e\in \bbE_{int}(L)\cup  \bbE_{bd}(L)} \CCZ_{v_{eA} e v_{eB}}.
\end{split}
\end{align}
Here $ \CCZ_{v_{eA} e v_{eB}}$ is the CCZ-operator on $\bbC^2_{v_{eA}}\otimes \bbC_{e}^2\otimes
\bbC^2_{v_{eB}}$
given by
\begin{align}
\begin{split}
\CCZ_{v_{eA} e v_{eB}}\ket{x,y,z}
=(-1)^{xyz}\ket{x,y,z},\quad x,y,z\in \{0,1\},
\end{split}
\end{align}
where $\{\ket{x,y,z}\}_{ x,y,z\in \{0,1\}}$ is the orthogogonal basis of 
$\bbC^2_{v_{eA}}\otimes \bbC_{e}^2\otimes\bbC^2_{v_{eB}}$ 
consisting of simultatnious eigenvectors of
$\sigma_z^{(v_{eA})}\otimes\unit_e\otimes \unit_{v_{eB}}$,
$\unit_{v_{eA}}\otimes\sigma_z^{(e)}\otimes\unit_{v_{eB}}$,
$\unit_{v_{eA}}\otimes\unit_e\otimes\sigma_z^{(v_{eB})}$,
with eigenvalues $(-1)^x,(-1)^y,(-1)^z$ respectively.
Note that
\begin{align}\label{sada}
\begin{split}
\Ad\lmk \sigma_x^{(v_{eA})}\otimes \unit\rmk\lmk \CCZ_{v_{eA} e v_{eB}}\rmk
=\CCZ_{v_{eA} e v_{eB}} \CZ_{v_{eB}e},
\end{split}
\end{align}
where 
\begin{align}
\begin{split}
\CZ_{e v_{eB}}\ket{y,z}
=(-1)^{yz}\ket{y,z},\quad y,z\in \{0,1\},
\end{split}
\end{align}
where $\{\ket{y,z}\}_{ y,z\in \{0,1\}}$ is the orthogogonal basis of 
$ \bbC_{e}^2\otimes\bbC^2_{v_{eB}}$ 
consisting of simultatnious eigenvectors of
$\sigma_z^{(e)}
\otimes\unit_{v_{eB}}$,
$\unit_e\otimes\sigma_z^{(v_{eB})}$.
For the latter use, we also note that
\begin{align}\label{shio}
\begin{split}
\Ad\sigma_x^{(v_{eA})}\lmk\CZ_{v_{eB}v_{eA}}\rmk
=\CZ_{v_{eB}v_{eA}}\sigma_Z^{(v_{eB})}.
\end{split}
\end{align}

Because $\CCZ_{v_{eA} e v_{eB}}$s commute,
for any local $A\in\caA_{\rm loc}$,
we have
\begin{align}
\begin{split}
\Ad\lmk U_{\CCZ}^{\ld(L)}\rmk(A)
=\Ad\lmk U_{\CCZ}^{\ld(\tilde L)}\rmk(A)
\end{split}
\end{align}
for $\ld(L)$, $\ld(\tilde L)$ large enough.
From this, and the fact that ${U_{\CCZ}^{\ld(L)}}^{-1}=U_{\CCZ}^{\ld(L)}$,
there exists an automorphism $\alpha$ on $\caA$
such that
\begin{align}
\begin{split}
\alpha(A)=\lim_{\ld(L)\uparrow \Gamma}\Ad\lmk U_{\CCZ}^{\ld(L)}\rmk(A),\quad
A\in\caA.
\end{split}
\end{align}
Note that $\alpha=\alpha^{-1}$.
We set
\begin{align}
\begin{split}
\varphi:=\lmk \omega_{\bbE}\otimes \psi_{\bbV}\rmk\circ\alpha.
\end{split}
\end{align}
Because $\alpha$ is a finite depth quantum circuit and $\omega_{\bbE}\otimes \psi_{\bbV}$
satisfies the
hexagonal version of Assumption \ref{a1}, Assumption \ref{a1r}, 
Assumption \ref{a1l}, 
as in \cite{MTC}, $\varphi$ also satisfies the
hexagonal version of Assumption \ref{a1}, Assumption \ref{a1r}, 
Assumption \ref{a1l}.The triple 
\begin{align}
\begin{split}
\lmk \caH,\pi,\Omega\rmk
=\lmk\caH_{\bbE}\otimes \caH_{\bbV},\lmk\pi_{\bbE}\otimes \pi_{\bbV}\rmk\circ\alpha,
\Omega_{\bbE}\otimes\Omega_{\bbV}\rmk
\end{split}
\end{align}
is a GNS triple of $\varphi$.

Next we introduce the $\bbZ_2\times\bbZ_2$-action
given by
\begin{align}
\begin{split}
\beta_{\bbA}:=\bigotimes_{v\in V_{\bbA}} \Ad \sigma_x^{(v)},\quad
\beta_{\bbB}:=\bigotimes_{v\in V_{\bbB}} \Ad \sigma_x^{(v)},
\end{split}
\end{align}
with Pauli matrices $\sigma_x^{(v)}$ at $v\in\bbV$.
To show the invariance of $\varphi$ under this action and the Assumption \ref{a5},
let
\begin{align}
\begin{split}
u_A(L):=\bigotimes_{v\in \lmk \bbV_{int}(L)\cup \bbV_{bd}(L)\rmk\cap \bbV_A}\sigma_x^{(v)}
\end{split}
\end{align}
for a finite simple loop $L$ of edges in $\Gamma$.
Set
\begin{align}
W_L:=\lmk \prod_{\substack{e\in\bbE: \\ v_{eA}\in \bbV_{bd}(L)\cup \bbV_{int}(L)}}
\CZ_{v_{eB} e}\rmk.
\end{align}
Choose $4\le N\in \bbN$ large enough so that
$\ld(L)\subset [-\frac N2, \frac N 2]^2$.
Choose another finite simple loop $L_N$ such that
$ [-N, N ]^2\subset \ld(L_N)$.
We claim
\begin{align}\label{piero}
\begin{split}
u_A(L) U_{\CCZ}^{\ld(L_N) }
=W_L U_{\CCZ}^{\ld(L_N) }  u_A(L).
\end{split}
\end{align}
To see this, first note from (\ref{sada}) that
\begin{align}
\begin{split}
&\Ad\lmk u_A(L)\rmk\lmk
\prod_{e\in \lmk \bbE_{int}(L)\cup \bbE_{bd}(L)\rmk }\CCZ_{v_{eB} e v_{eA}}\rmk 
\\
&=\Ad\lmk
\bigotimes_{v\in \lmk \bbV_{int}(L)\cup \bbV_{bd}(L)\rmk\cap \bbV_A}\sigma_x^{(v)}
\rmk\lmk
\prod_{e\in \lmk \bbE_{int}(L)\cup \bbE_{bd}(L)\rmk }\CCZ_{v_{eB} e v_{eA}}\rmk\\
&=\prod_{e\in \lmk \bbE_{int}(L)\cup \bbE_{bd}(L)\rmk }
\Ad\lmk \sigma_x^{(v_{eA})}\rmk \lmk \CCZ_{v_{eB} e v_{eA}}\rmk\\
&=\prod_{e\in \lmk \bbE_{int}(L)\cup \bbE_{bd}(L)\rmk }
\CCZ_{v_{eB} e v_{eA}} \CZ_{v_{eB} e}.
\end{split}
\end{align}
We also have
\begin{align}
\begin{split}
&\Ad u_A(L)\lmk
\prod_{\substack{e\in \lmk \bbE_{ext}(L)\rmk\\
e\subset \ld(L_N)
} }\CCZ_{v_{eB} e v_{eA}}\rmk
\\
&=\Ad\lmk
\bigotimes_{v\in \lmk \bbV_{int}(L)\cup \bbV_{bd}(L)\rmk\cap \bbV_A}\sigma_x^{(v)}
\rmk\lmk
\prod_{\substack{e\in \lmk \bbE_{ext}(L)\rmk\\
e\subset \ld(L_N)
} }\CCZ_{v_{eB} e v_{eA}}\rmk\\
&=
\prod_{\substack{e\in \lmk \bbE_{ext}(L)\rmk\\
e\subset \ld(L_N)\\
v_{eA}\in \bbV_{ext}(L)
} }\CCZ_{v_{eB} e v_{eA}}\\
&
\prod_{\substack{e\in \lmk \bbE_{ext}(L)\rmk\\
e\subset \ld(L_N)\\
v_{eA}\in \bbV_{bd}(L)
} }
\Ad\lmk
\sigma_x^{(v_{eA})}
\rmk\lmk \CCZ_{v_{eB} e v_{eA}}\rmk\\
&=
\prod_{\substack{e\in \lmk \bbE_{ext}(L)\rmk\\
e\subset \ld(L_N)\\
v_{eA}\in \bbV_{ext}(L)
} }\CCZ_{v_{eB} e v_{eA}}
\prod_{\substack{e\in \lmk \bbE_{ext}(L)\rmk\\
e\subset \ld(L_N)\\
v_{eA}\in \bbV_{bd}(L)
} }
\CCZ_{v_{eB} e v_{eA}} \CZ_{v_{eB} e}
\\
&=
\prod_{\substack{e\in \lmk \bbE_{ext}(L)\rmk\\
e\subset \ld(L_N)
} }\CCZ_{v_{eB} e v_{eA}}
\prod_{\substack{e\in \lmk \bbE_{ext}(L)\rmk\\
e\subset \ld(L_N)\\
v_{eA}\in \bbV_{bd}(L)
} }
\CZ_{v_{eB} e},
\end{split}
\end{align}
using (\ref{sada}).
Therefore,
\begin{align}
\begin{split}
&u_A(L) U_{\CCZ}^{\ld(L_N) }
=U_{\CCZ}^{\ld(L_N) }
\lmk \prod_{e\in \lmk \bbE_{int}(L)\cup \bbE_{bd}(L)\rmk }
\CZ_{v_{eB} e}\rmk
\lmk \prod_{\substack{e\in \lmk \bbE_{ext}(L)\rmk\\
e\subset \ld(L_N)\\
v_{eA}\in \bbV_{bd}(L)
} }\CZ_{v_{eB} e}\rmk
u_A(L)\\
&=U_{\CCZ}^{\ld(L_N) }
\lmk \prod_{\substack{e\in\bbE: \\ v_{eA}\in \bbV_{bd}(L)\cup \bbV_{int}(L)}}
\CZ_{v_{eB} e}\rmk u_A(L)
=U_{\CCZ}^{\ld(L_N) } W_L u_A(L),
\end{split}
\end{align}
proving the claim (\ref{piero}).
Note that
\begin{align}
\begin{split}
&W_L=\prod_{\substack{e\in\bbE: \\ v_{eA}\in \bbV_{bd}(L)\cup \bbV_{int}(L)}}
\CZ_{v_{eB} e}\\
&=\lmk \prod_{\substack{v\in\bbV_{\bbB}\\ s(v)\subset \bbE_{int}(L)\cup \bbE_{bd}(L)}}
\prod_{e\in s(v)}\CZ_{v e}\rmk\\
&\lmk \prod_{\substack{v\in \bbV_{\bbB}\\ s(v)\nsubseteq \lmk \bbE_{int}(L)\cup \bbE_{bd}(L) \rmk}}
\prod_{\substack{
e\in s(v)\\
v_{eA}\in \bbV_{bd}(L)\cup \bbV_{int}(L)}}\CZ_{v e}\rmk\\
&=: W_L^{(1)}W_L^{(2)}
\end{split}
\end{align}
Hence for any finite simple loop $L$,
$4\le N\in \bbN$ large enough so that
$\ld(L)\subset [-\frac N2, \frac N 2]^2$.
and another finite simple loop $L_N$ such that
$ [-N, N ]^2\subset \ld(L_N)$,
we get
\begin{align}\label{keage}
\begin{split}
u_A(L) U_{\CCZ}^{\ld(L_N) }
=W_L^{(2)} U_{\CCZ}^{\ld(L_N) }  W_L^{(1)} u_A(L).
\end{split}
\end{align}

Note 
\begin{align}
\begin{split}
&\prod_{e\in s(v)}\CZ_{v e}
= \lmk \prod_{e\in s(v)}\CZ_{v e}\rmk
\lmk\unit_{s(v)}\otimes \lmk \ket{0}\bra{0}+\ket{1}\bra{1}\rmk_v\rmk
=\prod_{e\in s(v)}\unit_{e}\otimes \ket{0}_v\bra{0}
+ \prod_{e\in s(v)}\sigma_z\otimes \ket{1}_v\bra{1}\\
&=\unit_{s(v)}\otimes  \ket{0}_v\bra{0}
+A_v\otimes \ket{1}_v\bra{1}
\end{split}
\end{align}
for each $v\in \bbV_{\bbB}$.
Here $\ket{0}$, $\ket{1}$ are standard basis of $\bbC_v^2$.
Therefore,
we have
\begin{align}\label{toric}
\begin{split}
&W_L^{(1)}
=\prod_{\substack{v\in\bbV_{\bbB}\\ s(v)\subset \bbE_{int}(L)\cup \bbE_{bd}(L)}}
\prod_{e\in s(v)}\CZ_{v e}\\
&=\prod_{\substack{v\in\bbV_{\bbB}\\ s(v)\subset \bbE_{int}(L)\cup \bbE_{bd}(L)}}
\lmk
\unit_{s(v)}\otimes  \ket{0}_v\bra{0}
+A_v\otimes \ket{1}_v\bra{1}
\rmk.
\end{split}
\end{align}

We claim that  for each finite simple loop $L$ of edges in $\Gamma$,
and any local $Y\in\caA$, 
\begin{align}\label{kanpaku}
\begin{split}
\varphi\lmk Yu_A(L) \rmk
=\varphi\lmk Y W_L^{(2)}\rmk.
\end{split}
\end{align}
To see this, choose $4\le N\in\bbN$ large enough so that
$L\subset [-\frac N2, \frac N 2]^2$
and $Y\in \caA_{[-\frac N2, \frac N 2]^2}$.
Choose another finite simple loop $L_N$ such that
$ [-N, N ]^2\subset \ld(L_N)$.
Then using (\ref{keage}), we have
\begin{align}
\begin{split}
&\varphi\lmk
 Yu_A(L)
\rmk
=\lmk \omega_{\bbE}\otimes \psi_{\bbV}\rmk 
\Ad\lmk U_{CCZ}^{\ld(L_N)}\rmk
\lmk
 Yu_A(L)
\rmk\\
&=\lmk \omega_{\bbE}\otimes \psi_{\bbV}\rmk 
\lmk
U_{CCZ}^{\ld(L_N)} Y
W_L^{(2)} U_{\CCZ}^{\ld(L_N) }  W_L^{(1)} u_A(L)
\rmk\\
&=
\lmk \omega_{\bbE}\otimes \psi_{\bbV}\rmk 
\lmk
U_{CCZ}^{\ld(L_N)} Y
W_L^{(2)} U_{\CCZ}^{\ld(L_N) }  
\rmk
=\varphi\lmk
 Y
W_L^{(2)}
\rmk.
\end{split}
\end{align}
In the last line, we used the definition of $\psi_{\bbV}=\psi_+^\otimes$ as a product state,
and (\ref{toric}).
This proves the claim (\ref{kanpaku}).
Because this holds for any local $Y\in\caA$,
we obtain
\begin{align}
\begin{split}
\pi\lmk
u_A(L)
\rmk\Omega
=\pi\lmk W_L^{(2)}\rmk\Omega,
\end{split}
\end{align}
for each finite simple loop $L$ of edges in $\Gamma$.
Because 
 \begin{align}
 \begin{split}
  W_L^{(2)}:=
  \lmk \prod_{\substack{v\in \bbV_{\bbB}\\ s(v)\nsubseteq \lmk \bbE_{int}(L)\cup \bbE_{bd}(L) \rmk}}
\prod_{\substack{
e\in s(v)\\
v_{eA}\in \bbV_{bd}(L)\cup \bbV_{int}(L)}}\CZ_{v e}\rmk
 \end{split}
 \end{align}
 is a product of commuting terms, this proves
 {\it 1., 2., 3.,4.} of Assumption \ref{a7} for $\beta_{\bbA}$.
 The same argument proves Assumption \ref{a7} for $\beta_{\bbB}$.

Furthermore, choose a sequence of  finite simple loops $L_n$ of edges in $\Gamma$
whose distance from the origin goes to infinity as $n\to \infty$.
Then for each $A\in \caA_{\rm loc}$, for $n$ large enough,
we have
\begin{align}
\begin{split}
\varphi\beta_{\bbA}(A)
=\varphi\Ad u_A(L_n)(A)
=\varphi \Ad  W_{L_n}^{(2)}(A)
=\varphi(A).
\end{split}
\end{align}
This implies $\varphi\beta_{\bbA}=\varphi$.
Similarly, we have  $\varphi\beta_{\bbB}=\varphi$.
Hence we are in the hexagon version of the setting \ref{sec:qss}.
By Lemma \ref{amayadori}, Assumption \ref{a5} is satisfied for this model.

Next we consider an anyon of $\varphi$.
Let $\gamma$ be a half-infinite line of edges starting from a vertex $v_0\in\bbV_{\bbB}$
going to infinity
and set
\begin{align}
\begin{split}
&\tilde\rho_{X,\gamma}(A):=
\lim_{N\to\infty}
\Ad\lmk
\bigotimes_{e\in \gamma\cap[-N,N]^2} \sigma_x^{(e)}\rmk(A)
,\\
&\rho_{X,\gamma}(A)=\alpha^{-1}\tilde\rho_{X,\gamma}\alpha(A), \quad A\in \caA.
\end{split}
\end{align}
Recall from \cite{Na1} that $\pi_{\bbE}\tilde \rho_{X,\gamma}$
satisfies the superselection criterion for $\pi_{\bbE}$.
From this, $\pi\rho_{X,\gamma}$ satisfies the superselection criterion of $\pi$.
To see this, let $\ld$ be a cone. Note that $\alpha\lmk\caA_{\ld^c} \rmk\subset 
\caA_{\lmk \ld_{-\varepsilon}+t\bm e_{\ld}\rmk^c}$ for some $\varepsilon>0$ and $t>0$.
Therefore, we have
\begin{align}
\begin{split}
&\left. \pi\rho_{X,\gamma}\right\vert_{\caA_{\ld^c}}
=\left.  \lmk\pi_{\bbE}\otimes \pi_{\bbV}\rmk\circ\alpha\circ \alpha^{-1}\tilde\rho_{X,\gamma}\alpha\right\vert_{\caA_{\ld^c}}\\
&=\left.  \lmk\pi_{\bbE}\otimes \pi_{\bbV}\rmk\circ
\tilde\rho_{X,\gamma}\alpha\right\vert_{\caA_{\ld^c}}\\
&=\left. \Ad\lmk
V_{\pi_{\bbE}\tilde\rho_{X,\gamma},  \ld_{-\varepsilon}+t\bm e_{\ld}}^*\otimes\unit
\rmk\lmk\pi_{\bbE}\otimes \pi_{\bbV}\rmk\alpha\right\vert_{\caA_{\ld^c}}\\
&=\left. \Ad\lmk
V_{\pi_{\bbE}\tilde\rho_{X,\gamma},  \ld_{-\varepsilon}+t\bm e_{\ld}}^*\otimes\unit
\rmk\pi \right\vert_{\caA_{\ld^c}},
\end{split}
\end{align}
with a unitary $V_{\pi_{\bbE}\tilde\rho_{X,\gamma},  \ld_{-\varepsilon}+t\bm e_{\ld}}\in \caV_{\pi_{\bbE}\tilde\rho_{X,\gamma},  \ld_{-\varepsilon}+t\bm e_{\ld}}$.

For an edge $e$ of $\Gamma$,  $\alpha^{-1}$ acts on
the Pauli $x$ matrix at $e$,
$\sigma_x^{(e)}$ 
as
\begin{align}
\alpha^{-1}\lmk\sigma_x^{(e)}\rmk
=\Ad\lmk \CCZ_{v_{eA} e v_{eB}}\rmk\lmk\sigma_x^{(e)}\rmk
=\sigma_x^{(e)}\CZ_{v_{eA} v_{eB}}.
\end{align}
From this, we obtain
\begin{align}
\begin{split}
&\rho_{X,\gamma}(A)=\alpha^{-1}\tilde\rho_{X,\gamma}\alpha(A)
=\lim_{N\to\infty}
\Ad\lmk
\alpha^{-1}
\lmk \bigotimes_{e\in \gamma\cap[-N,N]^2} \sigma_x^{(e)}\rmk\rmk
(A)\\
&=\lim_{N\to\infty}
\Ad\lmk
\prod_{e\in \gamma\cap[-N,N]^2} \sigma_x^{(e)}\CZ_{v_{eA} v_{eB}}\rmk
(A),\quad A\in \caA.
\end{split}
\end{align}

Now we consider the $\bbZ_2\times \bbZ_2$-action
$\Theta$ on 
$\pi\rho_{X,\gamma}$ given by $\beta_{\bbA}$, $\beta_{\bbB}$.
We have
\begin{align}
\begin{split}
&\Theta\lmk1, 0\rmk\lmk \pi\rho_{X,\gamma}\rmk(A)
=\pi\lmk\beta_{\bbA}\rho_{X,\gamma}  \beta_{\bbA}^{-1}\rmk(A)\\
&=\lim_{N\to\infty}
\pi\lmk
\Ad\lmk
\prod_{e\in \gamma\cap[-N,N]^2} 
\beta_{\bbA}\lmk \sigma_x^{(e)}\CZ_{v_{eA} v_{eB}}\rmk\rmk
(A)
\rmk\\
&=\lim_{N\to\infty}
\pi\lmk
\Ad\lmk
\prod_{e\in \gamma\cap[-N,N]^2} 
 \sigma_x^{(e)}\CZ_{v_{eA} v_{eB}} \sigma_z^{(v_{eB})}\rmk
(A)
\rmk\\
&=\Ad\lmk \pi\lmk \sigma_z^{(v_0)}\rmk\rmk\circ\pi\rho_{X,\gamma}(A)
\end{split}
\end{align}
because $\gamma$ is starting from $v_0\in \bbV_{\bbB}$.
Here we used (\ref{shio}).
Similarly, we have 
\begin{align}
\begin{split}
&\Theta\lmk 0,1\rmk\lmk \pi\rho_{X,\gamma}\rmk(A)
=\pi\beta_{\bbB} \rho_{X,\gamma} \beta_{\bbB}^{-1}(A)\\
&=\lim_{N\to\infty}
\pi\lmk
\Ad\lmk
\prod_{e\in \gamma\cap[-N,N]^2} 
 \sigma_x^{(e)}\CZ_{v_{eA} v_{eB}} \sigma_z^{(v_{eA})}\rmk
(A)
\rmk\\
&=\pi \rho_{X,\gamma}(A).
\end{split}
\end{align}
We also have 
\begin{align}
\begin{split}
&\Theta\lmk 1,1\rmk\lmk \pi\rho_{X,\gamma}\rmk(A)
=\pi\beta_{\bbB}\beta_{\bbA} \rho_{X,\gamma} \beta_{\bbA}^{-1}\beta_{\bbB}^{-1}(A)\\
&=\pi \beta_{\bbB} \Ad\lmk \sigma_z^{(v_0)}\rmk \rho_{X,\gamma}\beta_{\bbB}^{-1}(A)
=\Ad\lmk \pi \beta_{\bbB} \lmk \sigma_z^{(v_0)}\rmk\rmk
\pi \beta_{\bbB} \rho_{X,\gamma}\beta_{\bbB}^{-1}(A)\\
&=\Ad\lmk \pi \beta_{\bbB} \lmk \sigma_z^{(v_0)}\rmk\rmk\pi \rho_{X,\gamma}(A).
\end{split}
\end{align}
Therefore, taking $\pi\rho_{X,\gamma}$ the representative,
 we may set $W_{\left[\pi\rho_{X,\gamma}\right]}^{(g)}$ in (\ref{wa}) as 
\begin{align}
\begin{split}
\begin{gathered}
W_{\left[\pi\rho_{X,\gamma}\right]}^{(0,0)}:=\unit,\quad
W_{\left[\pi\rho_{X,\gamma}\right]}^{(1,0)}:= \pi\lmk \sigma_z^{(v_0)}\rmk,\quad
W_{\left[\pi\rho_{X,\gamma}\right]}^{(0,1)}:=\unit,\quad
W_{\left[\pi\rho_{X,\gamma}\right]}^{(1,1)}:= \pi \beta_{\bbB} \lmk \sigma_z^{(v_0)}\rmk.
\end{gathered}
\end{split}
\end{align}
With this, we obtain
\begin{align}
\begin{split}
\omega^{({\left[\pi\rho_{X,\gamma}\right]})}\lmk g,h\rmk
=\begin{pmatrix}
1&1&1&1\\
1&1&-1&-1\\
1&1&1&1\\
1&1&-1&-1
\end{pmatrix}
\end{split}
\end{align}
with $g,h$ ordered as $(0,0)$,$(1,0)$,$(0,1)$,$(1,1)$.

\appendix
\section{$G$-crossed Category}\label{gcross}
In this section, under the split property (Assumption \ref{split}),
we derive braided crossed $G$-categories in the sense of
\cite{muger2005conformal}.
Here is the additional assumption we require.
\begin{assum}\label{split}
Consider the setting in subsection \ref{sec:qss}.
For any $(\lo,\lt)\in PC$, there exists a type I factor $F$ such that
\begin{align}
\pi\lmk\caA_{\lo}\rmk''\subset F\subset \pi\lmk\caA_{\lt}\rmk'.
\end{align}
\end{assum}
The main reason to assume this is the following Lemma.
\begin{lem}\label{gnh}
Consider the setting in subsection \ref{sec:qss}.
Assume Assumption \ref{split}.
Let $\rho\in O^{(g)}_{(\loz,\ltz)}$ and $\sigma\in O^{(h)}_{(\loz,\ltz)}$ with $g,h\in G$.
If $\Mor_G\lmk\rho,\sigma\rmk\neq \emptyset$, then $g=h$.
\end{lem}
\begin{proof}
Suppose that $g\neq h$.
Because (\ref{ffl}) is a faithful action, 
there exists $A_0\in \caA_{\{\bm 0\}}$ 
with $\lV A_0\rV=1$ such that 
\[
\delta:=\lV \beta_g(A_0)-\beta_h(A_0)\rV>0.
\]
Suppose that there exists an element
$X\in \Mor_G\lmk\rho,\sigma\rmk$ with $\lV X\rV=1$.
We derive a contradiction out of this.
Because $X\in \caF$ by the definition, there exists  $(\lo,\lt)\in PC$
and $x\in \pi\lmk \caA_{(\lo\cup\lt )^c}\rmk''$ with $\lV x\rV=1$
such that $\lV X-x\rV<\frac \delta 4$.
Choose $\ld_3\in \Cl$ such that $(\ld_3,\lmk \lo \rmk^c)\in PC$.
By Assumption \ref{split}, there exists a type I factor such that
\begin{align}
\begin{split}
\pi\lmk\caA_{\ld_3}\rmk''\subset F\subset \pi\lmk\caA_{\lmk \lo \rmk^c}\rmk'.
\end{split}
\end{align}
Choose a point $\bm z\in \ld_3\cap \loz$, and let  $A_z\in \caA_{\{\bm z\}}$ be
the copy of $A_0$ in $\caA_{\{\bm z\}}$.
Then we have
\begin{align}
\begin{split}
&x\in \pi\lmk \caA_{(\lo\cup\lt )^c}\rmk''\subset \pi\lmk \caA_{(\lo)^c}\rmk''\subset F'\\
& \pi\lmk \beta_g(A_z)-\beta_h(A_z)\rmk\in \pi\lmk\caA_{\ld_3}\rmk''
\subset F.
\end{split}
\end{align}
Therefore we have
\begin{align}
\begin{split}
\lV x\cdot \lmk \pi\lmk \beta_g(A_z)-\beta_h(A_z)\rmk \rmk\rV
=\lV x\rV \lV \beta_g(A_z)-\beta_h(A_z)\rV
=\delta.
\end{split}
\end{align}
Substituting this, 
we have
\begin{align}
\begin{split}
&0=\lV X \rho(A_z)-\sigma(A_z)X\rV
=\lV X\cdot \pi\beta_g(A_z)-\pi\beta_h(A_z)\cdot X\rV\\
&\ge \lV x\cdot \pi\beta_g(A_z)-\pi\beta_h(A_z)\cdot x\rV
-2\lV X-x\rV \\
&\ge \lV x\cdot \lmk \pi\lmk \beta_g(A_z)-\beta_h(A_z)\rmk \rmk\rV-\frac12 \delta\\
&=\delta-\frac12 \delta>0,
\end{split}
\end{align}
which is a contradiction.

\end{proof}
To show the existence of subobjects, we will need the following Lemma.
\begin{lem}\label{lem69}Consider the setting in subsection \ref{sec:qss}.
Assume Assumption \ref{a2}.
Let $(\lo,\lt), (\go,\gt)\in PC$,
$D\in \CUbk$ with
$\lo\subset \go$, $\lt\subset \gt$, 
$D\cap\lmk \lo\cup \lt\rmk=\emptyset$, and $D\subset \go$.
Then for any projection $p\in  \pi\lmk\caA_{\go\cup\gt}\rmk'\cap \caF$,
there exists an isometry
$w\in \pi\lmk \caA_{\lo\cup\lt }\rmk'\cap\caF$
such that $ww^*=p$.
\end{lem}
\begin{proof}
We apply Lemma 5.10 of \cite{MTC}.
Let $\delta>0$ be the number given in  Lemma 5.10 of \cite{MTC}.
We also use Lemma B.2,Lemma A.3, Lemma A.2 of \cite{BA}.
Consider the number $\delta_3(\delta_2(\delta))$ given for 
 $\delta>0$(the number given in  Lemma 5.10 of \cite{MTC}) with the functions
 $\delta_2$, $\delta_3$ given in Lemma A.2, Lemma A.3 \cite{BA}.

Let $p\in \pi\lmk \caA_{\go\cup\gt}\rmk'\cap\caF$ be a projection.
Then, because $p\in\caF$, there exists $\lmk\tilde\go,\tilde \gt\rmk\in PC$, 
and a self-adjoint $x\in \pi(\caA_{\lmk \tilde\go\cup \tilde \gt\rmk^c})''$ such that
$\lV
p-x
\rV\le \delta_3(\delta_2(\delta))$.
We may assume that $\tilde\go\subset \go$ and $\tilde\gt\subset \gt$,
$D\subset \lmk \tilde\go\cup \tilde \gt\rmk^c$.

Apply Lemma B.2  of \cite{BA} with
$\lm 1$, $\lm 2$, $\Gamma$, $(\caH,\pi)$,
$x,y,\varepsilon$
replaced by 
$\lmk \go\cup\gt\rmk^c $, $\lmk \tilde\go\cup \tilde \gt\rmk^c$, $\bbZ^2$, $(\caH,\pi)$
$p,x,\delta_3(\delta_2(\delta))$.
Then we obtain self-adjoint
$z\in \pi\lmk\caA_{\go\cup\gt}\rmk'\cap\pi(\caA_{\lmk \tilde\go\cup \tilde \gt\rmk^c})''$
such that $\lV p-z\rV\le \delta_3(\delta_2(\delta))$.

Now apply Lemma A.3 \cite{BA} with
$\caH$, $\caA$, $p$, $x$
replaced by $\caH$, $ \pi\lmk\caA_{\go\cup\gt}\rmk'\cap\pi(\caA_{\lmk \tilde\go\cup \tilde \gt\rmk^c})''$,
$p$, $z$ respectively.
Then, we obtain
a projection $q\in  \pi\lmk\caA_{\go\cup\gt}\rmk'\cap\pi(\caA_{\lmk \tilde\go\cup \tilde \gt\rmk^c})''$
such that $\lV q-p\rV<\delta_2(\delta)$.
%
%
Next we apply Lemma A.2 of \cite{BA} with
$\caA$, $p$, $q$ replaced by
$ \pi\lmk\caA_{\go\cup\gt}\rmk'\cap\caF$, $p$, $q$ respectively.
Then from Lemma A.2 of \cite{BA}, we obtain a unitary
$u\in \pi\lmk\caA_{\go\cup\gt}\rmk'\cap\caF$
such that $p=uqu^*$, $\lV u-\unit\rV<\delta$.

As in Lemma 2.6 \cite{BA}, 
$\caM:=\pi\lmk\caA_{\lo\cup\lt }\rmk'\cap \pi\lmk\caA_{\lmk \tilde\go\cup \tilde \gt\rmk^c}\rmk''$
is a factor.
There is a cone $\ld$ such that $\ld\subset \lmk \go\cup\gt\rmk^c\subset \lmk \lo\cup\lt \rmk^c\cap \lmk \tilde\go\cup \tilde \gt\rmk^c$.
Because $\caM$ includes $\pi(\caA_{\Lambda})''$, it means
$\caM$ is an infinite factor.

Now we apply Lemma 5.10 of \cite{MTC}, with
$\caH$, $\caN$, $\caM$, $p$, $u$
replaced by $\caH$, $\pi(\caA_{D})''$, $\caM$, $p$, $u^*$.
Then from Lemma 5.10 of \cite{MTC}, we have
$q\sim \unit$ in $\caM$.
Namely, there exists an isometry $v\in\caM$
 such that
 $vv^*=q$.
 
 Set $w:=uv\in \pi\lmk\caA_{\lo\cup\lt }\rmk'\cap \caF$.
 Then this $w$ is isometry with
 $ww^*=uqu^*=p$, proving the claim.

\end{proof}
From this Lemma, we obtain the following.
\begin{lem}\label{sg}Consider the setting in subsection \ref{sec:qss}.
Assume Assumption \ref{a2}.
For any $\rho\in O^{(g)}_{(\loz,\ltz)}$ and a projection
$p\in \Mor_G\lmk\rho,\rho\rmk$,
there exists a $\gamma\in O^{(g)}_{(\loz,\ltz)}$
and an isometry $v\in \Mor_G\lmk\gamma,\rho\rmk$
such that $vv^*=p$.
\end{lem}
\begin{proof}
For any $(\lo,\lt)\in PC$, choose $V_{{\rho},{(\lo,\lt)}}^{(g)} \in \Vbk{\rho}{(\lo,\lt)}^{(g)}$.
Then we have
\begin{align}
\begin{split}
p_{(\lo,\lt)}:=\Ad\lmk V_{{\rho},{(\lo,\lt)}}^{(g)} \rmk(p)
\in \pi\lmk\caA_{\lo\cup\lt}\rmk'\cap \caF.
\end{split}
\end{align}
In fact, for any $A\in \caA_{\lo\cup\lt}$,
we have
\begin{align}
\begin{split}
&p_{(\lo,\lt)}\pi(A)
= V_{{\rho},{(\lo,\lt)}}^{(g)} p \lmk V_{{\rho},{(\lo,\lt)}}^{(g)}\rmk^*\pi(A)
V_{{\rho},{(\lo,\lt)}}^{(g)} \lmk V_{{\rho},{(\lo,\lt)}}^{(g)}\rmk^*\\
&=V_{{\rho},{(\lo,\lt)}}^{(g)} p \rho\beta_{g^{-1}}^{\lo}
(A)
 \lmk V_{{\rho},{(\lo,\lt)}}^{(g)}\rmk^*
 = V_{{\rho},{(\lo,\lt)}}^{(g)} \rho\beta_{g^{-1}}^{\lo}
(A) p
 \lmk V_{{\rho},{(\lo,\lt)}}^{(g)}\rmk^*\\
 &= V_{{\rho},{(\lo,\lt)}}^{(g)} \rho\beta_{g^{-1}}^{\lo}
(A) \lmk V_{{\rho},{(\lo,\lt)}}^{(g)}\rmk^*  V_{{\rho},{(\lo,\lt)}}^{(g)}p
 \lmk V_{{\rho},{(\lo,\lt)}}^{(g)}\rmk^*
 =\pi(A)p_{(\lo,\lt)}.
\end{split}
\end{align}

Next, for each $(\lo,\lt)\in PC$, choose 
$\lmk\Gamma_{1}^{(\lo,\lt)}, \Gamma_{2}^{(\lo,\lt)}\rmk\in PC$
with $\lo\subset \Gamma_{1}^{(\lo,\lt)}$, $\lt\subset \Gamma_{2}^{(\lo,\lt)}$, 
so that there exists 
$D\in \CUbk$ with
$D\cap\lmk\lo\cup \lt\rmk=\emptyset$, and $D\subset \Gamma_{1}^{(\lo,\lt)}$.
Applying the previous Lemma to 
\begin{align}
\begin{split}
p_{\lmk\Gamma_{1}^{(\lo,\lt)}, \Gamma_{2}^{(\lo,\lt)}\rmk}
\in \pi\lmk\caA_{\Gamma_{1}^{(\lo,\lt)}\cup \Gamma_{2}^{(\lo,\lt)}}\rmk'\cap \caF,
\end{split}
\end{align}
there exists an isometry 
$w_{(\lo,\lt)}\in \pi\lmk \caA_{\lo\cup\lt }\rmk'\cap\caF$
such that $w_{(\lo,\lt)}w_{(\lo,\lt)}^*=p_{\lmk\Gamma_{1}^{(\lo,\lt)}, \Gamma_{2}^{(\lo,\lt)}\rmk}$.
Then we have
\begin{align}
\begin{split}
\gamma:=\Ad\lmk w_{(\loz,\ltz)}^* V_{{\rho},{\lmk \Gamma_{1}^{(\loz,\ltz)},\Gamma_{2}^{(\loz,\ltz)}\rmk }}^{(g)} \rmk
\rho\in O^{(g)}_{(\loz,\ltz)},
\end{split}
\end{align}
and
\begin{align}
\begin{split}
v:=\lmk V_{{\rho},{\lmk \Gamma_{1}^{(\loz,\ltz)} ,\Gamma_{2}^{(\loz,\ltz)}\rmk }}^{(g)}\rmk^*  w_{(\loz,\ltz)}\in \Mor_G(\gamma, \rho).
\end{split}
\end{align}
Here $v$ is an isometry such that $vv^*=p$.
In fact, for any $(\lo,\lt)\in PC$, 
set
\begin{align}
\begin{split}
X_{(\lo,\lt)}:=w_{(\lo,\lt)}^* V_{{\rho},{\lmk \Gamma_{1}^{(\lo,\lt)},\Gamma_{2}^{(\lo,\lt)}\rmk }}^{(g)} 
\lmk
w_{(\loz,\ltz)}^* V_{{\rho},{\lmk \Gamma_{1}^{(\loz,\ltz)},\Gamma_{2}^{(\loz,\ltz)}\rmk }}^{(g)} 
\rmk^*.
\end{split}
\end{align}
This $X_{(\lo,\lt)}$ is a unitary in $\caF$
and
\begin{align}
\begin{split}
&\left.\Ad X_{(\lo,\lt)}\circ\gamma\right\vert_{\caA_{\lo\cup\lt}}
=\left.\Ad \lmk w_{(\lo,\lt)}^* V_{{\rho},{\lmk \Gamma_{1}^{(\lo,\lt)},\Gamma_{2}^{(\lo,\lt)}\rmk }}^{(g)} \rmk
\rho\right\vert_{\caA_{\lo\cup\lt}}\\
&=\left.\Ad \lmk w_{(\lo,\lt)}^*\rmk
\pi\beta_g^{\Gamma_1^{(\lo,\lt)}}
\right\vert_{\caA_{\lo\cup\lt}}
=\left.\Ad \lmk w_{(\lo,\lt)}^*\rmk
\pi\beta_g^{\Lambda_1}
\right\vert_{\caA_{\lo\cup\lt}}
=\left.
\pi\beta_g^{\Lambda_1}
\right\vert_{\caA_{\lo\cup\lt}},
\end{split}
\end{align}
because $\lo\subset \Gamma_{1}^{(\lo,\lt)}$, $\lt\subset \Gamma_{2}^{(\lo,\lt)}$.
Similarly, 
for each cone $\Gamma$, we choose 
$V_{\rho\Gamma}\in \caV_{\rho\Gamma}$ (see \cite{MTC}).
Then we have $\Ad V_{\rho\Gamma} (p)\in \pi\lmk\caA_{\Gamma^c}\rmk'$.
For each cone $\ld$, choose cones $\Gamma_\ld, D_\ld$
such that $D_\Lambda\cap\Gamma_\Lambda=\emptyset$
and $\Gamma_\ld\subset \ld$.
Then there exists an isometry $w_\ld\in \pi\lmk\caA_{\ld^c}\rmk'$
with $w_\ld w_\ld^*=\Ad V_{\rho\Gamma_\ld} (p)$.
Setting 
\[Y_\ld:=w_\ld^* V_{\rho\Gamma_\ld} \lmk
w_{(\loz,\ltz)}^* V_{{\rho},{\lmk \Gamma_{1}^{(\loz,\ltz)},\Gamma_{2}^{(\loz,\ltz)}\rmk }}^{(g)} 
\rmk^
*,\]
we have $\Ad Y_\ld\gamma\vert_{\caA_{\ld^c}}=\pi\vert_{\caA_{\ld^c}}$.
Hence we also have $\gamma\in \caO$.
\end{proof}
The following Lemma gives the direct sums.
\begin{lem}\label{gds}
Consider the setting in subsection \ref{sec:qss}.
Assume Assumption \ref{a2}.
For any $g\in G$ and
 $\rho,\sigma\in O^{(g)}_{(\loz,\ltz)}$,
there exits $\gamma\in O^{(g)}_{(\loz,\ltz)}$
and isometries 
$u\in \Mor_G(\rho,\gamma)$, $v\in \Mor_G(\sigma,\gamma)$
such that $uu^*+vv^*=\unit$.
\end{lem}
\begin{proof}
Choose $\ld_3^{(0)}\in \CUbk$ with $\ld_3^{(0)}\subset \lmk \loz\cup\ltz\rmk^c$.
Because $\pi\lmk\caA_{\ld_3^{(0)}}\rmk''$ is properly infinite,
there exist isometries 
\begin{align}
u,v\in \pi\lmk\caA_{\ld_3^{(0)}}\rmk''\subset \caF
\end{align}
such that $uu^*+vv^*=\unit$.
These isometries $u,v$ and
\begin{align}
\gamma:=\Ad u\cdot \rho+\Ad v \cdot \sigma
\end{align}
satisfy the condition.

\end{proof}
Recall the definition of a strict crossed $G$-category from Definition 2.9 of \cite{muger2005conformal}.
\begin{thm}Consider the setting in subsection \ref{sec:qss}.
Assume Assumption \ref{a1}, Assumption \ref{a1l}, Assumption \ref{a1r},
Assumption \ref{a2}, Assumption \ref{split}.
We define a category $C_G$ as follows.
The objects $\Obj C_G$ of $C_G$ are finite direct sums of
elements in $O_G$. Namely,
$\rho\in \Obj C_G$ is of the finite sum form
\begin{align}\label{yama}
\begin{split}
\rho=\sum_{g\in G}\Ad u_g\circ \rho_g.
\end{split}
\end{align}
Here $u_g$ is either an 
isometry $u_g\in \caF$ or $0$ such that 
$\sum_{g\in G} u_gu_g^*=\unit$,
and for $g\in G$ with 
nonzero 
$u_g$, $\rho_g$ belongs to $O^{(g)}_{(\loz,\ltz)}$.
Morphisms between $\rho,\sigma\in\Obj C_G$ are
\begin{align}
\begin{split}
\Mor_{C_G}\lmk \rho,\sigma\rmk:=
\left\{
S\in\caF\mid S\rho(A)=\sigma(A)S,\quad A\in \caA
\right\}.
\end{split}
\end{align}
For each $\rho\in \Obj C_G$, the identity morphism $\id_\rho$ is $\id_{\caH}$,
and composition of morphisms are just multiplication of $\caB(\caH)$.

For  $\rho\in \Obj C_G$ of the form (\ref{yama}),
we set
\begin{align}\label{shat}
\begin{split}
\hat S_\rho^{(l)}:=\sum_{g\in G}\Ad u_g\circ S_{\rho_g}^{(l)\unit},
\end{split}
\end{align}
which defines an endomorphism on $\caB_l$ satisfying
$\hat S_\rho^{(l)}(\caF)\subset \caF$.

The category $C_G$ is a 
strict $C^*$-tensor category
and 
crossed $G$-category with a homogeneous objects $O_G$,
 with respect to the tensor product
\begin{align}\label{tencg}
\begin{split}
\rho\otimes_{C_G} \sigma:=\hat S_\rho^{(l)\unit}\hat S_\sigma^{(l)\unit}\pi,\quad
\rho,\sigma\in \Obj C_G,
\end{split}
\end{align}
and 
\begin{align}\label{mortencg}
\begin{split}
X\otimes_{C_G} Y:=X \hat S_\rho^{(l)\unit}(Y),\quad X\in \Mor_{C_G}(\rho,\rho'),\quad
Y\in \Mor_{C_G}(\sigma,\sigma'),\quad \rho,\sigma, \rho',\sigma'\in \Obj C_G.
\end{split}
\end{align}
The tensor unit is the representation $\pi$.
\end{thm}
\begin{proof}
First we prove that $C_G$ is a $C^*$-tensor category.
That $C_G$ is a $C^*$-category is trivial.
That (\ref{shat}) defines an endomorphism on $\caB_l$ preserving $\caF$ follows 
from the properties
of $u_g$ and the fact that $u_g\in\caF$.
Note that for $\rho,\sigma\in \Obj C_G$, of the form
\begin{align}
\begin{split}
\rho=\sum_{g\in G} \Ad u_g\circ\rho_g,\quad
\sigma=\sum_{g\in G}\Ad v_g\circ \sigma_g\in \Obj C_G,
\end{split}
\end{align}
we have
\begin{align}
\begin{split}
&\rho\otimes_{C_G} \sigma:=\hat S_\rho^{(l)\unit}\hat S_\sigma^{(l)\unit}\pi
=\sum_{g,h\in G}\Ad\lmk u_g S_{\rho_g}^{(l)\unit}(v_h)\rmk
S_{\rho_g}^{(l)\unit}S_{\sigma_h}^{(l)\unit}\pi\\
&=\sum_{g,h\in G}\Ad\lmk u_g S_{\rho_g}^{(l)\unit}(v_h)\rmk
\rho_g\otimes\sigma_h
=\sum_{g\in G} \sum_{k\in G} \Ad\lmk u_g S_{\rho_g}^{(l)\unit}(v_{g^{-1}k})\rmk
\rho_g\otimes\sigma_{g^{-1}k}
\end{split}
\end{align}
Let $k\in G$ be an element such that there exist $g\in G$ with
 $u_g\neq 0$ and $v_{g^{-1}k}\neq 0$.
For such $k\in G$,
by Lemma \ref{gds}, we obtain
a $\gamma_k\in O^{(k)}_{(\loz,\ltz)}$
and $V_{k,g}$,
$g\in G$ such that
\begin{description}
\item[(i)] for $g\in G$ with  $u_g=0$ or $v_{g^{-1}k}=0$, $V_{k,g}=0$, and 
\item[(ii)] for $g\in G$ with $u_g\neq 0$ and $v_{g^{-1}k}\neq 0$, 
$V_{k,g}\in \Mor_G\lmk\rho_g\otimes\sigma_{g^{-1}k},\gamma_k\rmk$ is an isometry
\end{description}
such that $\sum_{g\in G} V_{k,g}V_{k,g}^*=\unit$.
For $k\in G$ with $u_g=0$ or $u_{g^{-1}k}=0$ for all $g\in G$,
we set $V_{k,g}=0$, $g\in G$.
Set
\begin{align}
\begin{split}
U_k:=
\sum_{g\in G} u_g S_{\rho_g}^{(l)\unit}(v_{g^{-1} k}) V_{k,g}^* \in \caF.
\end{split}
\end{align}
Then $U_k=0$ or $U_k$ is an isometry.
Furthermore, we have $\sum_k U_k U_k^*=\unit$.
Note that for $k$ with $U_k\neq 0$, we have
\begin{align}
\begin{split}
\gamma_k=\sum_{g\in G}\Ad V_{k,g} \lmk \rho_g\otimes\sigma_{g^{-1}k}\rmk\in O^{(k)}_{(\loz,\ltz)}.
\end{split}
\end{align}
Then we have
\begin{align}
\begin{split}
&\sum_k \Ad U_k\gamma_k
=\sum_{k,g}\Ad U_k \Ad V_{k,g}\lmk \rho_g\otimes\sigma_{g^{-1}k}\rmk
=\sum_{k,g}\Ad \lmk u_g S_{\rho_g}^{(l)\unit}(v_{g^{-1} k})  \rmk\lmk \rho_g\otimes\sigma_{g^{-1}k}\rmk\\
&=\rho\otimes_{C_G} \sigma.
\end{split}
\end{align}
Hence we conclude $\rho\otimes_{C_G} \sigma\in \Obj C_G$.
That (\ref{mortencg}) defines a tensor product of morphisms
follows from the fact $\hat S_\rho^{(l)}$ preserves $\caF$ and by the standard argument (see \cite{BA}).
That it makes $C_G$ a strict tensor category follows as in \cite{BA}.

Next we show the existence of direct sums.
Let us consider objects 
\begin{align}
\begin{split}
\rho=\sum_{g\in G} \Ad u_g\circ\rho_g,\quad
\sigma=\sum_{g\in G}\Ad v_g\circ \sigma_g\in \Obj C_G,
\end{split}
\end{align}
If $u_g=v_g=\unit$, we set $W_g=0$.
For each $g\in G$ with $u_g\neq 0$ or $v_g\neq0$, from Lemma \ref{gds}, there are
a endomorphism $\gamma_g\in O^{(g)}_{(\loz,\ltz)}$ and isometries 
$\hat u_g\in \Mor_G(\rho_g,\gamma_g),\hat v_g\in\Mor_G(\sigma_g, \gamma_g)$
with $\hat u_g \hat u_g^*+\hat v_g \hat v_g^*=\unit$.
We also choose isometries $W_g\in\caF$ for such $g$ so that $\sum_{g\in G}W_g W_g^*=\unit$.
(See from the proof of Lemma \ref{gds} it is possible.)
We set 
\begin{align}
\begin{split}
&\gamma:=\sum_{g\in G}\Ad W_g\circ \gamma_g,\\
&U:=\sum_{g\in G}W_g\lmk\hat u_g u_g^*\rmk,\\
&V:=\sum_{g\in G}W_g\lmk\hat v_g v_g^*\rmk.
\end{split}
\end{align}
Then $\gamma\in \Obj C_G$ by definition and
$U\in \Mor_{C_G}(\rho,\gamma)$,
$V\in \Mor_{C_G}(\sigma,\gamma)$
are isometries such that $UU^*+VV^*=\unit$.
This proves the existence of the direct sum.

Finally, we prove the existence of subobjects.
Let us consider
\begin{align}
\begin{split}
\rho=\sum_{g\in G} \Ad u_g\circ\rho_g\in \Obj C_G,
\end{split}
\end{align}
and a projection $p\in\Mor_{C_G}(\rho,\rho)$.
By definition, we have
\begin{align}
\begin{split}
p \sum_{g\in G} \Ad u_g\circ\rho_g(A)=p\rho(A)=\rho(A)p
=\sum_{h\in G} \Ad u_h\circ\rho_h(A)\cdot p,\quad A\in\caA.
\end{split}
\end{align}
Mutliplying $u_h^*$, $u_g$ from left and right side of this equation,
we obtain
\begin{align}
\begin{split}
u_h^*p u_g \rho_g(A)=\rho_h(A)u_h^* pu_g,\quad  A\in\caA.
\end{split}
\end{align}
Hence we have
\begin{align}
u_h^* p u_g\in \Mor_{C_G}(\rho_g,\rho_h).
\end{align}
By Lemma \ref{gnh},
$u_h^* p u_g=0$ unless $h=g$.
Hence we have
\begin{align}
\begin{split}
p=\sum_{g,h\in G}u_g u_g^* p u_h u_h^* 
=\sum_{g\in G}u_g u_g^* p u_g u_g^*. 
\end{split}
\end{align}
Because $p$ is a projection,
\begin{align}
\begin{split}
\sum_{g\in G}u_g u_g^* p u_g u_g^*
=p=p^2=\sum_{g\in G}u_g u_g^* p u_g u_g^*u_g u_g^* p u_g u_g^*
=\sum_{g\in G}u_g u_g^* p u_g u_g^* p u_g u_g^*.
\end{split}
\end{align}
Multipluing $u_g^*$, $u_g$ from left and right respectively, we obtain
\begin{align}
\begin{split}
 u_g^* p u_g 
=u_g^* p u_g u_g^* p u_g,
\end{split}
\end{align}
i.e., $u_g^* p u_g \in \Mor_{G}(\rho_g,\rho_g)$ is a projection.
Therefore, by Lemma \ref{sg}, if $u_g\neq 0$,
there exists a $\gamma_g\in O^{(g)}_{(\loz,\ltz)}$
and an isometry $v_g\in \Mor_{G}(\gamma_g, \rho_g)$  such that
$v_gv_g^*=u_g^* p u_g$.
Set
\begin{align}
\begin{split}
&\gamma:=\sum_{g\in G}\Ad u_g\circ\gamma_g\in \Obj C_G,\\
&v:=\sum_{g\in G}u_gv_gu_g^*.
\end{split}
\end{align}
Then $v$ is an isometry such that
$v\in \Mor_{C_G}(\gamma,\rho)$ and $vv^*=p$.
This proves the existence of subobjects.
Hence $C_G$ is a strict $C^*$-tensor category.

Next we show that $C_G$ is a strict crossed $G$-category.
If we consider all the objects in $O_G$ and all the morphisms between them, they form a full tensor subcategory of $C_G$
by Lemma \ref{framingo}.
For any $\rho\in O_G$, there exists a unique $g\in G$ such that
$\rho\in O^{(g)}_{(\loz,\ltz)}$:
if $\rho\in O^{(h)}_{(\loz,\ltz)}$ as well, then 
\begin{align}
0\neq \id_\rho=\id_{\caH}\in \Mor_G(\rho,\rho),
\end{align}
hence $h=g$ by Lemma \ref{gnh}.
This defines a map $\partial : O_G\to G$, and it is constant on isomorphism classes because of
 Lemma \ref{gnh}.
 
 We extend the $G $-action on $O_G$ and its morphisms to $C_G$ by the same formula (\ref{ninnaji}).
 Let us check 
  that $\Theta(g)$ is a functor from $C_G$ to $C_G$.
  For any $\rho=\sum_{h\in G} \Ad u_h\circ\rho_h,\in \Obj C_G$,
  we have
  \begin{align}
  \begin{split}
  \Theta(g)(\rho)=\sum_{h\in G} \Ad\lmk R_g u_h R_g^*\rmk \circ\Theta(g)\lmk \rho_h\rmk.
  \end{split}
  \end{align} 
  Because $R_g \caF R_g^*=\caF$ and $\Theta(g)\lmk \rho_h\rmk\in O_G$ by Lemma \ref{tokage},
  this proves $ \Theta(g)(\rho)\in \Obj C_G$.
  It is clear that $\Theta(g)(X)\in \Mor_{C_G}\lmk\Theta(g)(\rho),\Theta(g)(\sigma)\rmk$
  for any $X\in\Mor_{C_G}(\rho,\sigma)$, $\rho,\sigma\in\Obj C_G$ by definition.
  It is also clear that $\Theta(g)(YX)=\Theta(g)(Y)\Theta(g)(X)$ for
   $X\in \Mor_{C_G}(\rho,\rho')$, $Y\in \Mor_{C_G}(\rho',\rho'')$, and $\Theta(g)(\id_\rho)=\id_{\caH}=\id_{\Theta(g)(\rho)}$.
   Hence $\Theta(g)$ is a functor from $C_G$ to $C_G$.

  Next, we see that $\Theta(g)$ is a tensor functor.
Note for the tensor unit $\pi$ of $C_G$ that
\begin{align}
\Theta(g)(\pi)=\Ad R_g\pi\beta_{g^{-1}}=\pi.
\end{align}
We also note that
\begin{align}\label{torio}
\begin{split}
\hat S_{\Theta(g)(\rho)}^{(l)\unit}
=\Ad R_g \hat S_{\rho}^{(l)\unit}\Ad R_g^*,\quad \rho\in \Obj C_G.
\end{split}
\end{align}
In fact, because $\Ad R_g$ preserves $\caB_{l}$,
\begin{align}
\begin{split}
\Ad R_g \hat S_{\rho}^{(l)\unit}\Ad R_g^*
\end{split}
\end{align}
is a well-defined endomorphism of $\caB_{l}$
$\sigma$weak continuous on each $\pi(\caA_{\lm l})''$ with $\lm l\in \Cl$
such that
\begin{align}
\begin{split}
\Ad R_g \hat S_{\rho}^{(l)\unit}\Ad R_g^*\pi
=\Ad R_g\rho\beta_g^{-1}=\Theta(g)(\rho).
\end{split}
\end{align}
Because $\hat S_{\Theta(g)\rho}^{(l)\unit}$ also satisfies the same property, 
we get (\ref{torio}).

From (\ref{torio}), we obtain
\begin{align}
\begin{split}
&\Theta(g)(\rho)\otimes_{C_G}\Theta(g)(\sigma)
=\hat S_{\Theta(g)(\rho)}^{(l)\unit}\hat S_{\Theta(g)(\sigma)}^{(l)\unit}\pi
=\Ad R_g \hat S_{\rho}^{(l)\unit}\Ad R_g^*\Ad R_g \hat S_{\sigma}^{(l)\unit}\Ad R_g^*\pi\\
&=\Ad R_g\hat S_{\rho}^{(l)\unit} \hat S_{\sigma}^{(l)\unit} \Ad R_g^*\pi
=\Ad R_g\lmk \rho\otimes_{C_G} \sigma \rmk\beta_{g^{-1}}
=\Theta(g)\lmk \rho\otimes_{C_G} \sigma\rmk,
\end{split}
\end{align}
for all $\rho,\sigma\in\Obj C_G$.
Then we have
\begin{align}
\begin{split}
&\varphi_0:=\id_{\caH}\in \Mor_{C_G}\lmk \pi,\Theta(g)(\pi)\rmk,\\
&\varphi_2(\rho,\sigma):=\id_{\caH}
\in \Mor_{C_G}\lmk \Theta(g)(\rho)\otimes_{C_G}\Theta(g)(\sigma), \Theta(g)\lmk \rho\otimes_{C_G} \sigma\rmk\rmk,\quad
\rho,\sigma\in \Obj C_G.
\end{split}
\end{align}
We claim that $(\Theta(g), \varphi_0, \varphi_2)$ gives a tensor functor from $C_G$ to $C_G$.
Note that $\varphi_2$ is natural because
\begin{align}
\begin{split}
&\Theta(g)(X)\otimes_{C_G} \Theta(g)(Y)
=\Theta(g)(X)\hat S_{\Theta(g)(\rho)}^{(l)\unit}\lmk \Theta(g)(Y)\rmk\\
&=\Ad R_g\lmk X\rmk
\Ad R_g \hat S_{\rho}^{(l)\unit} \Ad R_g^*
\Ad R_g\lmk Y\rmk\\
&=\Ad R_g\lmk X \hat S_{\rho}^{(l)\unit}  (Y)\rmk
=\Theta(g) \lmk X\otimes_{C_G} Y\rmk
\end{split}
\end{align}
$\rho,\rho',\sigma,\sigma'\in \Obj C$, $X\in \Mor(\rho,\rho')$,
$Y\in \Mor(\sigma,\sigma')$.
That $\varphi_0$, $\varphi_2$ are consistent with associativity morphisms and left/right constraint is
trivial because all the involved morphisms are $\id_{\caH}$.
Hence $\Theta(g)$ is a tensor functor from $C_G$ to $C_G$.

From the definition, it is clear that
the composition of tensor functors $(\Theta(g), \id_{\caH}, \id_{\caH})$,
$(\Theta(h), \id_{\caH}, \id_{\caH})$
is $(\Theta(gh), \id_{\caH}, \id_{\caH})$.
It in particular tells us that $(\Theta(g), \id_{\caH}, \id_{\caH})$ is an auto-equivalence of
$C$, and 
$\Theta : G\ni g\mapsto \Theta(g)\in \Aut C_G$ is
a group homomorphism. 
 
 By Lemma \ref{framingo}, we have
 $\partial \lmk \rho\otimes\sigma\rmk=\partial\rho\partial \sigma$ for
 $\rho,\sigma\in O_G$.
 By Lemma \ref{tokage}, we have 
 $\Theta(g)\lmk O^{(h)}_{(\loz,\ltz)}\rmk\subset O^{(ghg^{-1})}_{(\loz,\ltz)}$.
 Hence $C_G$ is a strict crossed $G$-category.
\end{proof}

\section{Proof of Lemmas in section \ref{gato}}\label{brprf}
\begin{lem}\label{lem35}Consider setting in the subsection \ref{sec:qss}
and assume Assumption \ref{a1l}.
Let $(\lo,\lt)\in PC$, $\lm r\in \Cr$
such that $\arg \lmk \lmk \lm r\rmk^c\rmk _\varepsilon \subset \arg\lo$.
Set
\begin{align}
\begin{split}
\ld_i(s):=\ld_i+s\bm e_0,\quad i=1,2,\quad s\ge 0.
\end{split}
\end{align}
Then for any $\sigma\in O^{(k)}_{(\loz,\ltz)}$ and
$V_{{\sigma},{(\lo(s),\lt(s))}}^{(k)} \in \Vbk{\sigma}{(\lo(s),\lt(s))}^{(k)}$,
we have
\begin{align}
\lim_{s\to\infty}
\lV\left.
\lmk S_\sigma^{(l)V_{{\sigma},{(\lo(s),\lt(s))}}^{(k)}}
-\Ad R_k
\rmk\right\vert_{\pi\lmk\caA_{\ld_r } \rmk'}
\rV=0.
\end{align}
\end{lem}
\begin{proof}
Note that $\pi\lmk\caA_{\ld_r } \rmk'\subset \caB_l$ by the approximate Haag duality.
If $\lm r=\ld_{(a,0),0, \varphi}\in \Cr$, then $\lm r^c=\overline{\ld_{(a,0),\pi , \pi-\varphi}}$.
Then there exists a $s_1\in\bbR_+$ such that
$\overline{ \ld_{({a+u},0), \pi,  {\pi-\varphi+\frac \varepsilon 2}} }\subset \lo(s)$ for all $s, u\in\bbR_+$ with
$s-u\ge s_1$.

By Assumption \ref{a1l}, for any $u\ge R_{\varphi,\frac \varepsilon 4}$,
we have
\begin{align}
{\pi\lmk\caA_{\lm r} \rmk'}\subset_{2 f_{\varphi,\frac{\varepsilon}4,\frac \varepsilon 4}(u)} \pi\lmk 
 \caA_{
  \overline{ \ld_{({a+u},0), \pi,  {\pi-\varphi+\frac \varepsilon 2}} }
 }\rmk''.
\end{align}
By Lemma \ref{ookami}, 
if $\overline{ \ld_{({a+u},0), \pi,  {\pi-\varphi+\frac \varepsilon 2}} }\subset \lo(s)$,
(hence if $s-u\ge s_1$),
we have 
\begin{align}
\left.
\lmk
 S_\sigma^{(l)V_{{\sigma},{(\lo(s),\lt(s))}}^{(k)}}-\Ad R_k
\rmk
\right\vert_{\pi\lmk 
 \caA_{
  \overline{ \ld_{({a+u},0), \pi,  {\pi-\varphi+\frac \varepsilon 2}} }
 }\rmk''}=0.
\end{align}
Therefore, if $s-u\ge s_1$, we have
\begin{align}
\lV
\left.
\lmk
S_\sigma^{(l)V_{{\sigma},{(\lo(s),\lt(s))}}^{(k)}}-\Ad R_k
\rmk
\right\vert_{{\pi\lmk\caA_{\lm r} \rmk'}}\rV
\le 4 f_{\varphi,\frac{\varepsilon}4,\frac \varepsilon 4}(u).
\end{align}
This completes the proof.
\end{proof}
\begin{lem}\label{ou}Consider setting in the subsection \ref{sec:qss}.
Assume Assumption \ref{a1r}, Assumption \ref{a1l}.
Let $(\lo,\lt )\in PC$, $\lm {l}\in \Cl$
with $\arg\lmk \lmk \ld_l\rmk ^c \rmk_\varepsilon \subset\arg\lt$
for some $\varepsilon>0$.
Then the following hold.
\begin{description}
\item[(i)]
Set
\begin{align}
\begin{split}
\ld_i(t):=\ld_i-t\bm e_0,\quad i=1,2,\quad t\ge 0.
\end{split}
\end{align}
For any $\rho\in O^{(g)}_{(\loz,\ltz)}$ and
$V_{{\rho},{(\lo(t),\lt(t))}}^{(g)} \in \Vbk{\rho}{(\lo(t),\lt(t))}^{(g)}$,
we have
\begin{align}
\lim_{t\to\infty}
\lV\left.
\lmk S_\rho^{(l)V_{{\rho},{(\lo(t),\lt(t))}}^{(g)}}
-\id
\rmk\right\vert_{\pi\lmk\caA_{\ld_l } \rmk'\cap \caF}
\rV=0.
\end{align}
\item[(ii)] 
Set $\lm l(t):=\lm l+t \bm e_0$
For any $\rho\in O^{(g)}_{(\loz,\ltz)}$ and
$V_{{\rho},{(\lo,\lt)}}^{(g)} \in \Vbk{\rho}{(\lo,\lt)}^{(g)}$,
we have
\begin{align}
\lim_{t\to\infty}
\lV\left.
\lmk S_\rho^{(l)V_{{\rho},{(\lo,\lt )}}^{(g)}}
-\id
\rmk\right\vert_{\pi\lmk\caA_{\lm l(t)} \rmk'\cap \caF}
\rV=0.
\end{align}

\end{description}
\end{lem}
\begin{proof}
The cone $\lm {l}$ is of the form $\lm {l}=\overline{\ld_{(a,0),\pi, \pi- \varphi}}$,
and $\lmk \lm l\rmk^c=\ld_{(a,0),0, \varphi}\in \Cr$.\\
(i)For each $s\ge R_{\varphi,\frac \varepsilon 4}$,
by Assumption \ref{a1r}, there exists $\tilde W_s\in\caU(\caF)$ such that 
\begin{align}
&{\pi\lmk\caA_{\ld_l } \rmk'}
\subset
\Ad\lmk\tilde W_s\rmk
\lmk
\pi\lmk
\caA_{\ld_{(a-s,0),0, \varphi+\frac \varepsilon 2}}
\rmk''
\rmk,\\
&\lV
\tilde W_s-\unit
\rV\le {f_{\varphi, \frac \varepsilon 4,\frac \varepsilon 4 }(s)}
\end{align}
For each $s\ge 0$, there exists a $t_0(s)\ge 0$ such that
$\ld_{(a-s,0),0, \varphi+\frac \varepsilon 2}\subset \lt(t)$ for all $t\ge t_0(s)$.
Applying Lemma \ref{ookami} we obtain
\begin{align}
\begin{split}
\left.S_\rho^{(l)V_{{\rho},{(\lo(t),\lt(t))}}^{(g)}}\right\vert_{\pi\lmk
\caA_{\lt(t)}\rmk''\cap \caF}
=\id.
\end{split}
\end{align}
Now for any $x\in{\pi\lmk\caA_{\ld_l } \rmk'\cap \caF}$,
 $s\ge R_{\varphi,\frac \varepsilon 2}$ and 
 $t\ge t_0(s)$,
 we have
 \begin{align}
 \begin{split}
 \Ad(\tilde W_s^*)(x)\in 
 \pi\lmk
\caA_{\ld_{(a-s,0),0, \varphi+\frac \varepsilon 2}}
\rmk''\cap\caF
\subset {\pi\lmk
\caA_{\lt(t)}\rmk''\cap \caF}.
 \end{split}
 \end{align}
Therefore, we have 
\begin{align}\label{nagasaki}
\begin{split}
&\lV \lmk S_\rho^{(l)V_{{\rho},{(\lo(t),\lt(t))}}^{(g)}}(x)-\id\rmk(x)\rV\\
&\le
\lV \lmk S_\rho^{(l)V_{{\rho},{(\lo(t),\lt(t))}}^{(g)}}-\id\rmk
\lmk \Ad(\tilde W_s^*)(x)\rmk\rV
+2\lV \Ad(\tilde W_s^*)(x)-x\rV
\le {4f_{\varphi, \frac \varepsilon 2,\frac \varepsilon 2 }(s)}\lV x\rV.
\end{split}
\end{align}
This proves (i).\\
(ii)
By Lemma \ref{ookami}
\begin{align}\label{taiwan}
\left.S_\rho^{(l)V_{{\rho},{(\lo,\lt )}}^{(g)}}\right\vert_{\pi\lmk \caA_{\lt}\rmk''\cap\caF}=\id.
\end{align}
By Assumption \ref{a1r}, for any $t\ge 2R_{\varphi,\frac \varepsilon 4}^{(r)}$,
there exists $W_t\in \caU(\fbd)$ such that
\begin{align}\label{ibaraki}
&{\pi\lmk\caA_{\lm l(t)} \rmk'}
\subset
\Ad\lmk W_t\rmk\lmk
\pi\lmk
\caA_{\ld_{ {(a+\frac t2,0)}, 0, {\varphi+\frac \varepsilon 2}}}
\rmk''
\rmk,\\
&\lV
W_t-\unit
\rV\le {f_{\varphi, \frac \varepsilon 4,\frac \varepsilon 4 }\lmk \frac t2\rmk}.
\end{align}
Furthermore, there exists $t_2\ge 0$ such that
$\ld_{ {(a+\frac t2,0)}, 0, {\varphi+\varepsilon}}\subset \lt
$ for all $t\ge t_2$.
If $t\ge \max\{t_2,2R^{(r)}_{\varphi,\frac \varepsilon 4}\}$, then 
for any 
$x\in {{\pi\lmk\caA_{\lm l(t)} \rmk'}\cap\fbd}$,
we have 
\begin{align}
\begin{split}
\Ad(W_t^*)(x)\in 
\pi\lmk
\caA_{\ld_{ {(a+\frac t2,0)}, 0, {\varphi+\frac \varepsilon 2}}}
\rmk''\cap \caF
\subset \pi\lmk\caA_{\lt}\rmk''\cap\caF.
\end{split}
\end{align}
From this we obtain (ii).
\end{proof}

 
\begin{lem}\label{kappa}Consider setting in the subsection \ref{sec:qss}.
Assume Assumption \ref{a1r} and Assumption \ref{a1l}.
Let $g,h\in G$ and $\rho,\rho'\in O^{(g)}_{(\loz,\ltz)}$, $\sigma,\sigma'\in O^{(h)}_{(\loz,\ltz)}$.
Let \[
(\lm {1\rho},\lm {2\rho}), (\lm {1\rho'}', \lm {2\rho'}'), 
(\lm {1\sigma},\lm {2\sigma}), (\lm {1\sigma'}', \lm {2\sigma'}')\in PC
\]
 such that
\begin{align}\label{hoshi}
\{(\lm {1\rho},\lm {2\rho}), (\lm {1\rho'}', \lm {2\rho'}')\}
\leftarrow 
\{
(\lm {1\sigma},\lm {2\sigma}), (\lm {1\sigma'}', \lm {2\sigma'}')
\}.
\end{align}
We set
\begin{align}
\begin{split}
&\lm{ i\rho}(t):=\lm {i\rho}-t\bm e_0,\quad \lm {i\rho'}'(t'):=\lm {i\rho'}'-t'\bm e_0,\\
&\lm {i\sigma}(s):=\lm {i\sigma}+s\bm e_0,\quad \lm {i\sigma'}'(s'):=\lm {i\sigma'}'+s'\bm e_0,
\end{split}
\end{align}
with $i=1,2$, $t,s,t',s'\ge 0$.
Let
\begin{align}
\begin{split}
\Vrl\rho{\lm {1\rho}(t)}{\lm {2\rho}(t)}g\in \Vbu\rho{\lm {1\rho}(t)}{\lm {2\rho}(t)}g,\quad
\Vrl{\rho'}{{\lm {1\rho'}}'(t')}{{\lm {2\rho'}}'(t')} g\in \Vbu{\rho'}{{\lm {1\rho'}}'(t')}{{\lm {2\rho'}}'(t')}g,\\
\Vrl\sigma{\lm {1\sigma}(s)}{\lm {2\sigma}(s)}h\in \Vbu\sigma{\lm {1\sigma}(s)}{\lm {2\sigma}(s)}h,\quad \Vrl{\sigma'}{\lm {1\sigma'}'(s')}{\lm {2\sigma'}'(s')}h\in \Vbu{\sigma'}{\lm {1\sigma'}'(s')}{\lm {2\sigma'}'(s')}h
\end{split}
\end{align}
for $t,t',s,s'\ge 0$.
Then for any
\begin{align}
\begin{split}
X_\rho^{(t,t')}\in \Mor_{\caB_l}\lmk \Srl\rho{\lm {1\rho}(t)}{\lm {2\rho}(t)}g,
\Srl{\rho'}{{\lm {1\rho'}}'(t')}{{\lm {2\rho'}}'(t')} g
\rmk,\\
X_\sigma^{(s,s')}\in
\Mor_{\caB_l}\lmk
\Srl\sigma{\lm {1\sigma}(s)}{\lm {2\sigma}(s)}h,
\Srl{\sigma'}{\lm {1\sigma'}'(s')}{\lm {2\sigma'}'(s')}h
\rmk
\end{split}
\end{align}
with $\lV X_\rho^{(t,t')}\rV, \lV X_\sigma^{(s,s')}\rV\le 1$,
we have 
\begin{align}
\begin{split}
\lim_{s,t,s',t'\to\infty}\lV
X_\rho^{(t,t')}\otimes_{\caB_l} X_\sigma^{(s,s')}
-X_\sigma^{(s,s')}\otimes_{\caB_l} \Theta\lmk h^{-1}\rmk\lmk X_\rho^{(t,t')}\rmk
\rV=0.
\end{split}
\end{align}
Here, the tensor product is taken with 
\begin{align}
\begin{split}
\Theta\lmk h^{-1}\rmk\lmk X_\rho^{(t,t')}\rmk
\in 
\Mor_{\caB_l}\lmk S_{\Theta(h^{-1})(\rho)}^{(l),\Theta(h^{-1})\lmk V_{{\rho},{{\lm {1\rho}}(t)}{{\lm {2\rho}}(t)}}^{(g)}\rmk},
S_{\Theta(h^{-1})(\rho')}^{(l),\Theta(h^{-1})\lmk V_{{\rho'},{{\lm {1\rho'}}'(t')}{{\lm {2\rho'}}'(t')}}^{(g)}\rmk}\rmk.
\end{split}
\end{align}
\end{lem}
\begin{proof}
Note from Lemma \ref{ookami} that
\begin{align}\label{xloc}
\begin{split}
X_\rho^{(t,t')}&\in 
\Mor_{\caB_l}\lmk \Srl\rho{\lm {1\rho}(t)}{\lm {2\rho}(t)}g,
\Srl{\rho'}{{\lm {1\rho'}}'(t')}{{\lm {2\rho'}}'(t')} g
\rmk,\\
&\subset \pi\lmk\caA_{{\lm {2\rho}(t)}\cap {{\lm {2\rho'}}'(t')}}\rmk'\cap \caF\\
X_\sigma^{(s,s')}&\in
\Mor_{\caB_l}\lmk
\Srl\sigma{\lm {1\sigma}(s)}{\lm {2\sigma}(s)}h,
\Srl{\sigma'}{\lm {1\sigma'}'(s')}{\lm {2\sigma'}'(s')}h
\rmk\\
&\subset\pi\lmk\caA_{{\lm {1\sigma}(s)}\cap {\lm {1\sigma'}'(s')}}\rmk'\cap \caF
\subset \pi\lmk\caA_{{\lm {1\sigma}}\cap {\lm {1\sigma'}'}}\rmk'\cap \caF.
\end{split}
\end{align}
Because $\Theta(h)$ doesn't change local algebras and $\caF$, we also have
\begin{align}
\begin{split}
\Theta\lmk h^{-1}\rmk\lmk X_\rho^{(t,t')}\rmk
\in \pi\lmk\caA_{{\lm {2\rho}(t)}\cap {{\lm {2\rho'}}'(t')}}\rmk'\cap \caF
\subset \pi\lmk\caA_{{\lm {2\rho}}\cap {{\lm {2\rho'}}'}}\rmk'\cap \caF.
\end{split}
\end{align}
We have
\begin{align}\label{geji}
\begin{split}
&\lV
X_\rho^{(t,t')}\otimes_{\caB_l} X_\sigma^{(s,s')}
-X_\sigma^{(s,s')}\otimes_{\caB_l} \Theta\lmk h^{-1}\rmk\lmk X_\rho^{(t,t')}\rmk
\rV\\
&=\lV
X_\rho^{(t,t')}\Srl\rho{\lm {1\rho}(t)}{\lm {2\rho}(t)}g\lmk X_\sigma^{(s,s')}\rmk
-X_\sigma^{(s,s')} 
\Srl\sigma{\lm {1\sigma}(s)}{\lm {2\sigma}(s)}h\lmk \Theta\lmk h^{-1}\rmk\lmk X_\rho^{(t,t')}\rmk\rmk
\rV.
\end{split}
\end{align}
By Lemma \ref{lem35},
we have
\begin{align}
\lim_{s\to\infty}
\sup_{t,t'}\lV
\lmk S_\sigma^{(l)V_{{\sigma},{(\lm {1\sigma}(s),\lm {2\sigma}(s))}}^{(h)}}
-\Ad R_h
\rmk\lmk \Theta\lmk h^{-1}\rmk\lmk X_\rho^{(t,t')}\rmk\rmk
\rV=0.
\end{align}
By Lemma \ref{ou}, we have
\begin{align}
\lim_{t\to\infty}\sup_{s,s'}
\lV
\lmk S_\rho^{(l)V_{{\rho},{(\lm {1\rho}(t),\lm {2\rho}(t))}}^{(g)}}
-\id
\rmk\lmk X_\sigma^{(s,s')}\rmk
\rV=0.
\end{align}
Substituting these to (\ref{geji}), and the locality (\ref{xloc})
combined with  Assumption \ref{a1l} Assumption \ref{a1r} and (\ref{hoshi})
we obtain
\begin{align}
\begin{split}
&\lim_{t,t',s,s'\to\infty}\lV
X_\rho^{(t,t')}\otimes_{\caB_l} X_\sigma^{(s,s')}
-X_\sigma^{(s,s')}\otimes_{\caB_l} \Theta\lmk h^{-1}\rmk\lmk X_\rho^{(t,t')}\rmk
\rV\\
&=\lim_{t,t',s,s'\to\infty}\lV
X_\rho^{(t,t')} X_\sigma^{(s,s')}
-X_\sigma^{(s,s')} 
\Ad R_h\lmk \Theta\lmk h^{-1}\rmk\lmk X_\rho^{(t,t')}\rmk\rmk
\rV=0.
\end{split}
\end{align}
\end{proof}
\begin{lem}\label{lem38}Consider setting in the subsection \ref{sec:qss}.
Assume Assumption \ref{a1l} and Assumption \ref{a1r}.
Let $g,h\in G$ and $\rho\in O^{(g)}_{(\loz,\ltz)}$, $\sigma\in O^{(h)}_{(\loz,\ltz)}$.
Let $(\lm {1\rho},\lm {2\rho}), 
(\lm {1\sigma},\lm {2\sigma})\in PC$ such that
$\{(\lm {1\rho},\lm {2\rho})\}
\leftarrow 
\{
(\lm {1\sigma},\lm {2\sigma})
\}$.
We set
\begin{align}
\begin{split}
&\lm{ i\rho}(t):=\lm {i\rho}-t\bm e_0,\quad 
\lm {i\sigma}(s):=\lm {i\sigma}+s\bm e_0,
\end{split}
\end{align}
with $i=1,2$, $t,s\ge 0$.
Let
\begin{align}
\begin{split}
\Vrl\rho{\lm {1\rho}(t)}{\lm {2\rho}(t)}g\in \Vbu\rho{\lm {1\rho}(t)}{\lm {2\rho}(t)}g,\quad
\Vrl\sigma{\lm {1\sigma}(s)}{\lm {2\sigma}(s)}h\in \Vbu\sigma{\lm {1\sigma}(s)}{\lm {2\sigma}(s)}h
\end{split}
\end{align}
for $t,s\ge 0$.
Then the limit 
\begin{align}\label{fune}
\begin{split}
&\epsilon_G\lmk \rho,\sigma\rmk\\
&:=\lim_{t,s\to\infty}
\lmk
\Vrl\sigma{\lm {1\sigma}(s)}{\lm {2\sigma}(s)}h\otimes_{\caB_l}
\Theta(h^{-1})\lmk \Vrl\rho{\lm {1\rho}(t)}{\lm {2\rho}(t)}g\rmk\rmk^*
\lmk
\Vrl\rho{\lm {1\rho}(t)}{\lm {2\rho}(t)}g
\otimes_{\caB_l}
\Vrl\sigma{\lm {1\sigma}(s)}{\lm {2\sigma}(s)}h
\rmk\\
&\in \caU(\caF)
\end{split}
\end{align}
exits and it is independent of the choices of
$(\lm {1\rho},\lm {2\rho})$, 
$(\lm {1\sigma},\lm {2\sigma})$,
$\Vrl\rho{\lm {1\rho}(t)}{\lm {2\rho}(t)}g$, $\Vrl\sigma{\lm {1\sigma}(s)}{\lm {2\sigma}(s)}h$.
Here the tensor product is taken for
\begin{align}
\begin{split}
\Vrl\rho{\lm {1\rho}(t)}{\lm {2\rho}(t)}g\in
\Mor_{\caB_l}\lmk
S_\rho^{(l)\unit}, S_{\rho}^{(l)\Vrl\rho{\lm {1\rho}(t)}{\lm {2\rho}(t)}g}
\rmk,\quad
\Vrl\sigma{\lm {1\sigma}(s)}{\lm {2\sigma}(s)}h\in
\Mor_{\caB_l}\lmk
S_\sigma^{(l)\unit}, S_{\rho}^{(l), \Vrl\sigma{\lm {1\sigma}(s)}{\lm {2\sigma}(s)}h}
\rmk
\end{split}
\end{align}

\end{lem}
\begin{proof}
If $(\lm {1\rho},\lm {2\rho}), (\lm {1\rho}', \lm {2\rho}'), 
(\lm {1\sigma},\lm {2\sigma}), (\lm {1\sigma}', \lm {2\sigma}')\in PC$
satisfy (\ref{hoshi}),
then we have
\begin{align}\label{dai}
\begin{split}
&\lV
\begin{gathered}
\lmk
\Vrl\sigma{\lm {1\sigma}(s)}{\lm {2\sigma}(s)}h\otimes_{\caB_l}
\Theta(h^{-1})\lmk \Vrl\rho{\lm {1\rho}(t)}{\lm {2\rho}(t)}g\rmk\rmk^*
\lmk
\Vrl\rho{\lm {1\rho}(t)}{\lm {2\rho}(t)}g
\otimes_{\caB_l}
\Vrl\sigma{\lm {1\sigma}(s)}{\lm {2\sigma}(s)}h
\rmk\\
-\lmk
\Vrl\sigma{\lm {1\sigma}'(s')}{\lm {2\sigma}'(s')}h\otimes_{\caB_l}
\Theta(h^{-1})\lmk \Vrl\rho{\lm {1\rho}'(t')}{\lm {2\rho}'(t')}g\rmk\rmk^*
\lmk
\Vrl\rho{\lm {1\rho}'(t')}{\lm {2\rho}'(t')}g
\otimes_{\caB_l}
\Vrl\sigma{\lm {1\sigma}'(s')}{\lm {2\sigma}'(s')}h
\rmk
\end{gathered}
\rV\\
&=
\lV
\begin{gathered}
\lmk
\Vrl\sigma{\lm {1\sigma}'(s')}{\lm {2\sigma}'(s')}h\otimes_{\caB_l}
\Theta(h^{-1})\lmk \Vrl\rho{\lm {1\rho}'(t')}{\lm {2\rho}'(t')}g\rmk\rmk
\lmk
\lmk \Vrl\sigma{\lm {1\sigma}(s)}{\lm {2\sigma}(s)}h\rmk^*\otimes_{\caB_l}
\lmk \Theta(h^{-1})\lmk \Vrl\rho{\lm {1\rho}(t)}{\lm {2\rho}(t)}g\rmk^*\rmk\rmk
\\
-
\lmk
\Vrl\rho{\lm {1\rho}'(t')}{\lm {2\rho}'(t')}g
\otimes_{\caB_l}
\Vrl\sigma{\lm {1\sigma}'(s')}{\lm {2\sigma}'(s')}h
\rmk
\lmk
\lmk \Vrl\rho{\lm {1\rho}(t)}{\lm {2\rho}(t)}g\rmk^*
\otimes_{\caB_l}
\lmk \Vrl\sigma{\lm {1\sigma}(s)}{\lm {2\sigma}(s)}h\rmk^*
\rmk
\end{gathered}
\rV\\
&=
\lV
\begin{gathered}
\lmk
\Vrl\sigma{\lm {1\sigma}'(s')}{\lm {2\sigma}'(s')}h
\lmk \Vrl\sigma{\lm {1\sigma}(s)}{\lm {2\sigma}(s)}h\rmk^*
\otimes_{\caB_l}
\Theta(h^{-1})\lmk \Vrl\rho{\lm {1\rho}'(t')}{\lm {2\rho}'(t')}g\lmk \Vrl\rho{\lm {1\rho}(t)}{\lm {2\rho}(t)}g\rmk^*
\rmk\rmk
\\
-
\lmk
\Vrl\rho{\lm {1\rho}'(t')}{\lm {2\rho}'(t')}g\lmk \Vrl\rho{\lm {1\rho}(t)}{\lm {2\rho}(t)}g\rmk^*
\otimes_{\caB_l}
\Vrl\sigma{\lm {1\sigma}'(s')}{\lm {2\sigma}'(s')}h\lmk \Vrl\sigma{\lm {1\sigma}(s)}{\lm {2\sigma}(s)}h\rmk^*
\rmk
\end{gathered}
\rV.
\end{split}
\end{align}
Because
\begin{align}
\begin{split}
\Vrl\rho{\lm {1\rho}'(t')}{\lm {2\rho}'(t')}g\lmk \Vrl\rho{\lm {1\rho}(t)}{\lm {2\rho}(t)}g\rmk^*
\in 
 \Mor_{\caB_l}\lmk \Srl\rho{\lm {1\rho}(t)}{\lm {2\rho}(t)}g,
\Srl{\rho}{{\lm {1\rho}}'(t')}{{\lm {2\rho}}'(t')} g
\rmk,\\
\Vrl\sigma{\lm {1\sigma}'(s')}{\lm {2\sigma}'(s')}h\lmk \Vrl\sigma{\lm {1\sigma}(s)}{\lm {2\sigma}(s)}h\rmk^*
\in\Mor_{\caB_l}\lmk
\Srl\sigma{\lm {1\sigma}(s)}{\lm {2\sigma}(s)}h,
\Srl{\sigma}{\lm {1\sigma}'(s')}{\lm {2\sigma}'(s')}h
\rmk,
\end{split}
\end{align}
 by Lemma \ref{kappa}, the left hand side of (\ref{dai}) converges to $0$
 as $t,s,t',s'\to \infty$.
 This proves the existence of the limit (\ref{fune}), and its value is independent
of the choices of
$\Vrl\rho{\lm {1\rho}(t)}{\lm {2\rho}(t)}g$, $\Vrl\sigma{\lm {1\sigma}(s)}{\lm {2\sigma}(s)}h$.
It also proves that if  $(\lm {1\rho},\lm {2\rho}), (\lm {1\rho'}', \lm {2\rho'}'), 
(\lm {1\sigma},\lm {2\sigma}), (\lm {1\sigma'}', \lm {2\sigma'}')\in PC$
satisfy (\ref{hoshi}), then 
 the corresponding
 limits are the same.
 For general  $(\lm {1\rho},\lm {2\rho}), 
(\lm {1\sigma},\lm {2\sigma}),  (\lm {1\rho}', \lm {2\rho}'),(\lm {1\sigma}', \lm {2\sigma}')\in PC$,
we can construct a finite sequence 
$\lmk (\lm {1\rho}^{(i)},\lm {2\rho}^{(i)}), (\lm {1\sigma}^{(i)},\lm {2\sigma}^{(i)})\rmk
\in PC^{\times 2}$ satisfying (\ref{hoshi})
\begin{align}
\{(\lm {1\rho}^{(i)},\lm {2\rho}^{(i)}), 
(\lm {1\rho}^{(i+1)},\lm {2\rho}^{(i+1)})
\}\leftarrow
\{(\lm {1\sigma}^{(i)},\lm {2\sigma}^{(i)}), (\lm {1\sigma}^{(i+1)},\lm {2\sigma}^{(i+1)})\}
\end{align}
connecting $(\lm {1\rho},\lm {2\rho}), (\lm {1\sigma},\lm {2\sigma})$ and 
$(\lm {1\rho}', \lm {2\rho}'),(\lm {1\sigma}', \lm {2\sigma}')$.
Hence the independence of the limit is proven
for general case.
\end{proof}

\begin{lem}\label{tako}Consider setting in the subsection \ref{sec:qss}.
Assume Assumption \ref{a1l} and Assumption \ref{a1r}.
Let $g,h\in G$ and $\rho\in O^{(g)}_{(\loz,\ltz)}$, $\sigma\in O^{(h)}_{(\loz,\ltz)}$.
Let $(\lm{1\rho},\lm{2\rho}), (\lm{1\rho},\lm{2\rho})\in PC$.
Then we have the following.
\begin{description}
\item[(i)]
If $\{(\loz,\ltz)\}\leftarrow \{ (\lm {1\sigma},\lm {2\sigma})\}$,
then
\begin{align}
\begin{split}
\epsilon_G(\rho,\sigma)=
\lim_{s\to\infty} 
\lmk
\Vrl\sigma{\lm {1\sigma}(s)}{\lm {2\sigma}(s)}h
\rmk^*
S_\rho^{(l) \unit}
\lmk
\Vrl\sigma{\lm {1\sigma}(s)}{\lm {2\sigma}(s)}h
\rmk
\end{split}
\end{align}
for $\lm {i\sigma}(s):=\lm {i\sigma}+s\bm e_0$, $i=1,2$, $s\ge 0$ and
any 
$\Vrl\sigma{\lm {1\sigma}(s)}{\lm {2\sigma}(s)}h\in \Vbu\sigma{\lm {1\sigma}(s)}{\lm {2\sigma}(s)}h$.
\item[(ii)]
If $\{(\lm {1\rho},\lm {2\rho})\}\leftarrow \{  (\loz,\ltz)\}$,
then
\begin{align}
\begin{split}
\epsilon_G(\rho,\sigma)=
\lim_{t\to\infty} 
S_\sigma^{(l) \unit}\lmk\Theta(h^{-1})\lmk
\lmk
\Vrl\rho{\lm {1\rho}(t)}{\lm {2\rho}(t)}g
\rmk^*\rmk
\rmk
\Vrl\rho{\lm {1\rho}(t)}{\lm {2\rho}(t)}g
\end{split}
\end{align}
for $\lm {i\sigma}(s):=\lm {i\sigma}+s\bm e_0$, $i=1,2$, $s\ge 0$ and
any 
$\Vrl\rho{\lm {1\rho}(t)}{\lm {2\rho}(t)}g\in \Vbu\rho{\lm {1\rho}(t)}{\lm {2\rho}(t)}g$.

\end{description}
\end{lem}
\begin{proof}
(i)Choose $(\lm {1\rho},\lm {2\rho})\in PC$ so that
$\{(\lm {1\rho},\lm {2\rho}), (\loz,\ltz)\}\leftarrow \{ (\lm {1\sigma},\lm {2\sigma})\}$.
With the notation in Lemma \ref{lem38}
we have
\begin{align}\label{kujira}
\begin{split}
&\epsilon_G\lmk \rho,\sigma\rmk\\
&:=\lim_{t,s\to\infty}
\lmk \Vrl\sigma{\lm {1\sigma}(s)}{\lm {2\sigma}(s)}h\rmk^*
\Srl\sigma{\lm {1\sigma}(s)}{\lm {2\sigma}(s)}h\lmk
\Theta(h^{-1})\lmk \Vrl\rho{\lm {1\rho}(t)}{\lm {2\rho}(t)}g\rmk^*
\rmk\\
&\quad\quad\quad\quad \Vrl\rho{\lm {1\rho}(t)}{\lm {2\rho}(t)}g
S_\rho^{(l)\unit}\lmk 
\Vrl\sigma{\lm {1\sigma}(s)}{\lm {2\sigma}(s)}h
\rmk
\end{split}
\end{align}
Because
\begin{align}
\begin{split}
\Theta(h^{-1})\lmk \lmk \Vrl\rho{\lm {1\rho}(t)}{\lm {2\rho}(t)}g\rmk^*\rmk
\in \pi\lmk\caA_{{\lm {2\rho}(t)}\cap \ltz}\rmk'\cap\caF
\end{split}
\end{align}
we have
\begin{align}
\begin{split}
\lim_{s\to\infty}\sup_t
\lV
\Srl\sigma{\lm {1\sigma}(s)}{\lm {2\sigma}(s)}h\lmk
\Theta(h^{-1})\lmk \Vrl\rho{\lm {1\rho}(t)}{\lm {2\rho}(t)}g\rmk^*\rmk
-\Ad R^h\lmk
\Theta(h^{-1})\lmk \Vrl\rho{\lm {1\rho}(t)}{\lm {2\rho}(t)}g\rmk^*
\rmk
\rV=0,
\end{split}
\end{align}by Lemma \ref{lem35}
Substituting this to (\ref{kujira}), we obtain (i).
\\
(ii)
Choose $(\lm {1\sigma},\lm {2\sigma})\in PC$ so that
$\{(\lm {1\rho},\lm {2\rho})\}\leftarrow \{ (\lm {1\sigma},\lm {2\sigma}), (\loz,\ltz)\}$.
Then with the notation in Lemma \ref{lem38}, we have
\begin{align}\label{kuma}
\begin{split}
&\epsilon_G\lmk \rho,\sigma\rmk\\
&\lim_{t,s\to\infty}
\lmk
\Vrl\sigma{\lm {1\sigma}(s)}{\lm {2\sigma}(s)}h\otimes_{\caB_l}
\Theta(h^{-1})\lmk \Vrl\rho{\lm {1\rho}(t)}{\lm {2\rho}(t)}g\rmk\rmk^*
\lmk
\Vrl\rho{\lm {1\rho}(t)}{\lm {2\rho}(t)}g
\otimes_{\caB_l}
\Vrl\sigma{\lm {1\sigma}(s)}{\lm {2\sigma}(s)}h
\rmk\\
&=\lim_{t,s\to\infty}
\lmk S_\sigma^{(l)\unit}\lmk \Theta(h^{-1})\lmk \Vrl\rho{\lm {1\rho}(t)}{\lm {2\rho}(t)}g\rmk\rmk^*\rmk
\lmk \Vrl\sigma{\lm {1\sigma}(s)}{\lm {2\sigma}(s)}h\rmk^*
\lmk
\Vrl\rho{\lm {1\rho}(t)}{\lm {2\rho}(t)}g
\otimes_{\caB_l}
\Vrl\sigma{\lm {1\sigma}(s)}{\lm {2\sigma}(s)}h
\rmk\\
&=\lim_{t,s\to\infty}
\lmk S_\sigma^{(l)\unit}\lmk \Theta(h^{-1})\lmk \Vrl\rho{\lm {1\rho}(t)}{\lm {2\rho}(t)}g\rmk\rmk^*\rmk
\lmk \Vrl\sigma{\lm {1\sigma}(s)}{\lm {2\sigma}(s)}h\rmk^*
\Vrl\rho{\lm {1\rho}(t)}{\lm {2\rho}(t)}g
S_\rho^{(l)\unit }\lmk \Vrl\sigma{\lm {1\sigma}(s)}{\lm {2\sigma}(s)}h\rmk\\
&=\lim_{t,s\to\infty}
\lmk S_\sigma^{(l)\unit}\lmk \Theta(h^{-1})\lmk \Vrl\rho{\lm {1\rho}(t)}{\lm {2\rho}(t)}g\rmk\rmk^*\rmk
\lmk \Vrl\sigma{\lm {1\sigma}(s)}{\lm {2\sigma}(s)}h\rmk^*
\Srl\rho{\lm {1\rho}(t)}{\lm {2\rho}(t)}g
\lmk \Vrl\sigma{\lm {1\sigma}(s)}{\lm {2\sigma}(s)}h\rmk
\Vrl\rho{\lm {1\rho}(t)}{\lm {2\rho}(t)}g.
\end{split}
\end{align}
Because
\begin{align}
\begin{split}
\Vrl\sigma{\lm {1\sigma}(s)}{\lm {2\sigma}(s)}h
\in \pi\lmk\caA_{{\lm {1\sigma}(s)}\cap \loz}\rmk'\cap\caF
\end{split}
\end{align}
we have
\begin{align}
\begin{split}
\lim_{t\to\infty}\sup_{s}\lV \lmk \Srl\rho{\lm {1\rho}(t)}{\lm {2\rho}(t)}g-\id\rmk
\lmk \Vrl\sigma{\lm {1\sigma}(s)}{\lm {2\sigma}(s)}h\rmk\rV=0,
\end{split}
\end{align}
by Lemma \ref{ou}.
Substituting this to (\ref{kuma}), we obtain (ii).
\end{proof}
\begin{lem}\label{intw}Consider setting in the subsection \ref{sec:qss}.
Assume Assumption \ref{a1l} and Assumption \ref{a1r}.
Let $g,h\in G$ and $\rho\in O^{(g)}_{(\loz,\ltz)}$, $\sigma\in O^{(h)}_{(\loz,\ltz)}$.
Then we have 
\begin{align}
\begin{split}
\epsilon_G(\rho,\sigma)\in \Mor_G\lmk \rho\otimes \sigma,
\sigma\otimes\Theta(h^{-1})(\rho)
\rmk.
\end{split}
\end{align}
\end{lem}
\begin{proof}
Let $A\in \caA_{\rm loc}$.
With $\{(\loz,\ltz)\}\leftarrow \{ (\lm {1\sigma},\lm {2\sigma})\}$,
from Lemma \ref{tako} (i),
we have 
\begin{align}
\begin{split}
&\epsilon_G(\rho,\sigma) \lmk \rho\otimes \sigma\rmk(A)=
\lim_{s\to\infty} 
\lmk
\Vrl\sigma{\lm {1\sigma}(s)}{\lm {2\sigma}(s)}h
\rmk^*
S_\rho^{(l) \unit}
\lmk
\Vrl\sigma{\lm {1\sigma}(s)}{\lm {2\sigma}(s)}h
\rmk S_\rho^{(l)\unit}S_\sigma^{(l)\unit}\pi(A)\\
&=\lim_{s\to\infty} 
\lmk
\Vrl\sigma{\lm {1\sigma}(s)}{\lm {2\sigma}(s)}h
\rmk^*
S_\rho^{(l) \unit}
\lmk
\Srl\sigma{\lm {1\sigma}(s)}{\lm {2\sigma}(s)}h\pi(A)\cdot \Vrl\sigma{\lm {1\sigma}(s)}{\lm {2\sigma}(s)}h
\rmk \\
\end{split}
\end{align}
for $\lm {i\sigma}(s):=\lm {i\sigma}+s\bm e_0$, $i=1,2$, $s\ge 0$ and
any 
$\Vrl\sigma{\lm {1\sigma}(s)}{\lm {2\sigma}(s)}h\in \Vbu\sigma{\lm {1\sigma}(s)}{\lm {2\sigma}(s)}h$.
Because $A\in \caA_{{\lm {1\sigma}(s)}}$ eventually as $s\to\infty$,
we have $\Srl\sigma{\lm {1\sigma}(s)}{\lm {2\sigma}(s)}h\pi(A)=\pi\beta_h(A)$ eventually,
by Lemma \ref{ookami}.
Hence we have
\begin{align}\label{ika}
\begin{split}
&\epsilon_G(\rho,\sigma) \lmk \rho\otimes \sigma\rmk(A)=
\lim_{s\to\infty} 
\lmk
\Vrl\sigma{\lm {1\sigma}(s)}{\lm {2\sigma}(s)}h
\rmk^*
\lmk
\rho\beta_h(A)\rmk S_\rho^{(l) \unit}
\lmk\Vrl\sigma{\lm {1\sigma}(s)}{\lm {2\sigma}(s)}h
\rmk .
\end{split}
\end{align}
Note from Lemma \ref{tokage} that
\begin{align}\label{uma}
\begin{split}
\Theta(h)\lmk\Vrl\sigma{\lm {1\sigma}(s)}{\lm {2\sigma}(s)}h\rmk
\in
\Theta(h)\lmk 
\Vbk{\sigma}{\lmk \lm {1\sigma}(s), {\lm {2\sigma}(s)}\rmk }
^{(h)}\rmk
=\Vbk{\Theta(h)\lmk \sigma\rmk}{\lmk \lm {1\sigma}(s), {\lm {2\sigma}(s)}\rmk}^{(h)}
\end{split}
\end{align}
For our local $A$, there exists $u_0\ge 0$ such that 
$A\in \caA_{\lmk \ltz+u_0\bm e_0\rmk^c}$.
Then because $\ltz+u_0\bm e_0\subset \ltz$,
we have
\begin{align}
\begin{split}
\rho\beta_h(A)\in
\rho\lmk \caA_{\lmk \ltz+u_0\bm e_0\rmk^c} \rmk
\subset \rho\lmk \caA_{\lmk \ltz+u_0\bm e_0\rmk} \rmk'
=\pi\lmk \caA_{\lmk \ltz+u_0\bm e_0\rmk}\rmk'.
\end{split}
\end{align}
Therefore, from Lemma \ref{lem35} and (\ref{uma}), we have
\begin{align}
\lim_{s\to\infty}
\lV
\lmk S_{\Theta(h)(\sigma)}^{(l)\Theta(h)\lmk\Vrl\sigma{\lm {1\sigma}(s)}{\lm {2\sigma}(s)}h\rmk}
-\Ad R_h
\rmk
\lmk \rho\beta_h(A)\rmk
\rV=0.
\end{align}
Substituting this to (\ref{ika}), and using Lemma \ref{tokage}, we obtain
\begin{align}
\begin{split}
&\epsilon_G(\rho,\sigma) \lmk \rho\otimes \sigma\rmk(A)\\
&=
\lim_{s\to\infty} 
\lmk
\Vrl\sigma{\lm {1\sigma}(s)}{\lm {2\sigma}(s)}h
\rmk^*
\Ad R_h^*
S_{\Theta(h)(\sigma)}^{(l)\Theta(h)\lmk\Vrl\sigma{\lm {1\sigma}(s)}{\lm {2\sigma}(s)}h\rmk}
\lmk
\rho\beta_h(A)\rmk 
S_\rho^{(l) \unit}
\lmk\Vrl\sigma{\lm {1\sigma}(s)}{\lm {2\sigma}(s)}h
\rmk \\
&=
\lim_{s\to\infty} 
\lmk
\Vrl\sigma{\lm {1\sigma}(s)}{\lm {2\sigma}(s)}h
\rmk^*
S_{\sigma}^{(l)\Vrl\sigma{\lm {1\sigma}(s)}{\lm {2\sigma}(s)}h}\Ad R_h^*
\lmk
\rho\beta_h(A)\rmk 
S_\rho^{(l) \unit}
\lmk\Vrl\sigma{\lm {1\sigma}(s)}{\lm {2\sigma}(s)}h
\rmk \\
&=
\lim_{s\to\infty} 
S_{\sigma}^{(l)\unit}\Ad R_h^*
\lmk\rho\beta_h(A)\rmk \cdot 
\lmk
\Vrl\sigma{\lm {1\sigma}(s)}{\lm {2\sigma}(s)}h
\rmk^*
S_\rho^{(l) \unit}
\lmk\Vrl\sigma{\lm {1\sigma}(s)}{\lm {2\sigma}(s)}h
\rmk \\
&=S_{\sigma}^{(l)\unit}\Theta(h^{-1})\lmk\rho\rmk(A)
\epsilon_G\lmk \rho,\sigma\rmk
=\lmk \sigma\otimes\Theta(h^{-1})\lmk\rho\rmk\rmk(A)
\epsilon_G\lmk \rho,\sigma\rmk
\end{split}
\end{align}
\end{proof}

\begin{lem}\label{hex}Consider setting in the subsection \ref{sec:qss}.
Assume Assumption \ref{a1l} and Assumption \ref{a1r}.
For any
$g,h,k\in G$ and $\rho\in O^{(g)}_{(\loz,\ltz)}$, $\sigma\in O^{(h)}_{(\loz,\ltz)}$,
$\gamma\in O^{(k)}_{(\loz,\ltz)}$, we have
\begin{align}
\begin{split}
\epsilon_G\lmk\rho\otimes \sigma,\gamma\rmk
=\lmk
\epsilon_G(\rho,\gamma)\otimes  \id_{\Theta(k^{-1})(\sigma)}
\rmk
\lmk
\id_\rho\otimes\epsilon_G( \sigma,\gamma)
\rmk.
\end{split}
\end{align}
\end{lem}
\begin{proof}
Choose $(\lm {1\gamma},\lm {2\gamma})\in PC$ such that
 $\{(\loz,\ltz)\}\leftarrow \{ (\lm {1\gamma},\lm {2\gamma})\}$.
By Lemma \ref{tako},
\begin{align}
\begin{split}
&\lmk
\epsilon_G(\rho,\gamma)\otimes  \id_{\Theta(k^{-1})(\sigma)}
\rmk
\lmk
\id_\rho\otimes\epsilon_G( \sigma,\gamma)
\rmk\\
&=
\epsilon_G(\rho,\gamma) S_\rho^{(l)\unit}\lmk \epsilon_G(\sigma,\gamma)\rmk\\
&=\lim_{s\to\infty} 
\lmk
\Vrl\gamma{\lm {1\gamma}(s)}{\lm {2\gamma}(s)}k
\rmk^*
S_\rho^{(l) \unit}
\lmk
\Vrl\gamma{\lm {1\gamma}(s)}{\lm {2\gamma}(s)}k
\rmk
S_\rho^{(l)\unit}\lmk
\lmk
\Vrl\gamma{\lm {1\gamma}(s)}{\lm {2\gamma}(s)}k
\rmk^*
S_\sigma^{(l) \unit}
\lmk
\Vrl\gamma{\lm {1\gamma}(s)}{\lm {2\gamma}(s)}k
\rmk
\rmk\\
&=\lim_{s\to\infty} 
\lmk
\Vrl\gamma{\lm {1\gamma}(s)}{\lm {2\gamma}(s)}k
\rmk^*
S_\rho^{(l)\unit}\lmk
S_\sigma^{(l) \unit}
\lmk
\Vrl\gamma{\lm {1\gamma}(s)}{\lm {2\gamma}(s)}k
\rmk
\rmk\\
&=\epsilon_G(\rho\otimes \sigma,\gamma).
\end{split}
\end{align}
\end{proof}
\begin{lem}Consider setting in the subsection \ref{sec:qss}.
Assume Assumption \ref{a1l} and Assumption \ref{a1r}.
For any
$g,h,k\in G$ and $\rho\in O^{(g)}_{(\loz,\ltz)}$, $\sigma\in O^{(h)}_{(\loz,\ltz)}$,
$\gamma\in O^{(k)}_{(\loz,\ltz)}$, we have
\begin{align}
\begin{split}
\epsilon_G\lmk\rho, \sigma\otimes \gamma\rmk
=\lmk
\id_{\sigma} \otimes \epsilon_G\lmk \Theta(h^{-1})(\rho),\gamma\rmk
\rmk
\lmk
\epsilon_G( \rho, \sigma) \otimes \id_\gamma
\rmk.
\end{split}
\end{align}
\end{lem}
\begin{proof}
Choose $(\lm {1\rho},\lm {2\rho})\in PC$ such that
$\{(\lm {1\rho},\lm {2\rho})\}\leftarrow \{  (\loz,\ltz)\}$.
Let $\Vrl\rho{\lm {1\rho}(t)}{\lm {2\rho}(t)}g\in \Vbu\rho{\lm {1\rho}(t)}{\lm {2\rho}(t)}g$.
By Lemma \ref{tokage},
we have
\begin{align}
\begin{split}
\Vrl{\Theta(h^{-1})(\rho)}{\lm {1{\rho}}(t)}{\lm {2{\rho}}(t)}{h^{-1}gh}
:=\Theta(h^{-1})\lmk
\Vrl\rho{\lm {1\rho}(t)}{\lm {2\rho}(t)}g\rmk
\in 
\Vbk
{\Theta(h^{-1})(\rho)}{\lmk \lm {1{\rho}}(t)\lm {2{\rho}}(t)\rmk}^{h^{-1}gh}.
\end{split}
\end{align}
By Lemma \ref{tako}, we have
\begin{align}
\begin{split}
&\lmk
\id_{\sigma} \otimes \epsilon_G\lmk \Theta(h^{-1})(\rho),\gamma\rmk
\rmk
\lmk
\epsilon_G( \rho, \sigma) \otimes \id_\gamma
\rmk\\
&=S_\sigma^{(l)\unit}\lmk \epsilon_G\lmk \Theta(h^{-1})(\rho),\gamma\rmk\rmk
\epsilon_G( \rho, \sigma)\\
&=\lim_{t\to\infty} 
S_\sigma^{(l)\unit}\lmk
S_\gamma^{(l) \unit}\lmk\Theta(k^{-1})\lmk
\lmk
\Vrl{\Theta(h^{-1})(\rho)}{\lm {1{\rho}}(t)}{\lm {2{\rho}}(t)}{h^{-1}gh}
\rmk^*\rmk
\rmk
\Vrl{\Theta(h^{-1})(\rho)}{\lm {1{\rho}}(t)}{\lm {2{\rho}}(t)}{h^{-1}gh}
\rmk\\
&\quad \quad S_\sigma^{(l) \unit}\lmk\Theta(h^{-1})\lmk
\lmk
\Vrl\rho{\lm {1\rho}(t)}{\lm {2\rho}(t)}g
\rmk^*\rmk
\rmk
\Vrl\rho{\lm {1\rho}(t)}{\lm {2\rho}(t)}g\\
&=\lim_{t\to\infty} 
S_\sigma^{(l)\unit}\lmk
S_\gamma^{(l) \unit}\lmk\Theta(k^{-1})\lmk
\lmk
\Theta(h^{-1})\lmk
\Vrl\rho{\lm {1\rho}(t)}{\lm {2\rho}(t)}g\rmk
\rmk^*
\rmk
\rmk
\Theta(h^{-1})\lmk
\Vrl\rho{\lm {1\rho}(t)}{\lm {2\rho}(t)}g\rmk
\rmk\\
&\quad \quad S_\sigma^{(l) \unit}\lmk\Theta(h^{-1})\lmk
\lmk
\Vrl\rho{\lm {1\rho}(t)}{\lm {2\rho}(t)}g
\rmk^*\rmk
\rmk
\Vrl\rho{\lm {1\rho}(t)}{\lm {2\rho}(t)}g\\
&=\lim_{t\to\infty} 
S_\sigma^{(l)\unit}\lmk
S_\gamma^{(l) \unit}\lmk\Theta(k^{-1})\lmk
\lmk
\Theta(h^{-1})\lmk
\Vrl\rho{\lm {1\rho}(t)}{\lm {2\rho}(t)}g\rmk
\rmk^*
\rmk
\rmk
\rmk\Vrl\rho{\lm {1\rho}(t)}{\lm {2\rho}(t)}g\\
&=\lim_{t\to\infty} 
S_\sigma^{(l)\unit}S_\gamma^{(l) \unit}
\lmk
\Theta((hk)^{-1})\lmk
\Vrl\rho{\lm {1\rho}(t)}{\lm {2\rho}(t)}g
\rmk^*
\rmk
\Vrl\rho{\lm {1\rho}(t)}{\lm {2\rho}(t)}g\\
&=\epsilon_G(\rho,\sigma\otimes \gamma).
\end{split}
\end{align}
\end{proof}

\begin{lem}\label{egn}Consider setting in the subsection \ref{sec:qss}.
Assume Assumption \ref{a1l} and Assumption \ref{a1r}.
For any
$g,h\in G$ and $\rho,\rho'\in O^{(g)}_{(\loz,\ltz)}$, $\sigma,\sigma'\in O^{(h)}_{(\loz,\ltz)}$,
and
$X\in \Mor_G(\rho,\rho')$, $Y\in \Mor_G(\sigma,\sigma')$, we have
\begin{align}
\begin{split}
\epsilon_G(\rho',\sigma')\lmk X\otimes Y\rmk=
\lmk Y\otimes\Theta(h^{-1})(X)\rmk\epsilon_G(\rho,\sigma)
\end{split}
\end{align}
\end{lem}
\begin{proof}
Let $g,h\in G$ and $\rho,\rho'\in O^{(g)}_{(\loz,\ltz)}$, $\sigma,\sigma'\in O^{(h)}_{(\loz,\ltz)}$.
Let $(\lm {1\rho},\lm {2\rho}), 
(\lm {1\sigma},\lm {2\sigma}), \in PC$ such that
$\{(\lm {1\rho},\lm {2\rho}) \}
\leftarrow 
\{
(\lm {1\sigma},\lm {2\sigma})
\}$.
We set
$\lm{ i\rho}(t):=\lm {i\rho}-t\bm e_0$,
$\lm {i\sigma}(s):=\lm {i\sigma}+s\bm e_0$,
with $i=1,2$, $t,s\ge 0$.
Let
\begin{align}
\begin{split}
\Vrl\rho{\lm {1\rho}(t)}{\lm {2\rho}(t)}g\in \Vbu\rho{\lm {1\rho}(t)}{\lm {2\rho}(t)}g,\quad
\Vrl{\rho'}{\lm {1\rho}(t)}{\lm {2\rho}(t)}g\in \Vbu{\rho'}{\lm {1\rho}(t)}{\lm {2\rho}(t)}g
\\
\Vrl\sigma{\lm {1\sigma}(s)}{\lm {2\sigma}(s)}h\in \Vbu\sigma{\lm {1\sigma}(s)}{\lm {2\sigma}(s)}h,\quad 
\Vrl{\sigma'}{\lm {1\sigma}(s)}{\lm {2\sigma}(s)}h\in \Vbu{\sigma'}{\lm {1\sigma}(s)}{\lm {2\sigma}(s)}h
\end{split}
\end{align}
for $t,s\ge 0$.
Then we have
\begin{align}\label{samurai}
\begin{split}
&\lV
\epsilon_G(\rho',\sigma')\lmk X\otimes Y\rmk-
\lmk Y\otimes\Theta(h^{-1})(X)\rmk\epsilon_G(\rho,\sigma)
\rV\\
&=\lim_{t,s\to\infty}
\lV
\begin{gathered}
\lmk \Vrl{\sigma'}{\lm {1\sigma}(s)}{\lm {2\sigma}(s)}h\otimes_{\caB_l}
\Theta(h^{-1})\lmk \Vrl{\rho'}{\lm {1{\rho}}(t)}{\lm {2{\rho}}(t)}g\rmk\rmk^*
\lmk
\Vrl{\rho'}{\lm {1{\rho}}(t)}{\lm {2{\rho}}(t)}g
\otimes_{\caB_l}
\Vrl{\sigma'}{\lm {1{\sigma}}(s)}{\lm {2{\sigma}}(s)}h
\rmk\\
\lmk X\otimes_{\caB_l} Y\rmk\\
-\lmk Y\otimes_{\caB_l}\Theta(h^{-1})(X)\rmk
\lmk \Vrl{\sigma}{\lm {1\sigma}(s)}{\lm {2\sigma}(s)}h\otimes_{\caB_l}
\Theta(h^{-1})\lmk \Vrl{\rho}{\lm {1{\rho}}(t)}{\lm {2{\rho}}(t)}g\rmk\rmk^*
\lmk
\Vrl{\rho}{\lm {1{\rho}}(t)}{\lm {2{\rho}}(t)}g
\otimes_{\caB_l}
\Vrl{\sigma}{\lm {1{\sigma}}(s)}{\lm {2{\sigma}}(s)}h
\rmk
\end{gathered}
\rV\\
&=\lim_{t,s\to\infty}
\lV
\begin{gathered}
\lmk
\Vrl{\rho'}{\lm {1{\rho}}(t)}{\lm {2{\rho}}(t)}gX
\lmk \Vrl{\rho}{\lm {1{\rho}}(t)}{\lm {2{\rho}}(t)}g\rmk^*
\otimes_{\caB_l}
\Vrl{\sigma'}{\lm {1{\sigma}}(s)}{\lm {2{\sigma}}(s)}hY\lmk \Vrl{\sigma}{\lm {1{\sigma}}(s)}{\lm {2{\sigma}}(s)}h\rmk^*
\rmk\\
-\lmk \Vrl{\sigma'}{\lm {1\sigma}(s)}{\lm {2\sigma}(s)}h Y
 \lmk \Vrl{\sigma}{\lm {1\sigma}(s)}{\lm {2\sigma}(s)}h\rmk^*
\otimes_{\caB_l}\Theta(h^{-1})\lmk \Vrl{\rho'}{\lm {1{\rho}}(t)}{\lm {2{\rho}}(t)}g\rmk
\Theta(h^{-1})(X)\lmk \Theta(h^{-1})\lmk \Vrl{\rho}{\lm {1{\rho}}(t)}{\lm {2{\rho}}(t)}g\rmk^*\rmk
\rmk\\
\end{gathered}
\rV\\
&=\lim_{t,s\to\infty}
\lV
\begin{gathered}
\lmk
\Vrl{\rho'}{\lm {1{\rho}}(t)}{\lm {2{\rho}}(t)}gX
\lmk \Vrl{\rho}{\lm {1{\rho}}(t)}{\lm {2{\rho}}(t)}g\rmk^*
\otimes_{\caB_l}
\Vrl{\sigma'}{\lm {1{\sigma}}(s)}{\lm {2{\sigma}}(s)}hY\lmk \Vrl{\sigma}{\lm {1{\sigma}}(s)}{\lm {2{\sigma}}(s)}h\rmk^*
\rmk\\
-\lmk \Vrl{\sigma'}{\lm {1\sigma}(s)}{\lm {2\sigma}(s)}h Y
 \lmk \Vrl{\sigma}{\lm {1\sigma}(s)}{\lm {2\sigma}(s)}h\rmk^*
\otimes_{\caB_l}\Theta(h^{-1})\lmk \Vrl{\rho'}{\lm {1{\rho}}(t)}{\lm {2{\rho}}(t)}g
X\lmk \Vrl{\rho}{\lm {1{\rho}}(t)}{\lm {2{\rho}}(t)}g\rmk^*\rmk
\rmk\\
\end{gathered}
\rV\\
\end{split}
\end{align}
Because
\begin{align}
\begin{split}
\Vrl{\rho'}{\lm {1{\rho}}(t)}{\lm {2{\rho}}(t)}gX
\lmk \Vrl{\rho}{\lm {1{\rho}}(t)}{\lm {2{\rho}}(t)}g\rmk^*
\in \Mor_{\caB_l}\lmk \Srl\rho{\lm {1\rho}(t)}{\lm {2\rho}(t)}g,
\Srl{\rho'}{{\lm {1\rho}}(t)}{{\lm {2\rho}}(t)} g
\rmk,\\
\Vrl{\sigma'}{\lm {1{\sigma}}(s)}{\lm {2{\sigma}}(s)}hY\lmk \Vrl{\sigma}{\lm {1{\sigma}}(s)}{\lm {2{\sigma}}(s)}h\rmk^*
\in
\Mor_{\caB_l}\lmk
\Srl\sigma{\lm {1\sigma}(s)}{\lm {2\sigma}(s)}h,
\Srl{\sigma'}{\lm {1\sigma}(s)}{\lm {2\sigma}(s)}h
\rmk
\end{split}
\end{align}
the right hand side of (\ref{samurai}) goes to $0$ as $t,s\to\infty$, from Lemma \ref{kappa}.
\end{proof}
\begin{lem}\label{tentomushi}Consider setting in the subsection \ref{sec:qss}.
Assume Assumption \ref{a1l} and Assumption \ref{a1r}.
For $\rho\in O^{(g)}_{(\loz,\ltz)}$ and $\sigma\in O^{(h)}_{(\loz,\ltz)}$,
we have
\begin{align}
\begin{split}
\Theta(k)\lmk\epsilon_G(\rho,\sigma)\rmk
=\epsilon_G(\Theta(k)(\rho),\Theta(k)(\sigma))
\end{split}
\end{align}
\end{lem}
\begin{proof}
If $\{(\loz,\ltz)\}\leftarrow \{ (\lm {1\sigma},\lm {2\sigma})\}$,
then
\begin{align}
\begin{split}
\epsilon_G(\rho,\sigma):=
\lim_{s\to\infty} 
\lmk
\Vrl\sigma{\lm {1\sigma}(s)}{\lm {2\sigma}(s)}h
\rmk^*
S_\rho^{(l) \unit}
\lmk
\Vrl\sigma{\lm {1\sigma}(s)}{\lm {2\sigma}(s)}h
\rmk
\end{split}
\end{align}
for $\lm {i\sigma}(s):=\lm {i\sigma}+s\bm e_0$, $i=1,2$, $s\ge 0$ and
any 
$\Vrl\sigma{\lm {1\sigma}(s)}{\lm {2\sigma}(s)}h\in \Vbu\sigma{\lm {1\sigma}(s)}{\lm {2\sigma}(s)}h$
by Lemma \ref{tako}.
From Lemma \ref{tokage}, we also have
$\Theta(k)\lmk \Vrl\sigma{\lm {1\sigma}(s)}{\lm {2\sigma}(s)}h\rmk
\in \Vbk{\Theta(k)\lmk \sigma\rmk}{({\lm {1\sigma}(s)}{\lm {2\sigma}(s)})}^{(khk^{-1})}$,
and $S_{\Theta(k)(\rho)}^{(l),\unit}
=\Ad R_k\circ S_{\rho}^{(l), \unit}\circ \Ad R_k^*$
Therefore, we have
\begin{align}
\begin{split}
&\Theta(k)\lmk \epsilon_G(\rho,\sigma)\rmk:=
\lim_{s\to\infty} 
\Theta(k)\lmk
\Vrl\sigma{\lm {1\sigma}(s)}{\lm {2\sigma}(s)}h
\rmk^*
S_{\Theta(k)(\rho)}^{(l),\unit}
\lmk
\Theta(k)\lmk
\Vrl\sigma{\lm {1\sigma}(s)}{\lm {2\sigma}(s)}h
\rmk
\rmk\\
&=\epsilon_G(\Theta(k)(\rho),\Theta(k)(\sigma)).
\end{split}
\end{align}

\end{proof}

\end{document}